\documentclass[a4paper,10pt]{article}

\usepackage{amsfonts,amssymb,graphics,epsfig,verbatim,bm,latexsym,amsmath,amsbsy,url}
\usepackage{amsthm}
\usepackage{graphicx}
\usepackage{placeins}
\usepackage{booktabs}
\usepackage{enumerate}
\usepackage{lscape,rotating}
\usepackage{bbm}
 \usepackage{pdfpages}
\usepackage{threeparttable}
\usepackage{todonotes}
\usepackage[capposition=top]{floatrow}
\usepackage{subfigure}
\usepackage{dsfont}
\usepackage{xr}
\usepackage{commands}
\usepackage{multirow}
\usepackage{natbib}
\usepackage{mathtools}
\usepackage{todonotes}
\usepackage{algorithm}
\usepackage{algpseudocode}
\usepackage{pgfplotstable}
\usepackage{setspace} 
\pgfplotsset{compat=1.9}

\newtheorem{theorem}{Theorem}

\newtheorem{assumption}{Assumption}

\newtheorem{lemma}{Lemma}

\newtheorem{definition}{Definition}

\usepackage{chngcntr}
\counterwithout{equation}{subsection}
\counterwithout{equation}{section}

\usepackage{setspace} 
\title{Testing for a Threshold in Models with Endogenous Regressors}
\author{Mario P. Rothfelder\thanks{Amsterdam School of Economics, University of Amsterdam, and Tinbergen Institute. Address: University of Amsterdam, Roetersstraat 11, 1018WB Amsterdam, The Netherlands. Email: m.p.rothfelder@uva.nl} \, and Otilia Boldea\thanks{Corresponding Author. Department of Econometrics and Operations Research, Tilburg School of Economics and Management, Tilburg University. Address: Tilburg University, Warandelaan 2, 5037AB, Tilburg, The Netherlands. Email: o.boldea@tilburguniversity.edu}\thanks{We would like to thank for very valuable comments and suggestions on this work: Jeffrey Campbell, Mehmet Caner, Pavel \u{C}\'{i}\u{z}ek, Bruce Hansen, Frank Kleibergen, Andr\'{e} Lucas, Bertrand Melenberg, David Pacini, Valerie Ramey, Nikolaus Schweizer, Bas Werker, and the participants at the Netherlands Econometrics Study Group Conference in Tilburg (2014) and Maastricht (2015), the Bristol Econometric Study Group (2015), the IAAE Conference in Thessaloniki (2015), the NBER-NSF Time Series Conference in Vienna (2015), the ESEM Conference in Geneva (2016) and the NASM Conference in Philadelphia (2016). Otilia Boldea gratefully acknowledges the support of the NWO VENI Grant 451-11-001.}}

\begin{document}

\maketitle
\begin{abstract}
We show by simulation that the test for an unknown threshold in models with endogenous regressors - proposed in \cite{canerhansen2004} - can exhibit severe size distortions both in small and in moderately large samples, pertinent to empirical applications. We propose three new tests that rectify these size distortions. The first test is based on GMM estimators. The other two are based on unconventional 2SLS estimators, that use additional information about the linearity (or lack of linearity) of the first stage. Just like the test in \cite{canerhansen2004}, our tests are non-pivotal, and we prove their bootstrap validity. The empirical application revisits the question in \cite{rameyzubairy2018} whether government spending multipliers are larger in recessions, but using tests for an unknown threshold. Consistent with  \cite{rameyzubairy2018}, we do not find strong evidence that these multipliers are larger in recessions.\\

\end{abstract}
\newpage

\section{Introduction}
In the aftermath of the 2008 financial crisis, there has been a surge in the macroeconomic literature investigating whether the response of many key macroeconomic variables to monetary and fiscal policies depends on the state of the economy - see, among others, \cite{auerbachgorodnichenko2013}, \cite{owyangetal2013}, \cite{Caggiano:2015}, \cite{Cug:2015}, \cite{rameyzubairy2018}, \cite{Alloza:2022} and \cite{Jo:2022} for fiscal policy examples, and
\cite{Santoro:2014}, \cite{Barnichon:2018}, \cite{Jorda:2020}, \cite{Alpanda:2021},  \cite{Bruns:2021} and \cite{Klepacz} for monetary policy examples. These papers model state dependence in various ways, including via threshold models, in which case the state dependence is typically driven by a particular variable such as the unemployment rate, interest rates, or credit conditions.

Threshold models were also widely used in economics to model unemployment, growth, bank profits, asset prices, exchange rates,  and interest rates; see \cite{hansen2011} for a survey of  economic applications. While threshold models with exogenous regressors  have been widely studied and their asymptotic properties are well known\footnote{See \textit{inter alia} \cite{tong1990}, \cite{hansen1996,hansen1999,hansen2000} and \cite{gonwolf2005} for inference, \cite{gonpit2002} for multiple threshold regression and model selection, \cite{canerhansen2001} and \cite{gonpit2006} for threshold regression with unit roots, \cite{seolin2007} for smoothed estimators of threshold models, \cite{leeetal2011} for testing for thresholds, and \cite{hansen2016} for threshold regressions with a kink.}, the literature on threshold models with endogenous regressors remains relatively scarce.\footnote{For some contributions with endogenous regressors, see \textit{inter alia}: for time-series, \cite{canerhansen2004}, who consider exogenous threshold variables and \cite{kourtellosetal2013} who consider endogenous threshold variables; for cross-sections and (short) panels,  \cite{seoshin2016} (and references therein), \cite{yu2018} and \cite{mcadam2019}, who consider endogenous threshold variables.} Nevertheless, in many applications, the regressors are \textit{endogenous} and the existence of a threshold has important policy implications. For example, among the empirical papers cited above, \cite{owyangetal2013}, \cite{Cug:2015}, \cite{rameyzubairy2018} and \cite{Jo:2022} use a threshold model with endogenous regressors, where the state dependence of the macroeconomic response is driven by a  threshold variable being above or below a certain a-priori fixed value. \cite{rameyzubairy2018} (RZ henceforth) used a threshold model with endogenous regressors to investigate whether the government spending multiplier is larger  in recessions, where recessions were defined by the unemployment rate being below or above a threshold parameter. This has important policy implications, because if the government spending multiplier is larger (above one) in recessions, it implies that governments should spend more in recessions to boost the economy.

In their analysis, RZ fix this threshold parameter at an unemployment rate of $6.5\%$.\footnote{This is based on the Federal Reserve's use of this threshold in a policy announcement. They later do robustness checks with a larger threshold, and modelled time-varying thresholds.} As the threshold parameter is typically \textit{unknown}, we revisit their question and test for an unknown threshold, using - to our knowledge - the only parametric test available for linear \textit{time series} models with endogenous regressors that directly applies to the RZ model. This test was proposed in \cite{canerhansen2004} (CH henceforth). CH first compute a Wald test statistic for all candidate threshold values between the $\epsilon$ and $(1-\epsilon)$ quantiles of the threshold variable, then take the maximum over this sequence to obtain a test for the null hypothesis of no threshold against the alternative hypothesis of an unknown threshold in (otherwise) linear models with endogenous regressors and exogenous threshold variables. 

Our simulations show that this test has serious size distortions, with rejection frequencies up to three times the nominal size in small samples (see Tables \ref{tab:GMM.1} and \ref{tab:GMM.2}), accompanied by a reversal to severe under-rejections for larger samples of $1000$ observations. Tables \ref{tab:GMM.1} and \ref{tab:GMM.2} show that these size distortions are already present in just-identified models with strong instruments and homoskedastic data. 
We identify two problems with the CH test that lead to these size distortions, and proceed to correct them.
 
The first problem is illustrated in Figure \ref{fig:TestSeq}, where we test for an unknown threshold in the RZ model, and plot the sequence of the CH test statistics over the candidate threshold values, along with the same sequence for three tests we propose.\footnote{Section \ref{sec:TMTS} explains how these tests are calculated. Section \ref{sec:Appl} describes the threshold estimator, the model and the data. } The plot shows erratic behavior of the CH test sequence, switching frequently between low and high values, especially around the sample edges, but starting already at the 25\% and  the 75\% sample quantiles of the threshold variable. Therefore, the CH test, the maximum of the plotted sequence of tests, can change by a large amount when slightly changing the trimming.\footnote{Note that this is not due to the actual threshold estimate being between cut-off points: if there was a threshold, its consistent estimate, based on 2SLS or in CH,  with 25\% cut-off, is $8.33$; however, in our application in Section \ref{sec:Appl}, and in line with RZ, we do not find evidence of such a threshold.} This is problematic for its application in practice, as in general, it may lead to both over- or under-rejection of the null hypothesis, especially since this non-monotonic behavior is not well replicated by the bootstrap critical values even for samples of $1000$ observations, as shown in our simulations.
\FloatBarrier
\begin{figure}[ht!]
\begin{center}
\caption{Plot of the CH test sequence (black) versus our proposed test sequences (blue, red, green) for the model in RZ }\label{fig:TestSeq}
\includegraphics[width=0.9\textwidth]{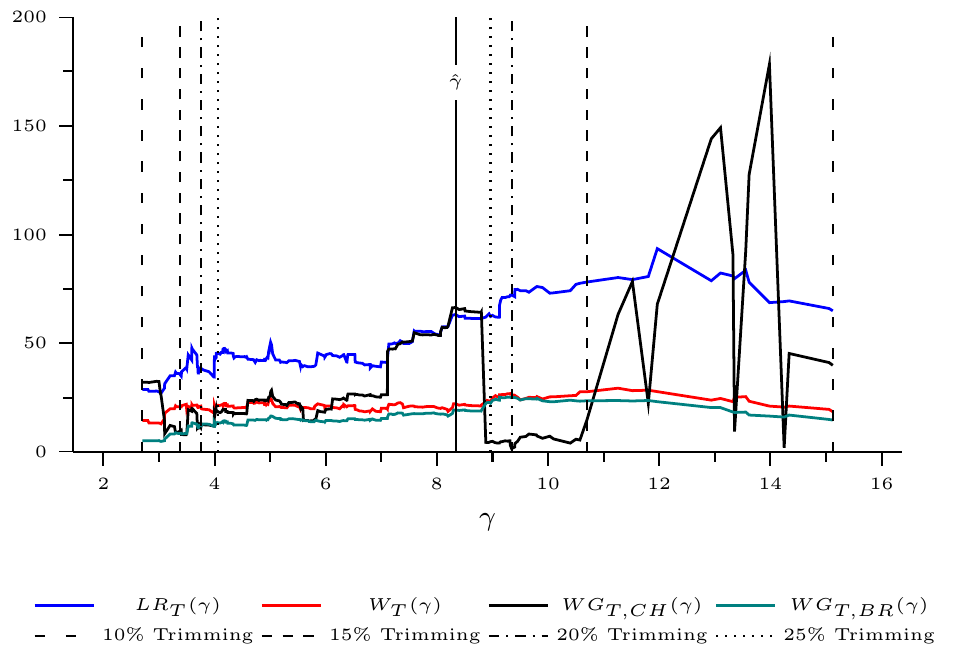}
\end{center}
Note: Above, $WG_{T,CH}(\gamma)$ refers to the CH test sequence, and $WG_{T,BR}(
\gamma)$, $W_T(\gamma)$ and $LR_T(\gamma)$ to our three test sequences: the first one is based on GMM estimators, and the last two on 2SLS estimators. The sample size is $500$.
\end{figure}
\FloatBarrier
We identify the source of this problem to lie in the computation of the variance estimator in the middle of each Wald test for a candidate threshold. The residuals in the variance estimator are obtained with sub-sample parameter estimators, using observations only below or above each candidate threshold value. When the threshold value is close to the sample edges, these residuals can be very inaccurate approximations of the true underlying errors, because of the slow convergence rate of the sub-sample estimators employed to obtain them. We correct this by obtaining the residuals with full-sample estimators instead. Figure \ref{fig:TestSeq} shows that all three test statistics we propose no longer display this non-monotonic behavior, whether computed with generalized method of moment estimators (GMM) estimators as in CH, or with two-stage least squares estimators (2SLS).

A second, yet related issue arises in the construction of the critical values of the CH test. The critical values of unknown threshold tests typically depend on the data, and therefore need to be simulated or bootstrapped. CH propose to bootstrap the critical values via a wild fixed regressor bootstrap and prove the bootstrap validity of their test. However, just like the variance estimator, the bootstrap residuals (and therefore the bootstrap samples) are computed with estimators under the alternative of each candidate threshold value.\footnote{Note that bootstrapping under the alternative is not necessary even when the variance estimator is computed with residuals under each alternative hypothesis of a candidate threshold value.}

 While bootstrapping under the alternative does not affect the asymptotic validity of the CH test, it is problematic for two reasons. First, it is computationally much more intensive than computing the bootstrap samples just once, under the null hypothesis, using full-sample estimators. This is because for each bootstrapped test, one needs to compute many bootstrap samples corresponding to each candidate threshold value. Second, just like their sample equivalents, the bootstrap residuals will be inaccurate at the sample edges due to slow convergence of the sub-sample estimators used to employ them. When taking the maximum over the sequence of bootstrapped tests, then doing so for many bootstrap samples, the bootstrapped critical values can become highly unreliable for the original test statistic, even for sample sizes up to $1000$ observations. Tables \ref{tab:GMM.1} and \ref{tab:GMM.2} in the simulation section show severe under-rejection of the null hypothesis  for sample sizes of $1000$ observations.  They also  show that bootstrapping under the null hypothesis fixes this issue, leading to correctly sized tests, but only if the variance correction discussed earlier is also employed.\footnote{Note that all test statistics for an unknown threshold we consider are non-pivotal, so one cannot expect any bootstrap to provide asymptotic refinements. While for the (trimmed) edges of the sample, the residuals computed with sub-sample estimators and their bootstrapped version are clearly inaccurately estimating the true underlying errors, because of slow convergence of the sub-sample estimator employed to construct them, this is not the case for the middle of the sample. Because both our tests and the CH test take the maximum over all candidate threshold values, around the (trimmed) sample edges or not, it is not possible to derive uniform asymptotic refinements of our tests over the CH test; these refinements will only hold for  candidate threshold values around the sample edges. We would like to thank a referee for raising this issue.}  

In this paper, we propose three test statistics for testing the null hypothesis of an unknown threshold in threshold models with endogenous regressors and exogenous regressors, and because both their computation and the bootstrap is different than CH, we derive for all three tests their asymptotic distribution and bootstrap validity. The first test we propose is similar to the CH test and uses sub-sample GMM estimators, but, unlike CH,  employs a different variance estimator and a null bootstrap. The other two tests are a likelihood ratio (LR) test and a Wald test, both based on 2SLS estimators. The 2SLS estimators are not conventional and therefore not a special case of the sub-sample GMM estimators used in the CH test, because they use additional information about the first stage being either linear or having itself a threshold, while the GMM estimators do not use this information by construction. Therefore, the resulting test statistics can be equally reliable to the test based on GMM estimators, as shown in our simulations. Because the 2SLS with a first stage threshold require consistent estimators of the first stage threshold parameter, as a by-product of our analysis, we also prove the consistency of ordinary least-squares threshold estimators with a fixed threshold, a result we could not find in the extant literature, only for very specific regression models.\footnote{ See Theorem \ref{theo:0} in the Online Supplement.}

\indent Our paper is closely related to several papers in the change-point literature. \cite{bch2019} study the same 2SLS-based test statistics as this paper but for change-points. They also prove bootstrap validity of their tests, however we employ different proof techniques in this paper because the threshold variable is typically correlated with regressors, while the change-points are not, and the asymptotic distributions will also be different. \cite{mm2014} use information about change-points in the first stage to improve the power of tests for moment conditions, while we use similar information to improve the size of our tests. \cite{antoineboldea2015b} and \cite{antoineboldea2018} also use a full sample first stage or change points in the first stage for more efficient estimation, while we focus on testing.\\
\indent It should be noted that we allow for endogenous regressors, but not for endogenous threshold variables. For the latter, see \textit{inter alia} \cite{kourtellosetal2013}, \cite{yu2018}, \cite{mcadam2019} and \cite{liao2019}. To account for regressor endogeneity, we use instruments for constructing parametric test statistics for thresholds. As a result, our tests have nontrivial local power for $O(T^{-1/2})$ threshold shifts, where $T$ is the sample size. This is in contrast to \cite{yu2018}, who do not use instruments, but rather local shifts around the threshold to construct a nonparametric threshold test.  As a result, their test covers more general functional forms, at the cost of losing power in $O(T^{-1/2})$ neighborhoods. Additionally, the later paper focuses on cross-sectional models, while our tests are applicable to both cross-sectional models and  time series models.\\
\indent In the empirical application, using the same data and model specification as in RZ, we revisit the question of whether the government spending multipliers are larger in recessions. As in RZ, we cannot rule out that the cumulative government spending multipliers are the same in recessions and expansions. However, we  estimate the threshold unemployment rate to be $8.3\%$, rather than $6.5\%$ as imposed in RZ. This new threshold  causes the military spending instrument constructed in RZ to become weaker for deep recessions, suggesting that this instrument is probably most informative at moderate unemployment rates somewhere between $6.5\%$ and $8.3\%$. \\
\indent The paper is organized as follows. Section 2 describes the model, the CH test and our test statistics,  the proposed bootstrap, as well as the assumptions and all the bootstrap validity results. Section 3 contains simulations and Section 4 contains the empirical application. Section 5 concludes. The Online Supplement, at the end of this document, contains all the proofs.
\section{Threshold Model, Test Statistics and Bootstrap Validity}\label{sec:TMBoot}
\subsection{Threshold Model and Test Statistics}\label{sec:TMTS}
Our framework is a linear model with a possible threshold at $\gamma^0$:
\begin{align} \nonumber
y_t & =\left(x_t^\top\thetaxo +z_{1t}^\top\thetazo\right)\indgamo+\left(x_t^\top\thetaxt+z_{1t}^\top\thetazt\right)\indinvgamo+\epsilon_t\\
& =w_t^\top\theta_1^0\, \indgamo+w_t^\top\theta_2^0\, \indinvgamo+\epsilon_t, \label{equ:TMTS.1}
\end{align}
where $y_t$ is the scalar dependent variable, $x_t$ is a $p_1\times1$ vector of endogenous variables, $z_{1t}$ a $p_2\times1$ vector of exogenous variables including the intercept and possibly lags of $y_t$, $q_t$ is the scalar exogenous threshold variable,  $\mathbf{1}[\cdot]$ is the indicator function, $w_t=(x_t^\top,z_{1t}^\top)^\top$ and $\theta_i^0=(\theta_{ix}^{0\top},\theta_{iz}^{0\top})^\top$. Let $\gamma^0\in \Gamma$, a strict subset of the support of $q_t$, and let $p=p_1+p_2$. The threshold variable is assumed exogenous and it can be a function of the exogenous regressors. As in CH, the first stage can be a linear model:      
\begin{equation}\label{equ:TMTS.2}
x_t= \Pi^{0\top}z_t+ u_t,
\end{equation}
or a threshold model:
\begin{equation}\label{equ:TMTS.3}
x_t=\Pi_1^{0\top}z_t \indrhoo +\Pi^{0\top}_2z_t \indinvrhoo+u_t,
\end{equation}
where $\rho^0 \in \Gamma$ is a threshold not necessarily coinciding with $\gamma^0$, and $z_t$ are $q \times 1$ strong and valid instruments, including $z_{1t}$, with $q-p_2 \geq p_1$. We assume that $\mathbb E[(\epsilon_t, u_t^\top)|\mathfrak F_{t}]=0$, where $\mathfrak F_t=\sigma\{z_{t-s},v_{t-s-1},q_{t-s}|s\geq0\}$, so that equation \eqref{equ:TMTS.1} can be estimated by either 2SLS or by GMM. \\
\indent We are interested in testing for an unknown threshold, i.e. the null hypothesis $\mathbb H_0: \theta_1^0=\theta_2^0 = \theta^0$. CH proposed a test based on GMM estimators of $\theta_i^0, (i=1,2)$ for each $\gamma \in \Gamma$. Because $z_t$ and $q_t$ are exogenous, the moment conditions 
\begin{equation}\label{equ:TMTS.4}
E[z_t \epsilon_t \indgam]=0, \qquad E[z_t \epsilon_t \indinvgam]=0
\end{equation} 
hold for all $\gamma \in \Gamma$. Based on these moment conditions, they construct the two-step GMM estimators:
\begin{equation}\label{equ:TMTS.5}
\hat \theta_{i \gamma,(2)}  =\Big(\hat N_{i\gamma} \hat H_{\epsilon,i\gamma}^{-1} \hat N_{i\gamma}^\top\Big)^{-1} \hat N_{i\gamma} \hat H_{\epsilon,i\gamma}^{-1} \Big( T^{-1} \sum_{i\gamma} z_t y_t \Big),
\end{equation} 
with $\sum_{1\gamma}(\cdot) = \sum_{t=1}^T (\cdot) \indgam$, $\sum_{2\gamma}(\cdot) = \sum_{t=1}^T (\cdot)\indinvgam$, $\hat N_{i\gamma} = T^{-1} \sum_{i\gamma} w_t z_t^\top$ and with $\hat H_{\epsilon,i\gamma} = T^{-1} \sum_{i\gamma} \hat \epsilon_{t,\gamma,(1)}^2 z_t z_t^\top$. Here, $\hat \epsilon_{t,\gamma,(1)} = y_t - w_t^\top \hat \theta_{1\gamma,(1)}\indgam - w_t^\top \hat \theta_{2\gamma,(1)}\indinvgam$ are the first step GMM residuals for each $\gamma$, and $\hat \theta_{i\gamma,(1)}$ are consistent first-step versions of $\hat \theta_{i \gamma,(2)}$, for example by replacing $\hat  H_{\epsilon,i\gamma}$ with $\hat M_{i\gamma} = T^{-1} \sum_{i\gamma} z_t z_t^\top$. These estimators can be used to construct a Wald test for each $\gamma$, and taking the supremum of this sequence of Wald tests over $\gamma \in \Gamma$ yields the test  in CH: 
\begin{equation}
WG_{T,CH}=\supl_{\gamma\in\Gamma}WG_{T,CH}(\gamma)=\supl_{\gamma\in\Gamma}T\left(\hat\theta_{1\gamma,(2)}-\hat\theta_{2\gamma,(2)}\right)^\top \hat V_{\gamma,(1)}^{-1}\left(\hat\theta_{1\gamma,(2)}-\hat\theta_{2\gamma,(2)}\right)\label{equ:TMTS.6},
\end{equation}
where $\hat V_{\gamma,(1)} = \sum_{i=1}^2 \Big(\hat N_{i\gamma} \hat H_{\epsilon,i\gamma}^{-1} \hat N_{i\gamma}^\top\Big)^{-1}$.\\  
\indent As shown in CH, the asymptotic distribution of the test statistic \eqref{equ:TMTS.6} is non-pivotal and therefore needs to be simulated/bootstrapped for a given application. CH propose to generate new pseudo-dependent variables $y^b_{t,\gamma}=\hat\epsilon_{t,\gamma,(2)}\eta_t$, where $\eta_t\stackrel{iid}\sim\mathcal N(0,1)$ and $\hat\epsilon_{t,\gamma,(2)}$ denote the second step GMM residuals for each value of $\gamma$, and recalculate the test statistic \eqref{equ:TMTS.6} for each $\gamma$ using $y^b_{t,\gamma}$ instead of $y_t$, and then for many bootstrap samples.\footnote{Note that the pseudo-dependent variables are generated without adding back the estimated mean to the bootstrap residuals. This is inconsequential to the analysis because the test statistic is based on mean differences across regimes of low or high $q_t$, and these are zero under the null of no threshold.} Even though CH prove validity of their bootstrap procedure in large samples, Tables \ref{tab:GMM.1} and \ref{tab:GMM.2} show that their bootstrap does not replicate well the empirical distribution of the test statistic in finite samples, being severely over-/undersized for small/large samples.\\
\indent Tables \ref{tab:GMM.1} and \ref{tab:GMM.2} in the simulation section show that these size distortions are due to two interacting phenomena: the type of bootstrap employed, and the way the heteroskedasticity-robust variance estimators are computed. We therefore employ two corrections. First, we adjust the bootstrap such that the pseudo-dependent variable $y_t^b$ is constructed using full-sample residuals. That is, we replace $y_t^b=\hat\epsilon_{t,\gamma,(2)}\eta_t$ by $y_t^b=\hat\epsilon_{t,(2)}\eta_t$ where $\hat\epsilon_{t,(2)}=y_t-w_t^\top\hat\theta_{(2)}$ and $\hat\theta_{(2)}$ is the second step GMM estimate under $\mathbb H_0$. This gets rid of the undersizing of the CH test statistic documented in the simulation section: the residuals become more accurate around the sample edges as they are not constructed with sub-sample estimators. 
 However, the simulations now indicate that the test is oversized (see Tables \ref{tab:GMM.1} and \ref{tab:GMM.2} , column ``Mix"). 

Therefore, we employ a second correction, where the heteroskedasticity-robust variance estimators are also computed with full-sample parameter estimates. More exactly, rather than using $\hat\epsilon_{t,\gamma,(1)}$ and $\hat\epsilon_{t,\gamma,(2)}$ in the expression for $\hat H_{\epsilon,i\gamma}$, we use $\hat\epsilon_{t,(1)}=y_t-w_t^\top\hat\theta_{(1)}$ instead, where $\hat\theta_{(1)}$ is the first step full-sample GMM estimate (so we redefine $\hat H_{\epsilon,i\gamma}=\sum_{i\gamma}
\hat\epsilon_{t,(1)}^2z_tz_t^\top$). As Tables 1 and 2 show (column ``BR", ``bootstrap/rectification''), this yields correctly sized sample test statistics in all samples considered. 

Note that both effects that we correct for are due to unstable estimates of the residuals at the sample edges below/above the 15\%/85\%-quantiles of the empirical distribution of $q_t$. The test employing these two corrections is denoted by $WG_{T,BR}=\supl_{\gamma\in\Gamma}WG_{T,BR}(\gamma)$.\\

\indent We also consider two 2SLS-based test statistics, because the GMM estimators involved in the computation of the tests above do not use information about the linearity or lack of linearity of the first stage. Therefore, they are not more efficient than the 2SLS estimators that use this information (see \cite{antoineboldea2015b} for a formal proof of this statement for change-point models), so there is no reason to expect that 2SLS-based tests will be inferior to the GMM-based tests.   

 The likelihood-ratio type and a Wald-type test statistic for $\mathbb H_0:\,\theta_1^0=\theta_2^0$ based on 2SLS estimators are:

\begin{subequations}
\begin{align}
LR_T & =  \sup_{\gamma\in\Gamma}\frac{SSR_0-SSR_1(\gamma)}{SSR_1(\gamma)/(T-2p)}\label{equ:TMTS.7a}\\
W_T & =  \sup_{\gamma\in\Gamma}T\Big(\hat\theta_{1\gamma}-\hat\theta_{2\gamma}\Big)^\top\hat V_{\gamma}^{-1}\Big(\hat\theta_{1\gamma}-\hat\theta_{2\gamma}\Big)\label{equ:TMTS.7b}
\end{align}
\end{subequations}
where $SSR_0=\sum_{t=1}^T(y_t-\hat w_t^\top\hat\theta)^2$, with $\hat\theta=(\sum_{t=1}^T\hat w_t\hat w_t^\top)^{-1}(\sum_{t=1}^T\hat w_ty_t)$ the full-sample 2SLS estimator, $SSR_1(\gamma)=\sum_{i=1}^2\sum_{i\gamma}(y_t-\hat w_t^\top\hat\theta_{i\gamma})^2$, with $\hat\theta_{i\gamma}=(\sum_{i\gamma}\hat w_t\hat w_t^\top)^{-1}(\sum_{i\gamma}\hat w_ty_t)$ the split-sample 2SLS estimators. Here, $\hat w_t=(\hat x_t^\top,z_{1t}^\top)^\top$ stacks the predicted endogenous variables $\hat x_t$  and the exogenous variables $z_{1t}$. The predicted endogenous variables are obtained either via estimating the linear first stage  equation \eqref{equ:TMTS.2}:
\begin{equation}
\hat\Pi = \Big(\sum_{t=1}^Tz_tz_t^\top\Big)^{-1}\Big(\sum_{t=1}^Tz_tx_t^\top\Big),\quad\hat x_t = \hat\Pi^\top z_t\label{equ:TMTS.8}
\end{equation}
or via estimating the threshold first-stage equation \eqref{equ:TMTS.3}:
\begin{subequations}
\begin{align}
\hat\Pi_{i\rho} & = \Big(\sum_{i\rho}z_tz_t^\top\Big)^{-1}\Big(\sum_{i\rho}z_tx_t^\top\Big)\,i=1,2, & \hat\Pi_{t\rho} & =\hat\Pi_{1\rho}\indrho+\hat\Pi_{2\rho}\indinvrho\label{equ:TMTS.9a}\\
\hat\rho & = \argmin_{\rho\in\Gamma}\left[\tr\sum_{t=1}^T(x_t-\hat\Pi^\top_{t\rho}z_t)(x_t-\hat\Pi^\top_{t\rho}z_t)^\top \right] & \hat\Pi_i & = \Big(\sum_{i\hat\rho}^Tz_tz_t^\top\Big)^{-1}\Big(\sum_{i\hat\rho}z_tx_t^\top\Big),\,i=1,2\label{equ:TMTS.9b}\\
\hat x_t& =\hat\Pi_1^\top z_t\mathbf{1}[q_t\leq\hat\rho]+\hat\Pi_2^\top z_t\mathbf{1}[q_t>\hat\rho] \label{equ:TMTS.9c}
\end{align}
\end{subequations}
Lastly, $\hat V_\gamma\inp V_\gamma=\lim\Var[T^{1/2}(\hat\theta_{1\gamma}-\hat\theta_{2\gamma})]$.\footnote{ The explicit expressions for $\hat V_\gamma$ and $V_\gamma$ are given in Online Supplement Section \ref{sec:ProofsLFSAsyms}, and Definition \ref{def:QuantsLFS} for a linear first stage, and in Online Supplement Section \ref{sec:ProofsTFSAsyms}, Definition \ref{def:QuantsTFS} for a threshold first stage, together with the expressions in the asymptotic distributions of the 2SLS test-statistics.} Unlike the sup Wald test in \cite{halletal2012}, which is the change-point counterpart of the test here, our test -- through the way $\hat V_\gamma$ is defined -- takes into account that the 2SLS estimators $\hat\theta_{1\gamma}$ and $\hat\theta_{2\gamma}$ are correlated through either a full-sample first-stage or through misalignment of $\rho^0$ and $\gamma$. Moreover, as in the case of CH's GMM-test, the 2SLS test-statistics are non-pivotal and, therefore, need to be simulated/bootstrapped. The next subsection describes the bootstrap we propose and contains results for asymptotic validity of this bootstrap for all three tests proposed.
\subsection{Bootstraps and their Validity}\label{sec:BootVal}
The bootstrap employed for both CH  GMM test and our GMM test is a wild bootstrap with fixed regressors because it does not bootstrap the regressors $w_t, x_t$ and the instruments $z_{1,t}$. We already alluded to the proposed change in the bootstrap procedure for the CH test in the previous section. These changes are summarized in the Algorithms \ref{alg:CHGMM} and \ref{alg:UsGMM} below. The difference between the CH test and our test are highlighted in lines 3--6 of the below algorithms. Since CH construct their pseudo-dependent variable for \textit{each} $\gamma$ separately, the for-loop over $\gamma$ starts already in line 3 of Algorithm \ref{alg:CHGMM}, as opposed to line 5 in Algorithm \ref{alg:UsGMM} when the \textit{same} pseudo-dependent variable is used for all values of $\gamma$. Line 6 in both algorithms indicates the difference in constructing the heteroskedasticity-robust variance estimators $\hat H_{\epsilon,i\gamma}^b$ and in $\hat H_{\epsilon,i\gamma}$.\\
\begin{minipage}[t]{0.48\textwidth}
\begin{algorithm}[H]
\centering
\caption{CH Bootstrap}\label{alg:CHGMM}
\footnotesize
\begin{algorithmic}[1]
\State \textbf{Input:} $\{\hat\epsilon_{t,\gamma,(2)}\}_{t,\gamma},\{w_t\}_t,\Gamma,B$
\For{$b\in\{1,..,B\}$}
\For{$\gamma\in\Gamma$}
\State $\eta_t\sim i.i.d.\, \cN(0,1)$
\State $y_{t,\gamma}^b\gets \hat\epsilon_{t,\gamma,(2)}\eta_t$
\State $WG_{T,CH}^b(\gamma)\gets$ use $\{(y_{t,\gamma}^b,w_t)\}_t$ to compute RHS in \eqref{equ:TMTS.6} using $\hat\epsilon_{t,\gamma,(1)}^b$ in $\hat H_{\epsilon,i\gamma}^b$
\EndFor
\State $WG_{T,CH}^b\gets\sup_{\gamma\in\Gamma}WG_{T,CH}^b(\gamma)$
\EndFor
\State \textbf{Return} $\{WG_{T,CH}^b\}_b$
\end{algorithmic}
\end{algorithm}
\end{minipage}
\hfil
\begin{minipage}[t]{0.48\textwidth}
\begin{algorithm}[H]
\centering
\caption{Modified CH Bootstrap}\label{alg:UsGMM}
\footnotesize
\begin{algorithmic}[1]
\State \textbf{Input:} $\{\hat\epsilon_{t,(2)}\}_{t},\{w_t\}_t,\Gamma,B$
\For{$b\in\{1,..,B\}$}
\State $\eta_t\sim i.i.d. \,(0,1)$
\State $y_{t}^b\gets \hat\epsilon_{t,(2)}\eta_t$
\For{$\gamma\in\Gamma$}
\State $WG_{T,BR}^b(\gamma)\gets$ use $\{(y_{t}^b,w_t)\}_t$ to compute RHS in \eqref{equ:TMTS.6} using $\hat\epsilon_{t,(1)}^b$ in $\hat H_{\epsilon,i\gamma}^b$
\EndFor
\State $WG_{T,BR}^b\gets\sup_{\gamma\in\Gamma}WG_{T,BR}^b(\gamma)$
\EndFor
\State \textbf{Return} $\{WG_{T,BR}^b\}_b$
\end{algorithmic}
\end{algorithm}
\end{minipage}\par\vspace*{1cm}
\indent Algorithm 3 below describes the wild fixed regressor bootstrap for our proposed 2SLS test-statistics, where regressors  $w_t$ and instruments $z_{1,t}$ are kept fixed in the bootstrap.  The first stage linearity or lack thereof is taken into account in computing $\hat x_t$ - equation \eqref{equ:TMTS.8} or \eqref{equ:TMTS.9c} respectively. For these tests, we need to know whether the first stage is linear or not; however, this is not necessarily a drawback in empirical work, because such knowledge is required for estimating the threshold parameter $\gamma^0$ consistently (see CH).\\
\begin{minipage}[t]{\textwidth}
\begin{algorithm}[H]
\centering
\caption{2SLS Bootstrap}\label{alg:Us2SLS}
\begin{algorithmic}[1]
\State\textbf{Input:} $\{y_t,x_t,z_t,\hat x_t\}_t,\,\hat\theta,\,\Gamma,\,B$
\State $w_t\gets(x_t^\top,z_{1t}^\top)^\top$
\State $\hat\epsilon_t\gets y_t-w_t^\top\hat\theta$
\State $\hat u_t\gets x_t-\hat x_t$
\For{$b\in\{1,...,B\}$}
\State $\eta_t\sim(0,1)$
\State $u_t^b\gets\hat u_t\eta_t$
\State $x_t^b\gets\hat x_t+u_t^b$
\State $w_t^b\gets(x_t^{b\top},z_{1,t}^\top)^\top$
\State $\epsilon_t^b\gets\hat \epsilon_t\eta_t$
\State $y_t^b\gets w_t^{b\top}\hat\theta+\epsilon_t^b$
\For{$\gamma\in\Gamma$}
\State $LR_T^b(\gamma)\gets$ use $(y_t^b,w_t^b,z_t)$ to compute RHS in \eqref{equ:TMTS.7a} using either \eqref{equ:TMTS.8} or \eqref{equ:TMTS.9a}--\eqref{equ:TMTS.9c}
\State $W_T^b(\gamma)\gets$ use $(y_t^b,w_t^b,z_t)$  \,\,\,to compute RHS in \eqref{equ:TMTS.7b} using either \eqref{equ:TMTS.8} or \eqref{equ:TMTS.9a}--\eqref{equ:TMTS.9c}
\EndFor
\State $LR_T^b\gets\sup_{\gamma\in\Gamma}LR_T^b(\gamma)$
\State $W_T^b\gets\sup_{\gamma\in\Gamma}W_T^b(\gamma)$
\EndFor
\State \textbf{Return:} $\{LR_T^b,W_T^b\}_b$
\end{algorithmic}
\end{algorithm}
\end{minipage}\par\vspace*{1cm}
We now derive the asymptotic properties of our tests\footnote{We focus on the GMM-based test; the asymptotic distributions of the 2SLS tests are given in the Online Supplement, Sections \ref{sec:ProofsLFSAsyms} and \ref{sec:ProofsTFSAsyms}.} and show their bootstrap validity.  First define $g_t = x_t - u_t$, $h_t = y_t-\epsilon_t$, $M_{1\gamma}=E[z_tz_t^\top\indgam]$, $M=\plim_{\gamma\to\infty}M_{1\gamma}=E[z_tz_t^\top ]$, $M_{2\gamma}=M-M_{1\gamma},$ and $v_t=(\epsilon_t,u_t^\top)^\top$. Let $\|\cdot \|$ be the Euclidean norm. The following assumptions are similar to CH.     
\begin{assumption}\label{assn:1}\hfill 
\begin{enumerate}[(a)]
\item $E[v_t|\mathfrak{F}_t]=0$ with $\mathfrak F_t=\sigma\{z_{t-s},v_{t-s-1},q_{t-s}|s\geq0\}$;
\item The series $(v_t,g_t, h_t,q_t, z_t)$ is strictly stationary with $\rho$-mixing coefficient $\rho(m)=\mathcal O(m^{-A})$ for some $A>\frac{a}{a-1}$ and $1<a\leq r$; 
\item $E\Big[\|z_t\|^{4r}\Big]  <\infty$, $E\Big[\|v_t\|^{4r}\Big]<\infty$ for some $r>1$;
\item $\infl_{\gamma\in\Gamma^0}\det M_{1\gamma}>0$, and if \eqref{equ:TMTS.3} holds, then $M_{1\gamma_1}-M_{1\gamma_2}$ has eigenvalues bounded away from zero for any $\gamma_1>\gamma_2$. Moreover, $M_{1\gamma}$ is continuous at $\gamma=\rho^0$;
\item The threshold variable $q_t$ has a continuous pdf $f(q_t)$ with $\supl_{q_t \in \Gamma} |f(q_t)|<\infty$;
\item $E[v_tv_t^\top]$ and $E [(v_t v_t^\top) \otimes (z_t z_t^\top) ]$ are two p.d. matrices of constants;     
\item The coefficient matrices $\Pi^0$  (for the linear first stage \eqref{equ:TMTS.2}) or  $\Pi_1^0, \Pi_2^0$ (for the threshold first stage \eqref{equ:TMTS.3}) are full rank, and $\Pi_1^0-\Pi_2^0 \neq 0$.
\end{enumerate}
\end{assumption}
Most of these assumptions are also used in CH. Assumption \ref{assn:1} $(a)$ is typically needed for nonlinear models.  Assumption \ref{assn:1} $(b)$ is also needed, as the only uniform law of large numbers and functional central limit theorem for partial sums in $\indgam$ that we are aware of derives from \cite{hansen1996} and require strict stationarity (see Lemma \ref{lem:1}-\ref{lem:2} in the Online Supplement). Assumption \ref{assn:1} $(c)$ is a typical moment condition. Assumption \ref{assn:1} $(d)$ is slightly different than CH: they also impose that $M_{1\gamma}$ is p.d. for all $\gamma$, but we require that the increments in $M_{1\gamma}$ are p.d. in the limit with eigenvalues bounded away from zero. The latter is technical in nature and required to obtain quantities bounded in probability in order to provide a self-contained proof of super-consistency of $\hat \rho$ in a threshold first stage model. Assumptions \ref{assn:1} $(e)$ is standard in the threshold literature, and \ref{assn:1} $(g)$ is an identification condition for a possible threshold in the first stage. Assumption \ref{assn:1} $(f)$ is needed for uniqueness of the asymptotic distributions of the test statistics proposed.\\
\indent With this assumption, we first show that employing the new heteroskedasticity-robust estimators do not alter the distribution of the CH test.
\begin{theorem}[Asymptotics Modified GMM Wald-Test]\label{theo:GMM.1}
Let $y_t$ be generated by \eqref{equ:TMTS.1}, $x_t$ be generated by either a linear first stage \eqref{equ:TMTS.2} or a threshold first stage \eqref{equ:TMTS.3}. Moreover, let the GMM test be computed as in \eqref{equ:TMTS.6} but with $\hat H_{\epsilon,i\gamma}$ constructed using $\hat\epsilon_{t,(1)}$. Then, under $\mathbb H_0$ and Assumption \ref{assn:1}, $WG_{T,BR}$ has the same limiting distribution as $WG_{T,CH}$.
\end{theorem}
To show the validity of the null bootstrap, we require the following additional assumption: 
\begin{assumption}\label{assn:2}\hfill
\begin{enumerate}[(a)]
\item $\eta_t \stackrel{iid}{\sim} (0,1)$ with $E^b[\eta_t^4]<\infty$, where $E^b[\cdot]$ is the expectation with respect to the bootstrap probability measure; 
\item $E\Big[\|z_t\|^{8r}\Big]<\infty$ and $E\Big[\|v_t\|^{8r}\Big]<\infty$ for some $r>1$.
\end{enumerate}
\end{assumption}
Assumption \ref{assn:2} $(a)$ is common for the wild bootstrap \citep[also see][]{bch2019}, and typical choices for $\eta_t$ are the normal distribution, the Rademacher distribution, and the asymmetric two-point distribution in \cite{mammen1993}. CH propose using the normal distribution, but we use both the normal distribution and the \cite{mammen1993} distribution, as the latter yields better results for the GMM-Wald test, see Tables \ref{tab:GMM.1}--\ref{tab:2SLS.2}. Assumption \ref{assn:2} $(b)$ is only needed for $\hat H_{\epsilon,i\gamma}^b$ to weakly converge to $H_{\epsilon,i\gamma}^b$ in probability under the bootstrap measure. Theorem \ref{theo:GMM.2} proves the asymptotic validity of the null bootstrap for $WG_{T,BR}$.\begin{theorem}[Bootstrap Validity Modified GMM Wald-Test]\label{theo:GMM.2}
Let $y_t$ be generated by \eqref{equ:TMTS.1}, $x_t$ be generated by either a linear first stage \eqref{equ:TMTS.2} or a threshold first stage \eqref{equ:TMTS.3}. Moreover, let the GMM Wald-test be computed as in \eqref{equ:TMTS.6} but with $\hat H_{\epsilon,i\gamma}$ constructed using $\hat\epsilon_{t,(1)}$ and using the bootstrap in Algorithm \ref{alg:UsGMM}. Then, under $\mathbb H_0$ and Assumptions \ref{assn:1} and \ref{assn:2},
\[
\sup_{c\in\mathbb R}|P^b(WG_{T,BR}^b\leq c)-P(WG_{T,BR}\leq c)|\xrightarrow{p}0.
\]
\end{theorem}
For the 2SLS-based test-statistics, the  asymptotic distributions are cumbersome and not of main interest. Therefore,  we relegate these results to the Online Supplement, Sections \ref{sec:ProofsLFSAsyms} and \ref{sec:ProofsTFSAsyms}. However, in order to derive these asymptotic distributions, we also provide in the Online Supplement, Theorem 5, a self-contained proof of super-consistency of the  first stage (ordinary least-squares) threshold parameter estimate $\hat\rho$. This was also shown in \cite{chan1993}, but for a threshold autoregressive model where $z_t$ and $q_t$ are lags of $x_t$. This proof may be of interest in its own right, as it extends proof techniques from change point analysis to threshold models. 

The asymptotic distributions of the 2SLS based tests are also non-pivotal, and we conclude this section by stating the asymptotic validity of the bootstrap for these tests.
\begin{theorem} [Bootstrap Validity -- 2SLS Tests]\label{thm:1}
Let $y_{t}$ be generated by \eqref{equ:TMTS.1} and $x_t$ be generated by the linear first stage \eqref{equ:TMTS.2} or by the threshold first stage \eqref{equ:TMTS.3}. Moreover, let the 2SLS-based test-statistics be computed as described in \eqref{equ:TMTS.7a} and \eqref{equ:TMTS.7b}, and using the bootstrap in Algorithm \ref{alg:Us2SLS}. Then, under $\mathbb H_0$ and Assumption \ref{assn:1}, 
\begin{itemize}
\item[$(i)$] $\sup_{c\in\mathbb{R}}\left| P^{b} \left(LR_T^b\leq c)- P(LR_T\leq c\right)\right|\inp 0$ if Assumption \ref{assn:2} $(a)$ holds in addition;
\item[$(ii)$] $\sup_{c\in\mathbb{R}}\left| P^{b} \left(W_T^b\leq c)-P(W_T\leq c\right)\right|\inp 0$ if Assumption \ref{assn:2} holds in addition.
\end{itemize}
\end{theorem}

\section{Simulations}\label{sec:Simus}
Consider the following data generating process (DGP) for $t=1,\ldots, T$:
\begin{align*}
y_t & =1 + x_t+ \delta_x\, x_t\indinvgamo+ \epsilon_t, & x_t & =1+z_t+\delta_\Pi\, z_t\indinvrhoo+u_t,
\end{align*}
where $z_t\stackrel{iid}{\sim}\mathcal N(1,1)$, $q_t=z_t+1$, and $z_t,\,x_t$, and $q_t$ are scalars. We set $\delta_\Pi=0$ for a linear first stage (LFS) and $\delta_\Pi\in\{-0.5,0.5,1\}$ for a threshold first stage (TFS) with $\rho^0=1.75$. Under the null hypothesis, $\delta_x=0$, and under the alternative hypothesis, $\delta_x =0.25$ with $\gamma^0=2.25$.\footnote{Note that because of just-identification, there is no difference between the first and the second-step GMM estimators, therefore $\hat \theta_{(2)} =\hat \theta_{(1)} = (\sum_{t=1}^T w_t z_t^\top )^{-1} ( T^{-1} \suml_{t=1}^T z_t y_t )$, and $\hat \theta_{i\gamma,(2)} = \hat \theta_{i\gamma,(1)} = (\hat N_{i\gamma}^\top)^{-1}  ( T^{-1} \suml_{i\gamma} z_t y_t )$.} To generate $\epsilon_t$, we define $e_t\stackrel{iid}\sim\cN(0,1)$ and consider the following three cases.\\
\indent In case (a), the errors are homoskedastic, i.e. $\epsilon_t=e_t$, and the econometrician knows this. Therefore, we use the i.i.d. bootstrap instead of the wild bootstrap, and make two adjustments to the computation of the test statistics. First, $v_t^{b}\stackrel{iid}\sim\mathcal N(0,\hat{\Sigma}_v)$ with $\hat\Sigma_v=T^{-1}\sum_{t=1}^T\hat v_t\hat v_t^\top$ for 2SLS. For GMM, $\epsilon_t^b \stackrel{iid}\sim\mathcal N(0,\hat\sigma^2_\epsilon)$ given the data, with $\hat\sigma^2_\epsilon=T^{-1}\sum_{t=1}^T\hat\epsilon_{t,(1)}^2$, and $\hat \epsilon_{t,(1)} = y_t- w_t^\top \hat \theta$, with $\hat \theta$ the full-sample 2SLS estimator. Second, all heteroskedasticity-robust estimators are replaced by their homoskedastic analogs. For example, $E[z_tz_t^\top\epsilon_t^2\indgam]$ is no longer estimated by $T^{-1}\sum_{1\gamma} z_tz_t^\top\hat\epsilon_t^2\indgam$, but by $\hat\sigma_\epsilon^2(T^{-1}\sum_{1\gamma}z_tz_t^\top\indgam)$. The same applies to the case of CH's original bootstrap, except that the residuals are computed for each value of $\gamma$ rather than under $\mathbb H_0$.\\ 
\indent In case (b), the errors are still homoskedastic, i.e. $\epsilon_t=e_t$, but this is not known to the econometrician. Therefore, the heteroskedasticity-robust variance estimators described in Sections \ref{sec:TMTS} and \ref{sec:BootVal} are employed. \\
\indent In case (c), the errors are conditional heteroskedastic, i.e. $\epsilon_t=e_t\cdot z_t/\sqrt{2}$ with $Var(e_t)=Var(u_t)=1$ and $Cov(u_t,e_t)=0.5$, and heteroskedasticity-robust variance estimators are employed.\\
\indent In Tables \ref{tab:GMM.1} and \ref{tab:2SLS.1} cases (b)-(c), the bootstrap is performed using $\eta_t\stackrel{iid}{\sim}\mathcal N(0,1)$. For all other results, $\eta_t \stackrel{iid}{\sim}(0,1)$ with draws from the asymmetric two-point distribution proposed by \cite{mammen1993}. In all cases besides (a), we use the wild bootstrap as described in Section \ref{sec:BootVal}.\\
\indent There are $500$ bootstrap samples. For each simulation, we compute the 95\% quantile of the bootstrap distribution of the test statistic, and if the test in the original sample is above this quantile, we reject, else we do not reject. $\gamma$ is varied between all sample realizations of $q_t$ from its $15\%$ quantile to its $85\%$ quantile. We report the rejection frequency of each test statistic in $1000$ simulations under the null and at 5\% nominal size (Tables \ref{tab:GMM.1}--\ref{tab:2SLS.2}), and under the alternative we plot the size-adjusted power, where the size-adjustment  is made relative to the null DGPs described above (Figure \ref{fig:Power}).\\
\\
\indent Tables \ref{tab:GMM.1} and \ref{tab:GMM.2} show that the bootstrap procedure originally proposed by CH has heavy size distortions in both directions. In particular, the test moves from being heavily oversized in small samples to severely undersized in large samples (columns ``CH''). This originates from imprecise residual estimates and imprecise $\hat H_{\epsilon,i\gamma}$ for small and moderate sample sizes pertinent to applications. In particular, when $\gamma$ is close to the 15\% or 85\% quantiles of $q_t$, there is not enough data to obtain precise residuals and precisely estimate $H_{\epsilon,i\gamma}$. Moreover, changing the CH bootstrap to a null bootstrap results in oversized tests for all considered sample sizes (columns ``Mix''). This problem is rectified by  modifying $\hat H_{\epsilon,i\gamma}$ in the original test statistic, as evident from columns ``BR'' in Tables \ref{tab:GMM.1} and \ref{tab:GMM.2}, where the empirical sizes are much closer to the nominal size. Tables \ref{tab:2SLS.1} and \ref{tab:2SLS.2} show that the 2SLS tests are in almost all cases close to nominal sizes, even in small samples, and that there is no clear ranking among the three proposed tests.

\FloatBarrier
\begin{table}[!ht]
\begin{center}
\caption{Rejection frequencies under the null DGP, $5\%$ nominal size and $\eta_t\stackrel{iid}\sim\cN(0,1)$ -- GMM  Tests}\label{tab:GMM.1}
\begin{threeparttable}
\begin{tabular}{ccccccccccccccccccccc}
\toprule
 & & \multicolumn{3}{c}{LFS} & \multicolumn{3}{c}{TFS, $\delta_{\Pi}=-0.5$} & \multicolumn{3}{c}{TFS, $\delta_{\Pi}=0.5$} & \multicolumn{3}{c}{TFS, $\delta_{\Pi}=1$} \\ 
\cmidrule(lr){3-5}\cmidrule(lr){6-8}\cmidrule(lr){9-11}\cmidrule(lr){12-14}
 & $T$ & $CH$ & $Mix$ & $BR$ & $CH$ & $Mix$ & $BR$ & $CH$ & $Mix$ & $BR$ & $CH$ & $Mix$ & $BR$\\
\midrule
\multirow{4}{*}{$(a)$} & 100  & 11.5\% & 17.9\% & 4.8\% & 10.8\% & 17.0\% & 3.9\% & 10.0\% & 15.0\%  & 4.5\% & 8.0\% & 13.4\% & 4.5\%\\ 
                        & 250  & 8.6\% & 12.9\% & 5.8\% & 7.8\% & 13.6\% & 4.8\% & 6.1\% & 10.1\% & 5.7\% & 5.1\% & 8.9\% & 5.7\%\\ 
                        & 500  & 4.4\% & 11.2\% & 4.3\% & 4.6\% & 14.2\% & 5.5\% & 2.9\% & 9.6\% & 4.6\% & 2.5\% & 8.8\% & 4.6\%\\ 
                        & 1000 & 1.8\% & 9.5\% & 4.9\% & 2.9\% & 11.3\% & 4.9\% & 1.7\% & 9.1\% & 4.9\% & 1.2\% & 8.8\% & 4.5\%\\ 
\midrule
\multirow{4}{*}{$(b)$} & 100 & 12.0\% & 17.4\% & 7.5\% & 11.1\% & 17.3\% & 6.2\% & 10.5\% & 15.6\% & 7.9\%& 9.6\% & 15.0\% & 8.0\%\\ 
                        & 250 & 4.9\% & 11.8\% & 7.5\% & 6.0\% & 14.1\% & 7.2\% & 4.8\% & 10.9\% & 7.2\% & 4.3\% & 10.2\% & 7.1\%\\ 
                        & 500  & 2.4\% & 10.7\% & 6.4\% & 3.5\% & 12.3\% & 6.3\% & 2.3\% & 9.4\% & 6.9\% & 2.4\% & 9.1\% & 6.8\%\\ 
                        & 1000 & 0.7\% & 8.7\% & 4.6\% & 1.3\% & 10.6\% & 4.6\% & 0.7\% & 8.2\% & 4.3\% & 0.6\% & 7.8\% & 4.4\%\\ 
\midrule
\multirow{4}{*}{$(c)$} & 100 & 6.9\% & 13.8\% & 8.1\% & 7.2\% & 12.9\% & 6.8\% & 7.8\% & 12.6\% & 8.6\% & 8.0\% & 12.6\% & 8.9\%\\ 
                        & 250  & 2.7\% & 11.2\% & 6.3\% & 3.1\% & 11.9\% & 5.6\% & 2.8\% & 9.6\% & 5.9\% & 2.2\% & 8.8\% & 6.0\%\\ 
                        & 500  & 2.0\% & 8.6\% & 6.6\% & 2.3\% & 11.0\% & 7.6\% & 1.8\% & 7.8\% & 6.7\% & 1.7\% & 7.1\% & 6.6\%\\ 
                        & 1000 & 0.4\% & 9.7\% & 4.7\% & 0.8\% & 11.2\% & 4.6\% & 0.5\% & 8.8\% & 4.4\% & 0.5\% & 8.4\% & 3.8\%\\ 
\bottomrule
\end{tabular}
\begin{tablenotes}{\footnotesize
\item $CH$ refers to the $WG_{T,CH}$ test in CH (Algorithm \ref{alg:CHGMM}) and $BR$ refers to our modified $WG_{T,BR}$ test (Algorithm \ref{alg:UsGMM}). $Mix$ refers to the mixture between the two, where the bootstrap is conducted using $\mathbb H_0$ residuals, whereas $\hat H_{\epsilon,i\gamma}$ are constructed using $\mathbb H_A$-residuals. Panel $(a)$ corresponds to known homoskedasticity, panel $(b)$ to unknown homoskedasticity and panel $(c)$ to known heteroskedasticity.
}
\end{tablenotes}
\end{threeparttable}
\end{center}
\end{table}

\begin{table}[!ht]
\begin{center}
\caption{Rejection frequencies under the null DGP, $5\%$ nominal size and $\eta_t\stackrel{iid}\sim Mammen$ -- GMM  Tests}\label{tab:GMM.2}
\begin{threeparttable}
\begin{tabular}{ccccccccccccccccccccc}
\toprule
 & & \multicolumn{3}{c}{LFS} & \multicolumn{3}{c}{TFS, $\delta_{\Pi}=-0.5$} & \multicolumn{3}{c}{TFS, $\delta_{\Pi}=0.5$} & \multicolumn{3}{c}{TFS, $\delta_{\Pi}=1$} \\ 
\cmidrule(lr){3-5}\cmidrule(lr){6-8}\cmidrule(lr){9-11}\cmidrule(lr){12-14}
 & $T$ & $CH$ & $Mix$ & $BR$ & $CH$ & $Mix$ & $BR$ & $CH$ & $Mix$ & $BR$ & $CH$ & $Mix$ & $BR$\\
\midrule
\multirow{4}{*}{$(b)$} & 100   & 12.0\% & 13.1\% & 5.5\% & 11.7\% & 14.1\% & 4.9\% & 10.6\% & 11.6\% & 6.1\% & 9.6\% & 11.0\% & 5.9\%\\ 
                                    & 250   & 4.6\%  & 11.9\% & 5.5\% & 5.5\% & 13.7\% & 4.7\%   & 4.0\% & 10.5\% & 5.6\%   & 3.9\% & 9.5\% & 5.6\%\\ 
                                    & 500   & 2.2\%  & 8.7\% & 5.3\% & 2.7\% & 11.3\% & 4.8\%   & 1.8\% & 8.6\% & 5.0\%   & 1.8\% & 8.5\% & 5.2\%\\ 
                                    & 1000 & 0.9\%  & 9.3\% & 5.6\% & 1.7\% & 9.4\% & 5.5\%   & 0.7\% & 8.8\% & 5.7\%   & 0.6\% & 8.0\% & 5.4\%\\ 
\midrule
\multirow{4}{*}{$(c)$} & 100  & 7.9\% & 10.5\% & 6.1\% & 7.3\% & 10.6\% & 4.8\% & 7.9\% & 9.4\% & 6.3\% & 8.3\% & 8.9\% & 6.3\%\\ 
                                   & 250  & 3.3\% & 9.9\% & 5.0\% & 3.7\% & 11.4\% & 4.5\% & 2.9\% & 7.7\% & 4.5\% & 2.8\% & 7.3\% & 4.1\%\\ 
                                   & 500  & 1.2\% & 6.6\% & 4.8\% & 1.5\% & 7.8\% & 4.6\% & 1.2\% & 5.0\% & 4.9\% & 1.0\% & 4.6\% & 4.9\%\\ 
                                   & 1000 & 0.8\% & 6.9\% & 5.3\% & 1.3\% & 9.3\% & 5.1\% & 0.6\% & 6.4\% & 5.4\% & 0.3\% & 5.9\% & 5.7\%\\ 
\bottomrule
\end{tabular}
\begin{tablenotes}{\footnotesize
\item See Table 1 Notes.
}
\end{tablenotes}
\end{threeparttable}
\end{center}
\end{table}
\begin{table}[!ht]
\begin{center}
\caption{Rejection frequencies under the null DGP, $5\%$ nominal size and $\eta_t\stackrel{iid}\sim\cN(0,1)$ -- 2SLS Tests}\label{tab:2SLS.1}
\begin{threeparttable}
\begin{tabular}{ccccccccccccc}
\toprule
 & & \multicolumn{2}{c}{LFS} & \multicolumn{2}{c}{TFS, $\delta_{\Pi}=-0.5$} & \multicolumn{2}{c}{TFS, $\delta_{\Pi}=0.5$} & \multicolumn{2}{c}{TFS, $\delta_{\Pi}=1$} \\ 
\cmidrule(lr){3-4}\cmidrule(lr){5-6}\cmidrule(lr){7-8}\cmidrule(lr){9-10}
 & $T$ & $LR$ & $W$ & $LR$ & $W$ & $LR$ & $W$ & $LR$ & $W$ \\
\midrule
\multirow{4}{*}{$(a)$} & 100  & 6.0\% & 6.2\% & 2.9\% & 4.5\% & 2.1\% & 2.4\% & 2.9\% & 2.6\% \\
                        & 250  & 5.0\% & 5.0\% & 2.0\% & 4.0\% & 2.6\% & 3.2\% & 3.4\% & 3.8\% \\
                        & 500  & 5.3\% & 5.0\% & 2.8\% & 3.4\% & 2.8\% & 2.8\% & 4.5\% & 4.3\% \\
                        & 1000 & 4.2\% & 4.4\% & 3.9\% & 3.4\% & 4.2\% & 3.3\% & 4.9\% & 3.7\%\\
\midrule
\multirow{4}{*}{$(b)$} & 100  & 4.6\% & 8.4\% & 2.3\% & 8.2\% & 1.5\% & 5.5\% & 2.1\% & 7.2\%\\
                        & 250  & 5.3\% & 8.8\% & 2.3\% & 6.0\% & 2.3\% & 5.4\% & 4.2\% & 7.0\%\\
                        & 500  & 5.7\% & 6.7\% & 4.0\% & 6.0\% & 3.3\% & 5.6\% & 5.5\% & 6.7\%\\
                        & 1000 & 6.5\% & 6.9\% & 3.0\% & 4.8\% & 3.7\% & 5.5\% & 4.8\% & 6.3\% \\
\midrule
\multirow{4}{*}{$(c)$} & 100  & 5.1\% & 8.2\% & 3.7\% & 7.8\% & 2.7\% & 5.1\% & 3.2\% & 6.3\% \\
                        & 250  & 5.3\% & 6.9\% & 3.9\% & 5.4\% & 3.9\% & 5.4\% & 5.7\% & 6.7\%\\
                        & 500  & 7.0\% & 7.6\% & 5.9\% & 6.6\% & 5.9\% & 5.2\% & 6.9\% & 7.2\% \\
                        & 1000 & 6.6\% & 5.3\% & 5.8\% & 4.1\% & 5.2\% & 3.5\% & 5.5\% & 4.0\% \\
\bottomrule
\end{tabular}
\begin{tablenotes}{\footnotesize
\item $LR$ and $W$ refer to the $LR_T$ and $W_T$ tests.
 Panel $(a)$ corresponds to known homoskedasticity, panel $(b)$ to unknown homoskedasticity and panel $(c)$ to known heteroskedasticity.
}
\end{tablenotes}
\end{threeparttable}
\end{center}
\end{table}

\begin{table}[!ht]
\begin{center}
\caption{Rejection frequencies under the null DGP, $5\%$ nominal size and $\eta_t\stackrel{iid}\sim Mammen$ -- 2SLS Tests}\label{tab:2SLS.2}
\begin{threeparttable}
\begin{tabular}{ccccccccccccc}
\toprule
 & & \multicolumn{2}{c}{LFS} & \multicolumn{2}{c}{TFS, $\delta_{\Pi}=-0.5$} & \multicolumn{2}{c}{TFS, $\delta_{\Pi}=0.5$} & \multicolumn{2}{c}{TFS, $\delta_{\Pi}=1$} \\ 
\cmidrule(lr){3-4}\cmidrule(lr){5-6}\cmidrule(lr){7-8}\cmidrule(lr){9-10}
 & $T$ & $LR$ & $W$ & $LR$ & $W$ & $LR$ & $W$ & $LR$ & $W$ \\
\midrule
\multirow{4}{*}{$(b)$} & 100  & 5.0\% & 5.9\% & 2.1\% & 6.3\% & 2.0\% & 4.1\% & 2.2\% & 5.1\% \\
                        & 250  & 5.6\% & 6.8\% & 2.5\% & 5.2\% & 1.9\% & 4.2\% & 3.3\% & 5.1\% \\ 
                        & 500  & 4.4\% & 4.6\% & 3.2\% & 4.0\% & 3.1\% & 4.1\% & 4.7\% & 4.3\% \\
                        & 1000 & 5.4\% & 6.0\% & 4.6\% & 4.7\% & 3.0\% & 4.5\% & 4.4\% & 5.2\% \\
\midrule
\multirow{4}{*}{$(c)$} & 100  & 5.3\% & 5.2\% & 3.8\% & 6.4\% & 3.2\% & 3.9\% & 3.7\% & 4.0\%\\
                        & 250  & 5.3\% & 5.6\% & 3.4\% & 4.8\% & 3.3\% & 3.8\% & 4.4\% & 3.9\%\\
                        & 500  & 5.3\% & 5.9\% & 4.6\% & 2.8\% & 4.7\% & 3.2\% & 5.4\% & 4.1\% \\
                        & 1000 & 5.1\% & 6.6\% & 4.0\% & 4.8\% & 4.4\% & 4.3\% & 4.4\% & 5.3\% \\
\bottomrule
\end{tabular}
\begin{tablenotes}{\footnotesize
\item See Table 3 Notes
}
\end{tablenotes}
\end{threeparttable}
\end{center}
\end{table}
\FloatBarrier

\indent We also assess the power of all tests. For a large threshold $\delta_x=1$, all tests have power virtually equal to one even for sample sizes of $T=250$ and therefore we do not report these results. Figure \ref{fig:Power} shows the power properties for a small threshold of $\delta_x=0.25$. In small samples, the Wald tests dominate the LR test for all cases (a)-(c). Note that this is not necessarily for classical reasons of correcting for heteroskedasticity, as all tests are non-pivotal and bootstrapped. The power differences among all three tests vanish as the sample size grows. Therefore, we argue that all the tests proposed provide reliable alternatives in moderate samples pertinent to macroeconomic applications.

\FloatBarrier
\begin{figure}[ht!]
\hspace*{-2cm}
\caption{Size-Adjusted Power Plots for $\delta_x=0.25$}\label{fig:Power}
\includegraphics[width=1.1\textwidth]{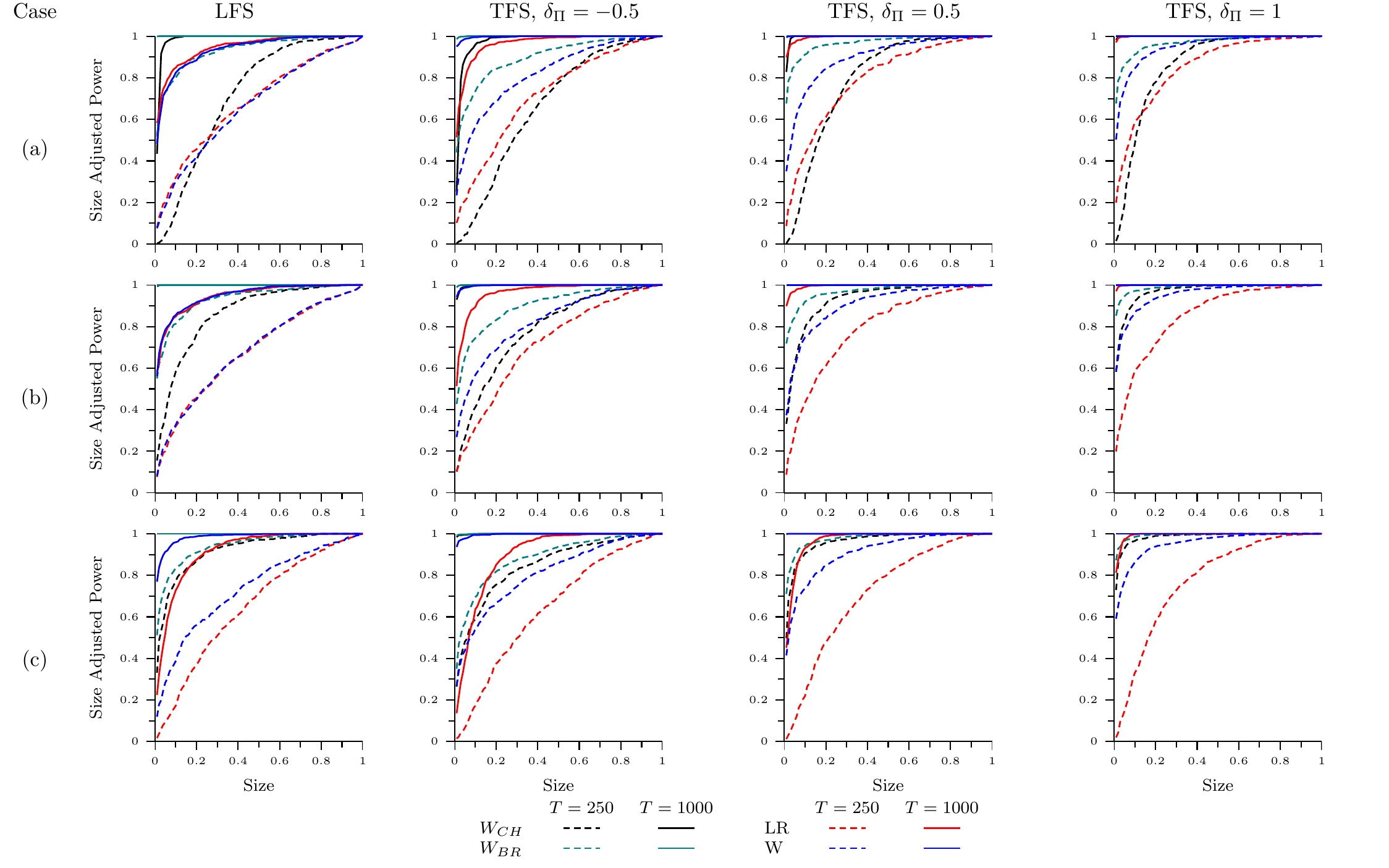}
\floatfoot{\normalsize $LR$ and $W$ refer to our $LR_T$ and $W_T$ tests. $CH$ refers to the $WG_{T,CH}$ test in CH (Algorithm \ref{alg:CHGMM}) and $BR$ refers to our modified $WG_{T,BR}$ test (Algorithm \ref{alg:UsGMM}).}
\end{figure}
\FloatBarrier

\section{Application to government spending multipliers}\label{sec:Appl}   
In this section, we revisit the question  whether government spending is more effective in recessions, and address it as in RZ, using exactly the same data and model specifications, except that we test and estimate an unknown threshold rather than imposing it. For simplicity, we first focus on the instantaneous government spending multiplier $\theta_{g,i}(i=1,2)$, estimated similarly to RZ from:
\begin{align}
y_{t} & =(\theta_{g,1} \, g_{t}+z_{1,t}^\top \theta_{z,1}) \indgamo +(\theta_{g,2} \, g_{t}+ z_{1,t}^\top \theta_{z,2}) \indinvgamo)+\epsilon_{t}\label{equ:RZ1} \\
g_t &= \Pi_1^\top z_t \indrhoo +\Pi_2^\top z_t \indinvrhoo+v_{t}\label{equ:RZ2}
\end{align}
where $y_t$ is real GDP divided by trend GDP, $g_t$ is real government spending divided by trend GDP -- which is endogenous and instrumented by military spending news $m_t$ -- and the threshold variable is $q_{t}$, the first lag of the unemployment rate. The exogenous regressors $z_{1t}$ are also included in $z_t$ and contain an intercept and four lags of $g_t,\,y_t,\,m_t$. Thus, $z_t=[z_{1t}^\top, m_t^\top]^\top$.\\
\indent The data is from the RZ replication package.\footnote{\url{http://econweb.ucsd.edu/~vramey/research/Ramey_Zubairy_replication_codes.zip}} For details on the data construction, instrument validity, or interpretation of $\theta_{g,i}(i=1,2)$ as cumulative spending multipliers, we refer the interested reader to RZ. \\
\indent Letting $\theta_i^0=(\theta_{g,i}, \theta_{z,i}^\top)^\top$ and $w_t = (g_t, z_{1,t}^\top)^\top$, the RZ estimators  of $\theta_i^0$ are exactly the just-identified GMM (or instrumental variables, IV henceforth) estimators $\hat \theta_{i\gamma,(1)}$ defined in Section \ref{sec:TMTS}, but evaluated in RZ at $\gamma=6.5$ (and ignoring the first stage which is irrelevant for conventional IV estimators).\footnote{All numbers referring to unemployment rates, such as $6.5$, should be interpreted as percentages: $6.5\%$.} The threshold $\gamma= 6.5$ is chosen by RZ as in \cite{owyangetal2013}, based on the US Federal Reserve use of this threshold in its policy announcement; RZ also do a robustness check with a threshold of $8.0$. Since it is unclear why $6.5$ or $8.0$ would be the threshold that defines recessions versus expansions,  we do not assume that the threshold $\gamma^0$ is known or even that there is a threshold $\gamma^0$; we instead test for the presence of $\gamma^0$ first. \\
\indent The 2SLS tests require first estimating $\rho^0$ in equation \eqref{equ:RZ2}. Table \ref{tab:FSThresh} reports the multivariate threshold estimates $\hat \rho$ described in Section \ref{sec:TMTS}, along with the decisions of a LFS or a TFS based on the BIC3 criterion proposed in \cite{gonpit2002} and on the ordinary least-squares (OLS) versions of $LR_T^b$ and $W_T^b$ tests described in Section \ref{sec:BootVal}, which were proposed in \cite{hansen1996}. The estimate of $\rho^0$  change with the cut-off considered, but there is considerable evidence of a threshold in the first stage. The maximizer of the OLS version of $LR_T(\gamma)$ is  exactly $\hat \rho$, a consistent estimator of $\rho^0$ as shown in Theorem \ref{theo:0}. Therefore, we use a TFS with $\hat \rho$ in Table \ref{tab:FSThresh}. \\
\FloatBarrier
\begin{table}[!ht]
\begin{center}
\caption{Presence of Thresholds in the First Stage}\label{tab:FSThresh}
\begin{threeparttable}
\begin{tabular}{ccccc}
\toprule
 Trim & $\hat{\rho}$ & BIC3 & LR & W \\ 
\midrule 
 10\% & 3.5264 & TFS & TFS & LFS \\ 
 15\% & 3.5264 & TFS & TFS & LFS \\ 
 20\% & 3.7530 & LFS & TFS & LFS \\ 
 25\% & 4.0636 & LFS & TFS & LFS \\ 
\bottomrule 
\end{tabular}
\begin{tablenotes}{\footnotesize
\item BIC3 is the BIC3 criterion in Gonzalo and Pitarakis (2002), and LR and W are the OLS bootstrap equivalents of our tests $LR_T$ and $W_T$
}
\end{tablenotes}
\end{threeparttable}
\end{center}
\end{table}

\begin{table}[!ht]
\begin{center}
\caption{Presence of Thresholds in the Equation of Interest}\label{tab:EIThresh}
\begin{threeparttable}
\begin{tabular}{ccccccccc}
\toprule
& & & \multicolumn{6}{c}{GMM Tests}\\
 \cmidrule(lr){4-9}
 Trim & TFS $\hat{\rho}$ & $\hat{\gamma}$ & $WG_{T,CH}$ & $5\%CV$ & Reject & $WG_{T,BR}$ & $5\%CV$ & Reject \\ 
 \midrule 
 10\% & 3.5264 & 11.9660 & 178.296 & 1213.488 & No & 25.258 & 27.007 & No \\ 
 15\% & 3.5264 & 10.7000 & 66.523 & 1097.046 & No & 25.258 & 27.327 & No \\ 
 20\% & 3.7530 & 9.3443 & 66.523 & 809.764 & No & 25.258 & 26.631 & No \\ 
 25\% & 4.0636 & 8.3363 & 66.523 & 250.830 & No & 22.787 & 27.133 & No \\ 
 \midrule\midrule
 & & & \multicolumn{6}{c}{2SLS Tests}\\
 \cmidrule(lr){4-9}
 Trim & TFS $\hat{\rho}$ & $\hat{\gamma}$  & $LR_T$ & $5\%CV$ & Reject & $W_T$ & $5\%CV$ & Reject \\
 \midrule
 10\% & 3.5264 & 11.9660 & 93.526 & 74.011 & Yes & 29.335 & 29.348 & No\\
 15\% & 3.5264 & 10.7000  & 78.158 & 56.963 & Yes & 27.698 & 29.348 & No \\
 20\% & 3.7530 & 9.3443  & 75.332 & 52.066 & Yes & 27.505 & 29.197 & No \\
 25\% & 4.0636 & 8.3363 & 65.719 & 50.498 & Yes & 24.365 & 29.137 & No \\
\bottomrule 
\end{tabular}
\begin{tablenotes}{\footnotesize
\item ``$5\%CV$' display the bootstrap $5\%$ critical values, ``Reject'' indicates whether the null of no threshold in (8) is rejected. For all specifications, a TFS is used with $\hat{\rho}$ obtained with the same cut-offs in the first stage as the column ``Trim'' indicates.
}
\end{tablenotes}
\end{threeparttable}
\end{center}
\end{table}

\indent Given $\hat \rho$ obtained for each cut-off, we test for an unknown threshold in equation \eqref{equ:RZ1}. Table \ref{tab:EIThresh} shows that the LR test rejects the null. The 2SLS Wald test and our modified GMM Wald-test do not reject (but their values are relatively close to the critical values at certain cut-offs). From Figure \ref{fig:TestSeq}, it is  evident that the sequence of all our test statistics  are relatively flat for all values of $\gamma$. The CH test also never rejects the null, but its sequence is not flat: its value is relatively large at 10\% trimming, and  its critical values are very large at all trimming levels. This is in line with our simulations, which indicated that the tests are undersized at 500 observations, the number of observations in our sample. Its erratic behavior near the sample edges was further illustrated in Figure 1.\\
\indent Because Equations \eqref{equ:RZ1}--\eqref{equ:RZ2} control for several lags -- in line with the RZ specification -- we choose the 25\% cut-off results with $\hat \rho=4.0636$ and $\hat \gamma=8.3363$, where the latter is the 2SLS threshold estimate proposed in CH (or, equivalently, the implicit maximizer of the $LR_T(\gamma)$ quantity in this paper).\footnote{The confidence sets for both these thresholds obtained by inverting the likelihood ratio tests in \cite{hansen2000} and CH, or by simulating the asymptotic distribution in CH, are very tight when using the default nonparametric kernel. However, since both estimators are close to the 25\% cut-off, and increase ($\hat \gamma$) or decrease ($\hat \rho$)  when decreasing the cut-offs used, we can only interpret these estimators as close to the lower bounds of the true threshold values that are identified in the sample.} \\
\indent We could conclude based on Figure \ref{fig:TestSeq} and Table \ref{tab:EIThresh} that there is little evidence that the instantaneous multipliers are different in recessions and expansions. In what follows, we also show that there is little evidence that the multipliers at other horizons than zero are different. To that end, as in RZ, we compute the cumulative government spending multipliers $\theta_{g,i}^h (i=1,2)$ at horizon $h=1,\ldots,H$ from the IV regression:
\begin{align*}
\suml_{h=0}^{H} y_{t+h}^{h} & =(\theta_{g,1}^h \, \suml_{h=0}^H g_{t+h}+z_{1,t}^\top \theta_{z,1}^h) \indgamhat \\
&+(\theta_{g,2}^h \, \suml_{h=0}^H g_{t+h}+ z_{1,t}^\top \theta_{z,2}^h) \indinvgamhat+\epsilon_{t},
\end{align*}   
where  $\suml_{h=0}^H g_{t+h}$ is instrumented by $m_t$.\footnote{It is unclear how to use the TFS specification \eqref{equ:RZ2} to obtain cumulative government spending multipliers at $h>0$, because of the misalignment between the first and the second stage threshold, and we leave this to future research.}\\
\indent Tables \ref{tab:RZMult}-\ref{tab:UsMult}  show the RZ multipliers (using $6.5$ and $8$ - robustness check in RZ - as thresholds), and our multipliers for fifteen quarters ahead, calculated exactly as in RZ but with $\hat \gamma=8.3383$. We also report classical heteroskedasticity and autocorrelation (HAC) robust standard errors, weak instrument HAC robust confidence sets, and classical and weak-instrument HAC-robust tests for the difference in multipliers at the imposed thresholds. These tables show that in all cases, there is no evidence that government spending multipliers are different in recessions, once the possibility of weak instruments is taken into account at all horizons. 
\FloatBarrier
\begin{table}[!ht]
\advance\leftskip-1.8cm
\centering
\caption{IV Multipliers with RZ threshold}\label{tab:RZMult}
\begin{threeparttable}
\scriptsize{\begin{tabular}{lccccccccrr}
\toprule
 & \multicolumn{4}{c}{State 1, $q_t\leq6.5$: 319 obs.} & \multicolumn{4}{c}{State 2, $q_t>6.5$: 181 obs. } & & \\
 \cmidrule(lr){2-5}\cmidrule(lr){6-9}
$h$ & Mult. & s.e. & AR LB & AR UB & Mult. & s.e. & AR LB & AR UB & $p$-val. & AR $p$-val.\\
\midrule
0  & 1.24 & 0.45 & -0.51 & 2.99 & -0.61 & 0.98 & -4.43 & 3.22 & 0.04 & 0.22\\
1  & 1.11 & 0.29 & -0.02 & 2.24 & -1.92 & 1.54 & -7.95 & 4.10 & 0.04 & 0.24\\
2  & 0.89 & 0.19 &  0.13 & 1.64 & -0.17 & 0.25 & -1.16 & 0.81 & 0.00 & 0.24\\
3  & 0.71 & 0.14 &  0.15 & 1.28 &  0.22 & 0.16 & -0.42 & 0.87 & 0.01 & 0.25\\
4  & 0.64 & 0.12 &  0.17 & 1.12 &  0.46 & 0.14 & -0.09 & 1.01 & 0.26 & 0.39\\
5  & 0.63 & 0.10 &  0.24 & 1.03 &  0.54 & 0.12 &  0.08 & 1.00 & 0.52 & 0.57\\
6  & 0.62 & 0.09 &  0.26 & 0.99 &  0.59 & 0.11 &  0.17 & 1.01 & 0.81 & 0.82\\
7  & 0.59 & 0.09 &  0.24 & 0.95 &  0.60 & 0.10 &  0.23 & 0.97 & 0.95 & 0.95\\
8  & 0.59 & 0.09 &  0.23 & 0.95 &  0.62 & 0.09 &  0.29 & 0.95 & 0.82 & 0.82\\
9  & 0.62 & 0.10 &  0.25 & 1.00 &  0.63 & 0.08 &  0.33 & 0.92 & 0.97 & 0.97\\
10 & 0.66 & 0.10 &  0.27 & 1.05 &  0.64 & 0.07 &  0.37 & 0.91 & 0.87 & 0.87\\
11 & 0.68 & 0.10 &  0.28 & 1.08 &  0.64 & 0.07 &  0.39 & 0.90 & 0.79 & 0.80\\
12 & 0.68 & 0.11 &  0.27 & 1.10 &  0.65 & 0.06 &  0.41 & 0.90 & 0.81 & 0.82\\
13 & 0.68 & 0.11 &  0.26 & 1.11 &  0.67 & 0.06 &  0.44 & 0.89 & 0.89 & 0.90\\
14 & 0.68 & 0.11 &  0.24 & 1.13 &  0.68 & 0.05 &  0.47 & 0.89 & 0.99 & 0.99\\
15 & 0.67 & 0.12 &  0.19 & 1.15 &  0.68 & 0.05 &  0.48 & 0.88 & 0.92 & 0.92\\
\bottomrule
\end{tabular}}
{\footnotesize
\begin{tablenotes}
\item "Mult." indicates the IV estimates at each horizon, "obs." the number of observations, and "s.e." the Newey-West HAC standard errors using the Bartlett kernel and the data-dependent bandwidth. "AR LB (AR UB)" refer to $95\%$ Anderson-Rubin confidence lower (upper) bounds. "$p$-val." indicate classical $p$-values for the $t$-test of no difference between the multipliers, and "AR $p$-val." indicate Anderson-Rubin $p$-values for the same test. All the results are computed with RZ's replication package code.
\end{tablenotes}
}
\end{threeparttable}
\end{table}

\begin{table}[!ht]
\centering
\caption{IV Multipliers with a threshold equal to $8$ (robustness check in RZ)}\label{tab:RobMult}
\begin{threeparttable}
\scriptsize{\begin{tabular}{lccccccccrr}
\toprule
 & \multicolumn{4}{c}{State 1, $q_t\leq8$: 396 obs. } & \multicolumn{4}{c}{State 2, $q_t>8$: 104 obs.} & & \\
 \cmidrule(lr){2-5}\cmidrule(lr){6-9}
$h$ & Mult. & s.e. & AR LB & AR UB & Mult. & s.e. & AR CI LB & AR CI UB & $p$-val. & AR $p$-val.\\
\midrule
0  & 1.29 & 0.37 & -0.17 & 2.74 & -0.44 & 1.11 & -4.79 & 3.92 &  0.10 & 0.26\\
1  & 1.13 & 0.24 &  0.20 & 2.05 & -1.05 & 1.25 & -5.94 & 3.83 &  0.08 & 0.30\\
2  & 0.92 & 0.15 &  0.31 & 1.52 & -0.29 & 0.41 & -1.89 & 1.31 &  0.01 & 0.31\\
3  & 0.73 & 0.12 &  0.27 & 1.19 &  0.16 & 0.25 & -0.82 & 1.13 &  0.05 & 0.36\\
4  & 0.66 & 0.10 &  0.26 & 1.07 &  0.48 & 0.23 & -0.42 & 1.38 &  0.47 & 0.59\\
5  & 0.65 & 0.09 &  0.28 & 1.02 &  0.65 & 0.21 & -0.17 & 1.47 &  0.99 & 0.99\\
6  & 0.63 & 0.09 &  0.26 & 1.00 &  0.76 & 0.21 & -0.06 & 1.57 &  0.60 & 0.58\\
7  & 0.60 & 0.09 &  0.24 & 0.97 &  0.80 & 0.21 & -0.01 & 1.60 &  0.40 & 0.42\\
8  & 0.60 & 0.09 &  0.24 & 0.95 &  0.79 & 0.18 &  0.08 & 1.50 &  0.35 & 0.38\\
9  & 0.63 & 0.09 &  0.28 & 0.98 &  0.76 & 0.15 &  0.18 & 1.34 &  0.45 & 0.48\\
10 & 0.66 & 0.09 &  0.30 & 1.01 &  0.75 & 0.13 &  0.25 & 1.24 &  0.58 & 0.58\\
11 & 0.68 & 0.09 &  0.31 & 1.04 &  0.72 & 0.11 &  0.28 & 1.16 &  0.77 & 0.77\\
12 & 0.68 & 0.09 &  0.31 & 1.05 &  0.71 & 0.11 &  0.29 & 1.12 &  0.85 & 0.84\\
13 & 0.67 & 0.10 &  0.29 & 1.05 &  0.72 & 0.10 &  0.33 & 1.12 &  0.72 & 0.72\\
14 & 0.66 & 0.10 &  0.26 & 1.07 &  0.75 & 0.10 &  0.36 & 1.13 &  0.55 & 0.55\\
15 & 0.65 & 0.12 &  0.20 & 1.10 &  0.76 & 0.10 &  0.37 & 1.15 &  0.45 & 0.43\\
\bottomrule
\end{tabular}}
{\footnotesize
\begin{tablenotes}
\item See Table \ref{tab:RZMult} notes.\end{tablenotes}
}
\end{threeparttable}
\end{table}
 \begin{table}[!ht]
\centering
\caption{IV Multipliers with our threshold}\label{tab:UsMult}
\begin{threeparttable}
\scriptsize{\begin{tabular}{lccccccccrr}
\toprule
 & \multicolumn{4}{c}{State 1, $q_t\leq8.3363$: 410 obs.} & \multicolumn{4}{c}{State 2, $q_t>8.3363$: 90 obs.} & & \\
 \cmidrule(lr){2-5}\cmidrule(lr){6-9}
$h$ & Mult. & s.e. & AR LB & AR UB & Mult. & s.e. & AR LB & AR UB & $p$-val. & AR $p$-val.\\
\midrule
0  & 1.30 & 0.38 & -0.20 & 2.80 & -0.93 & 1.36 & -6.27 &  4.41 & 0.05 & 0.19\\
1  & 1.14 & 0.25 &  0.17 & 2.11 & -1.68 & 1.37 & -7.05 &  3.68 & 0.03 & 0.23\\
2  & 0.93 & 0.16 &  0.28 & 1.57 & -0.55 & 0.48 & -2.42 &  1.33 & 0.00 & 0.25\\
3  & 0.74 & 0.13 &  0.25 & 1.23 & -0.02 & 0.20 & -0.80 &  0.76 & 0.00 & 0.26\\
4  & 0.67 & 0.11 &  0.24 & 1.09 &  0.32 & 0.17 & -0.36 &  1.00 & 0.07 & 0.31\\
5  & 0.65 & 0.10 &  0.28 & 1.03 &  0.52 & 0.18 & -0.19 &  1.22 & 0.45 & 0.52\\
6  & 0.63 & 0.10 &  0.26 & 1.01 &  0.62 & 0.19 & -0.14 &  1.38 & 0.94 & 0.94\\
7  & 0.60 & 0.09 &  0.24 & 0.96 &  0.66 & 0.20 & -0.13 &  1.45 & 0.77 & 0.78\\
8  & 0.60 & 0.09 &  0.25 & 0.94 &  0.66 & 0.18 & -0.04 &  1.37 & 0.71 & 0.73\\
9  & 0.62 & 0.09 &  0.29 & 0.96 &  0.66 & 0.15 &  0.07 &  1.25 & 0.84 & 0.85\\
10 & 0.65 & 0.09 &  0.31 & 0.10 &  0.66 & 0.13 &  0.14 &  1.18 & 0.96 & 0.96\\
11 & 0.67 & 0.09 &  0.32 & 1.02 &  0.65 & 0.12 &  0.17 &  1.13 & 0.87 & 0.87\\
12 & 0.67 & 0.09 &  0.32 & 1.03 &  0.64 & 0.12 &  0.17 &  1.11 & 0.85 & 0.85\\
13 & 0.66 & 0.09 &  0.30 & 1.03 &  0.67 & 0.11 &  0.22 &  1.11 & 0.98 & 0.98\\
14 & 0.65 & 0.10 &  0.27 & 1.03 &  0.70 & 0.11 &  0.28 &  1.11 & 0.75 & 0.76\\
15 & 0.64 & 0.11 &  0.22 & 1.05 &  0.71 & 0.10 &  0.31 &  1.18 & 0.58 & 0.60\\
\bottomrule
\end{tabular}}
{\footnotesize
\begin{tablenotes}
\item See Table \ref{tab:RZMult} notes.
\end{tablenotes}
}
\end{threeparttable}
\end{table}
\FloatBarrier
\begin{figure}[ht!]
\begin{center}
\caption{Effective F-Statistic}\label{fig:KPtest}
\includegraphics[width=1\textwidth]{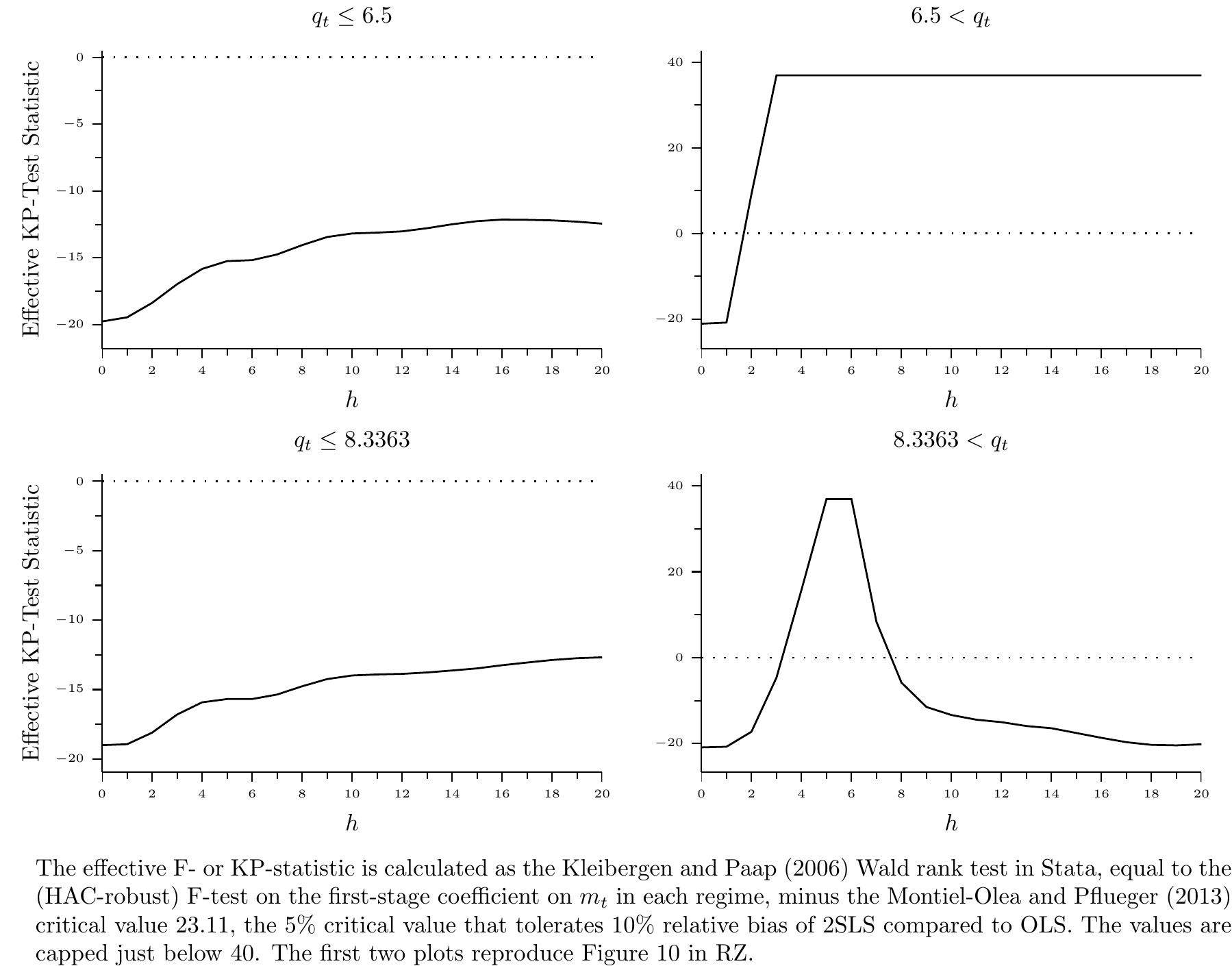}
\end{center}
\end{figure}

\begin{figure}[ht!]
\begin{center}
\caption{Effective F-Statistic}\label{fig:KPtestAdd}
\includegraphics[width=1\textwidth]{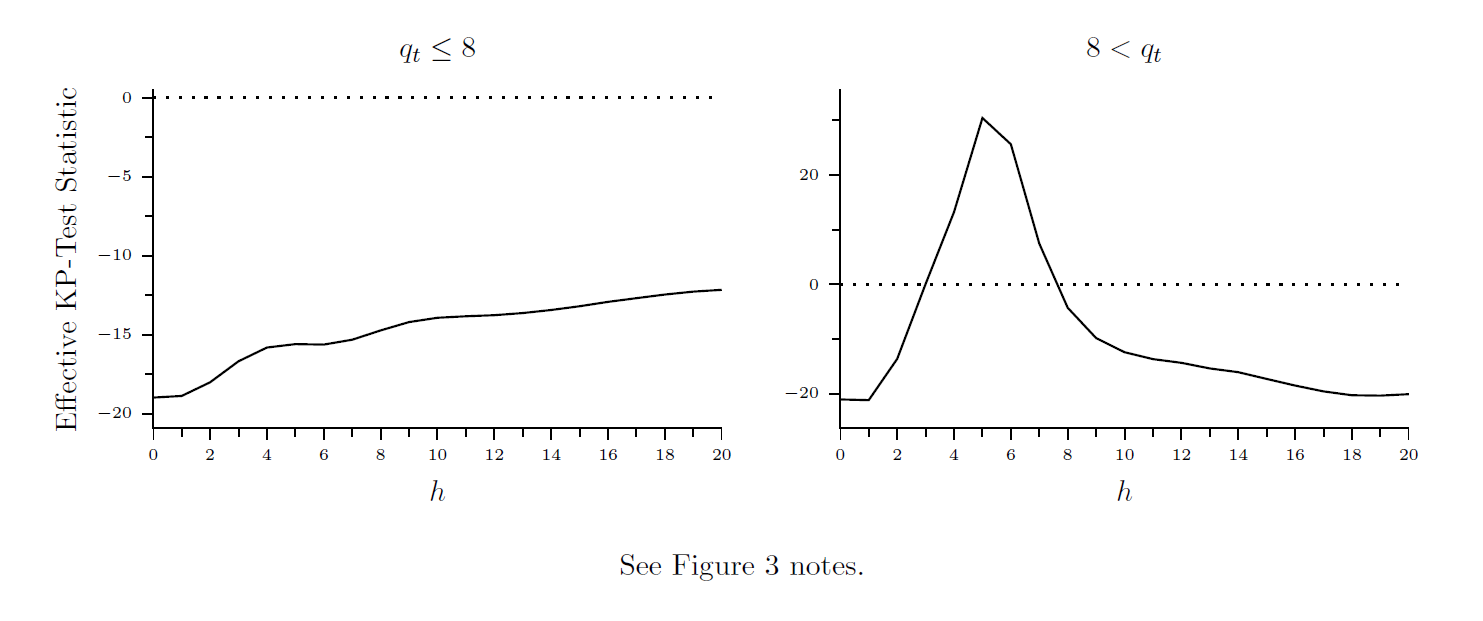}
\end{center}
\end{figure}

\begin{figure}[ht!]
\begin{center}
\caption{Data plots with with shaded areas $q_t>8.3363$}\label{fig:DataVsRegimes}
\includegraphics[width=1\textwidth]{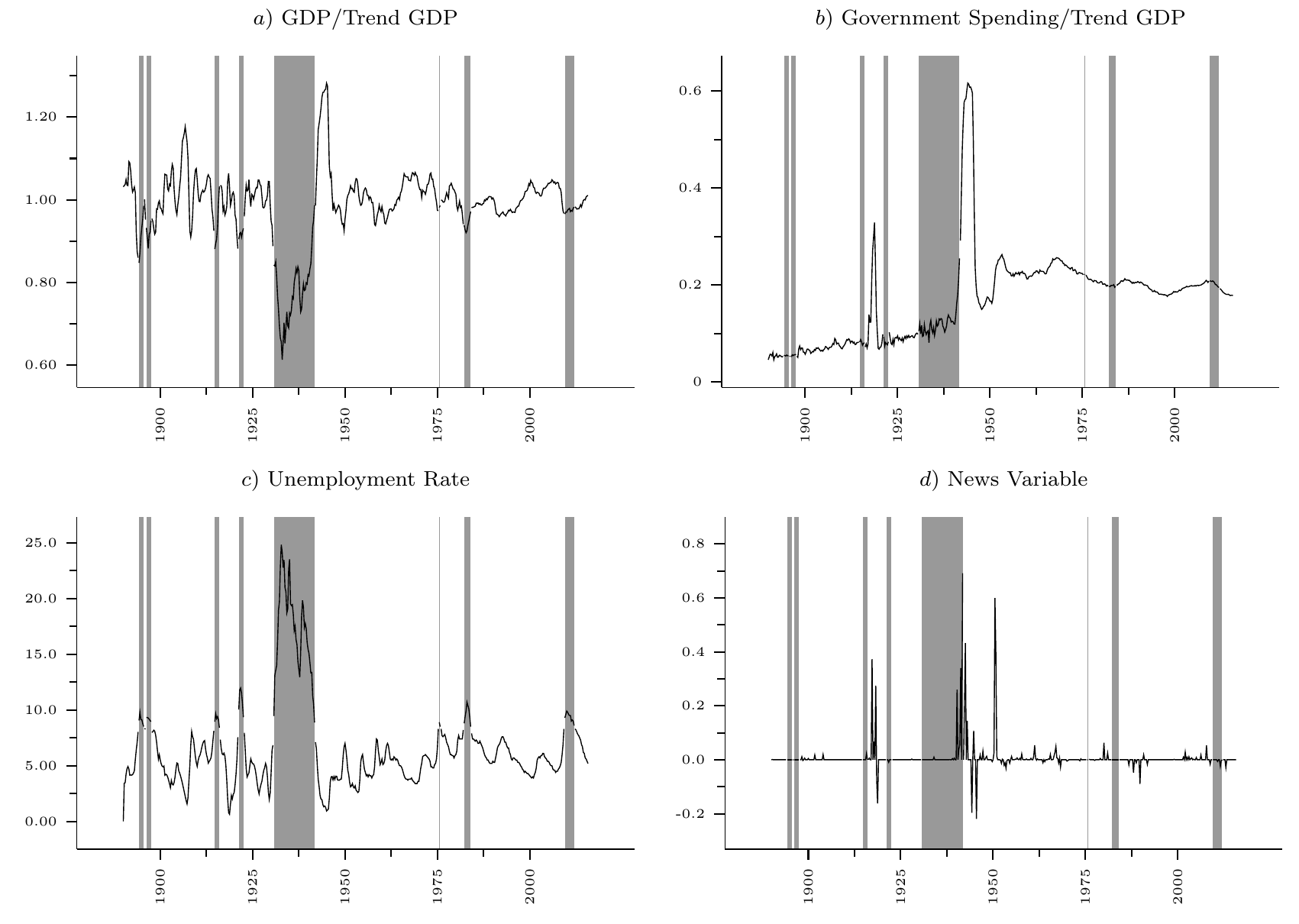}
\end{center}
\end{figure}
\FloatBarrier

\indent We therefore assess the possibility of weak instruments at various horizons in Figures \ref{fig:KPtest} and \ref{fig:KPtestAdd}, plotting the effective F-statistic  for the null hypothesis of weak instruments in each regime across horizons. These figures show evidence of weak instruments in both regimes at short horizons, for all thresholds.  This also holds for the effective F-statistics minus their critical value for our TFS specification with $\hat \rho=4.0636$: they are equal to approximately  $-19$ for $q_t \leq 4.0636$ (101 observations), and $-17.5$ for $q_t> 4.0636$ (399 observations), so well below zero. Therefore, the weak instrument robust p-values should be used, even for Table \ref{tab:UsMult}, at shorter horizons.  Hence, once  weak instruments are accounted for, there is no evidence that government spending multipliers are different in recessions, both in our paper and in RZ.\\
\indent What we do learn from the analysis is that military spending news becomes a weaker instrument for longer horizons when the threshold increases from $6.5$ to $8.0$ or to $8.3363$, and therefore that the instrument relevance is not robust to the threshold used. This is also indicated in Figure \ref{fig:DataVsRegimes}, which shows that, except for the World War II period, the news variable does not exhibit much variation when the unemployment rate is above $8.3363$. This suggests that the RZ military news instrument is more informative for intermediate values of unemployment, so for "normal" recessions rather than "deep" recessions.

\section{Conclusions}
In this paper we proposed two adjustments to the GMM Wald test of \cite{canerhansen2004}, and two new 2SLS test statistics for threshold detection in linear models with endogenous regressors and exogenous thresholds. We derived the asymptotic validity of their null bootstrap equivalents, and showed through simulations and an application that these tests have better finite sample properties than the test proposed in \cite{canerhansen2004}. \\
\indent \cite{rothfelderboldea2016} show in their Theorem 1 that under conditional homoskedasticity and one endogenous regressor, the 2SLS estimators with a linear first stage or a threshold first stage can be more efficient than the GMM estimators that ignore this information. It would be interesting to assess when this efficiency carries over to more general settings, and whether there exists an optimal GMM estimator that uses similar information from the first stage as the 2SLS estimators.

\bibliography{testing_bib}
\bibliographystyle{ecta}
\newpage
\appendix

\section{Online Supplement: Proofs}
\setcounter{theorem}{0}
\renewcommand{\theequation}{A.\arabic{equation}}
\renewcommand{\thetheorem}{A.\arabic{theorem}}
\renewcommand{\thelemma}{A.\arabic{lemma}}
\subsection{General Notation and Results}
Before proving our results, we introduce some more required notation and provide additional Lemmas with helpful results.\\
\\
\textbf{General notation.} Let $\|\cdot\|$ be the Euclidean norm for vectors, respectively the Frobenius norm for matrices: $\|P\|=\sqrt{\operatorname{tr}(P^\top P)}$. Also, for a (matrix valued) random variable  $P$, let $\| P \|_{\alpha} = (E\|P\|^{\alpha})^{1/\alpha}$, for any $\alpha >0$. Let $I_m$ the $m\times m$-identity matrix, $0_{a\times b}$ a $a\times b$ vector of zeros (we use this notation only when the dimension is not obvious from the derivations, else we use $0$), and let $K$ denote a generic constant. \\
\indent Let $P_t$ be a matrix of random variables. We define $P_{1\gamma}=E[ P_t \indgam]$, $P =\lim_{\gamma\to\infty}P_{1\gamma}=E[P_t]$, and $P_{2\gamma} = P-P_{1\gamma}$. For example, since $M_{1\gamma}= E[ z_t z_t^\top \indgam]$, we have $M = E[z_t z_t^\top]$ and $M_{2\gamma} = E[z_t z_t' \indinvgam]$. Let their sample equivalents (replacing expectations by averages and unobserved quantities with estimates)  be denoted by hats, for example, for $M_{1\gamma}$, its sample equivalent is $\hat M_{1\gamma}=T^{-1} \suma z_t z_t^\top$, for $M$ it is $\hat M = T^{-1} \sum_{t=1}^T z_t z_t^\top$, for $H_{\epsilon,1\gamma}= E[z_t z_t^\top \epsilon_t^2 \indgam]$, it is  $\hat H_{\epsilon, 1\gamma} = T^{-1} \suma z_t z_t^\top \hat \epsilon_t^2$, where $\hat \epsilon_t$ is an estimate of the residual $\epsilon_t$, and so on. When the notation $\hat P_{i\gamma}$ does not conform with this definition, it is specifically indicated in the text. \\
\indent Throughout the text, quantities of the form $v_tv_t^\top\otimes z_tz_t^\top$ should be read as $(v_tv_t^\top)\otimes (z_tz_t^\top)$. Let $\mathcal{G}_{1\gamma}=(\underset{1\times q}{\mathcal G_{\epsilon,1\gamma}^\top},\underset{1\times qp_1}{\mathcal G_{u,1\gamma}^\top})^\top$ be a $q(p_1+1)$ vector of zero mean Gaussian processes with covariance matrix
$$
H_{1\gamma} = E[(v_tv_t^\top\otimes z_tz_t^\top)\indgam]=\begin{pmatrix}
E[\epsilon_t^2z_tz_t^\top\indgam] & E[\epsilon_tu_t^\top\otimes z_tz_t^\top\indgam]\\
E[\epsilon_tu_t\otimes z_tz_t^\top\indgam] & E[u_tu_t^\top\otimes z_tz_t^\top\indgam]
\end{pmatrix}=\begin{pmatrix}
H_{\epsilon,1\gamma} & H_{\epsilon u,1\gamma}^\top\\
H_{\epsilon u,1\gamma} & H_{u,1\gamma}
\end{pmatrix}
$$
and covariance function $ E[\mathcal{G}_{1\gamma_1}\mathcal{G}_{1\gamma_2}^\top]=E[(v_tv_t^\top\otimes z_tz_t^\top)\mathbf{1}[q_t\leq(\gamma_1\wedge\gamma_2)].$  Hence, $\cG_{\epsilon,1\gamma}$, respectively $\cG_{u,1\gamma}$ corresponds to the parts of $\cG_{1\gamma}$ induced by $\epsilon_t$, respectively $u_t$. Moreover, $\mathcal G=\lim_{\gamma\to\infty}\mathcal G_{1\gamma}$ and $\mathcal G_{2\gamma} = \mathcal G-\mathcal G_{1\gamma}$.\\
\indent Let  $\sigma^2 =E[(\epsilon_t + u_t^\top \theta_x^0)^2]$, $\tilde \theta^0 =\vect(1,\theta_x^0)$ and $\check\theta^0 = \vect(0,\theta_x^0)$. Define $\tilde\epsilon_t=\epsilon_t+(x_t-\hat x_t)^\top\theta_x^0$, where $\hat x_t$ is obtained either with a LFS (linear first stage) or a TFS (threshold first stage) specification, depending on the context. Define $\hat C_{i\gamma} = T^{-1} \sum_{i\gamma} \hat w_t \hat w_t^\top$, and $\hat C = \hat C_{1\gamma} + \hat C_{2\gamma}$, both for a LFS and a TFS.\\
\indent Let $A^0=[\Pi^0,S^\top]^\top$ be the augmented matrix of the LFS slope parameters, where $S=[I_{p_2},0_{p_2\times (q-p_2)}]$. Hence, $z_{1t} = Sz_t$ and $w_t = A^0 z_t+(u_t^\top,0_{1\times q_1})^\top=A^0 z_t+\bar u_t$. Similarly, let $A_1^0=[\Pi_1^0,S^\top]^\top$ and $A_2^0=[\Pi_2^0,S^\top]^\top$ be the augmented matrix of TFS parameters such that $w_t=(A_1^0\indrhoo+A_2^0\indinvrhoo)z_t+\bar u_t$.\\
\indent All convergence results, if not stated otherwise, are uniformly in $\gamma \in \Gamma$, and all $\op(1)$ terms are uniform in $\gamma$. $``{\ind}"$ stands for weak convergence in Skorokhod metric, $``\indb"$ for weak convergence in Skorokhod metric under the bootstrap measure, and $``\inpb"$ for weak convergence in probability under the bootstrap measure.
 

\begin{lemma}\label{lem:1}[ULLN] 
 If (i) $\{a_t\}$ and $\{q_t\}$ are scalar  strictly stationary and $\rho$-mixing series, with mixing coefficient $\rho(m)=\mathcal O(m^{-A})$ for some $A>\frac{a}{a-1}$ and $1<a\leq r$;(ii) $\|a_t\|_r<\infty$ for some $r>1$; (iii) $q_t$ has a continuous distribution, with pdf $f(\cdot)$ bounded: $\sup_{x \in \Gamma} |f(x)|<\infty$, then $\sup_{\gamma \in \Gamma}\left| T^{-1} \sum_{1\gamma} a_t - E[a_t \indgam] \right| \xrightarrow{p} 0$.
\end{lemma}

\begin{proof}[\textbf{Proof of Lemma \ref{lem:1}}.] This uniform law of large numbers (ULLN) can be proven using the same steps as the proof of Lemma 1 in \cite{hansen1996}, with a slight modification as we do not assume that $a_t$ has a continuous and bounded pdf. First, note that $\rho$-mixing implies ergodicity. Second, set in the proof of their Lemma 1 $w_t = (a_t, q_t)$,  $\phi(w_t) =a_t$, and $\{w_t \leq \gamma\}=\indgam$. Follow the steps in \cite{hansen1996}, until their equation (15). Then note that $||a_t||_{r}=K<\infty$ by Assumption \ref{assn:1}(c), and for some $\epsilon>0$, set $K_\epsilon = (2K/\epsilon)^{r/(r-1)}$. Since $q_t$ is assumed to have a continuous and bounded pdf, there exists an $\epsilon$ such that $\int_{\gamma_k}^{\gamma_{k+1}} f(x) dx \leq 1/K_{\epsilon} = (\epsilon/(2K))^{r/(1-r)}$.  Therefore, replace equation (15) by the inequality below (derived using H\"older's inequality with $p=r$ and $q=r/(r-1)$): 
\begin{align*}
&E[|a_t| \1[\gamma_k < q_t \leq \gamma_{k+1}]] \leq \|a_t\|_{r} \| \1[\gamma_k < q_t \leq \gamma_{k+1}]]\|_{r/(r-1)} \\
= & K \left(\int_{\gamma_k}^{\gamma_{k+1}} f(x) dx \right)^{(r-1)/r}\leq K (\epsilon/(2K)) = \epsilon/2.
\end{align*}
The rest of the proof is as in \cite{hansen1996}, where only the last equation in their proof should be replaced by:
\begin{align*}
E|f_{\epsilon,k}^u(w_t) - f_{\epsilon,k}^l(w_t)| \leq 2 E[|a_t| \1[\gamma_k < q_t \leq \gamma_{k+1}]] \leq \epsilon.
\end{align*}

\end{proof}

\begin{lemma}\label{lem:2}[FCLT] 
If the assumptions in Lemma \ref{lem:1} hold but with $a_t$ being a vector of m.d.s, and additionally (i) $\|a_t\|_{2r}<\infty$ for some $r>1$; (ii) $E[a_t a_t' \indgam] = F_\gamma$, a p.d. matrix of constants, (iii) $\inf_{\gamma \in \Gamma} \det  F_\gamma>0$, then: $T^{-1/2} \sum_{1\gamma} a_t \Rightarrow \mathcal J_{1\gamma},$ a vector of Gaussian processes with covariance function $ E[a_t a_t^\top \mathbf 1[q_t \leq (\gamma_1 \wedge\gamma_2)]]$. 
\end{lemma} 
\begin{proof}[\textbf{Proof of Lemma \ref{lem:2}}.] This functional central limit theorem (FCLT) follows directly from Theorem 3 and then Theorem 1 in \cite{hansen1996}. Note that only $\|a_t\|_{2r}<\infty$ is needed, as evident from replacing $x_t \epsilon_t$ with $a_t$ in the first two equations of the proof of Theorem 3 in \cite{hansen1996}.
 \end{proof}
Note that, Lemmas \ref{lem:1} and \ref{lem:2} imply that $\sup_{\gamma\in\Gamma}|T^{-1}\sum_{2\gamma}a_t-E[a_t\indinvgam]|\inp0$, and $T^{-1/2}\sum_{2\gamma}a_t\Rightarrow \cJ_{2\gamma}$ since $\sum_{2\gamma}(\cdot)=\sum_{t=1}^T(\cdot)-\sum_{1\gamma}(\cdot)$ by definition.
\begin{lemma}\label{lem:3}
 Suppose Assumption \ref{assn:1} holds. Then: $(i)$  $T^{-1}\sum_{1\gamma}z_t z_t^\top \inp M_{1\gamma}$ and $T^{-1}\sum_{2\gamma}z_t z_t^\top \inp M_{2\gamma}$, and  $(ii)$ $T^{-1/2}\sum_{i\gamma} v_t \otimes z_t \Rightarrow \mathcal G_{i\gamma}$. \\
If, additionally, Assumption \ref{assn:2} holds, then, under the bootstrap measure: $(iii)$ $T^{-1} \sum_{i\gamma} z_t z_t^\top \eta_t = \opb(1)$ and $(iv)$ $T^{-1/2} \sum_{i\gamma} v_t \eta_t\,  \otimes z_t \indb \mathcal G_{i\gamma}$.
\end{lemma}
\begin{proof}[\textbf{Proof of Lemma \ref{lem:3}}]
Part $(i)$ follows from Assumptions \ref{assn:1} $(b)$, $(c)$ and $(e)$, ensuring that the assumptions of Lemma \ref{lem:1} are satisfied for elements of $z_t z_t^\top$. Hence, $ T^{-1}\sum_{i\gamma} z_t z_t^\top \inp M_{i\gamma}$. \\
\indent Part $(ii)$ follows from Assumptions \ref{assn:1} $(a)$--$(f)$, ensuring that the conditions of Lemma \ref{lem:2} are satisfied since $\|v_t \otimes z_t\|_{2r} \leq \|v_t \|_{4r}  \|z_t \|_{4r}<\infty$. So, $ T^{-1/2}\suma v_t \otimes z_t \ind \mathcal{G}_{1\gamma}$. \\
\indent We will show Parts $(iii)$ and $(iv)$ only for $i=1$; for $i=2$, the proofs follow the same arguments by noting that, for example, $\sum_{2\gamma}z_tz_t^\top\eta_t=\sum_{t=1}^Tz_tz_t^\top\eta_t-\sum_{1\gamma}z_tz_t^\top\eta_t$. To prove Part $(iii)$, note that $E^b[T^{-1}\sum_{1\gamma}z_tz_t^\top\eta_t]=0$. Hence, Chebyshev's and Minkowski's inequalities imply $P^b(\|T^{-1}\sum_{1\gamma}z_tz_t^\top\eta_t\|>K)\leq K^{-2}T^{-1}(T^{-1}\sum_{1\gamma}\|z_t\|^4+(E^b|\eta_t|)^2T^{-1}\sum_{1\gamma,t\neq s}\|z_t\|^2\|z_s\|^2)=K^{-2}T^{-1}(O_p(1)+O(1)O_p(1))=o_p(1)$, where the second to last equality follows from Assumptions \ref{assn:1} $(b)$, $(c)$, $(e)$, \ref{assn:2} $(a)$ and Lemma \ref{lem:1}. So, $T^{-1}\sum_{1\gamma}  z_t z_t^\top \eta_t=\op^b(1)$.\\
\indent Finally, to prove Part $(iv)$ for $i=1$, we apply Lemma \ref{lem:2} and verify that $Var^b(T^{-1/2} \sum_{1\gamma} v_t \eta_t \otimes z_t)\inp H_{1\gamma}$. First, $E^b(v_t \eta_t\,  \otimes z_t)= 0$. Conditions $(ii)$ and $(iii)$ in Lemma \ref{lem:2} are satisfied by Assumption \ref{assn:1} $(b)$, $(c)$, $(e)$ and Assumption \ref{assn:2} $(a)$. Condition $(i)$ is satisfied because $(E^b\| v_t \eta_t \otimes z_t \|^{r})^{1/(r)} = (E^b|\eta_t|^{r})^{1/(r)} \|v_t\|\|z_t||=\Op(1)$ by Assumptions \ref{assn:1} $(c)$ and \ref{assn:2}. Finally, $Var^b( T^{-1/2} \suml_{1\gamma} v_t \eta_t\,  \otimes z_t) = T^{-1} \suml_{1\gamma} v_t v_t^\top \otimes z_t z_t^\top \inp H_{1\gamma},$ where the last statement follows by applying Lemma \ref{lem:1}. So, by Lemma \ref{lem:2}, $T^{-1/2} \sum_{1\gamma} v_t \eta_t\,  \otimes z_t \indb \mathcal G_{1\gamma}$.
\end{proof}
\subsection{Proofs of GMM Results}
\subsubsection{Asymptotic Distribution of GMM}
In order to simplify exposition for these proofs, define the quantity
$$
\tilde A=\begin{cases}
A^0 & \text{if LFS}\\
A_1^0\indrhoo+A_2^0\indinvrhoo & \text{if TFS}
\end{cases}.
$$
Hence, $w_t=\tilde Az_t+\bar u_t$.
\begin{proof}[\textbf{Proof of Theorem \ref{theo:GMM.1}}]
\indent First, we show that $\hat N_{i\gamma}\xrightarrow{p}N_{i\gamma}$. Note, $\hat N_{i\gamma} = T^{-1}\sum_{i\gamma}w_tz_t^\top=\tilde A\Big(T^{-1}\sum_{i\gamma}z_tz_t^\top\Big)+T^{-1}\sum_{i\gamma}\bar u_tz_t^\top=\tilde AM_{i\gamma}+o_p(1)=N_{i\gamma}$, where the second to last equality follows from Lemma \ref{lem:3} $(i)$--$(ii)$ and the last equality from Assumption \ref{assn:1} $(a)$ implying that $N_{1\gamma}=\mathbb E[w_tz_t^\top\1\{q_t\leq\gamma\}]=\mathbb E[(A^0z_t+\bar u_t)z_t^\top\1\{q_t\leq\gamma\}]=A^0\mathbb E[z_tz_t^\top\1\{q_t\leq\gamma\}]=A^0M_{1\gamma}$ and, similarly, $N_{2\gamma}=A^0M_{2\gamma}$ in case of a LFS. In case of a TFS, some more algebra provides the appropriate result.\\
\indent Next, we show that $\hat H_{\epsilon,i\gamma}\xrightarrow{p}H_{\epsilon,i\gamma}$. Note that $\hat H_{\epsilon,i\gamma}=T^{-1}\sum_{i\gamma}z_tz_t^\top\hat\epsilon_t^2$ where $\hat\epsilon_t=y_t-w_t^\top\hat\theta_{(1)}$. Hence, $\hat H_{\epsilon,i\gamma}=T^{-1}\sum_{i\gamma}z_tz_t^\top\epsilon_t^2-2[T^{-1}\sum_{i\gamma}z_tz_t^\top\otimes\epsilon_tw_t^\top][I_q\otimes(\hat\theta_{(1)}-\theta^0)]+[I_q\otimes(\hat\theta_{(1)}-\theta^0)^\top][T^{-1}\sum_{i\gamma}z_tz_t^\top\otimes w_tw_t^\top][I_q\otimes (\hat\theta_{(1)}-\theta^0)]$. By Assumptions 4.1 $(b)$, $(c)$, $(e)$ and Lemma \ref{lem:1} it holds that $T^{-1}\sum_{i\gamma}z_tz_t^\top\epsilon_t^2=H_{\epsilon,i\gamma}+o_p(1)$. Moreover, for any $r>1$, $\|z_tz_t^\top\otimes w_tw_t^\top\|_r\leq\|I_q\otimes \tilde A\|^2\cdot\|z_tz_t^\top\otimes z_tz_t^\top\|_r+2\|I_q\otimes\tilde A\|\cdot\|z_tz_t^\top\otimes z_t\bar u_t^\top\|_r+\|z_tz_t^\top\otimes\bar u_t\bar u_t^\top\|_r$ by Minkowski's inequality and sub-multiplicativity of the Frobenius norm. Using H\"older's inequality and Assumption 4.1 $(c)$ it follows that $\|z_tz_t^\top\otimes z_tz_t^\top\|_r\leq\|z_t\|_{4r}^4<\infty$, $\|z_tz_t^\top\otimes z_t\bar u_t^\top\|_r\leq\|z_t\|_{4r}^3\|\bar u_t\|_{4r}<\infty$ and $\|z_tz_t^\top\otimes\bar u_t\bar u_t^\top\|_r \leq\|z_t\|_{4r}^2\|\bar u_t\|_{4r}^2<\infty.$ Additionally, $\|I_q\otimes\tilde A\|=O(1)$. Hence, $T^{-1}\sum_{i\gamma}z_tz_t^\top\otimes w_tw_t^\top=O_p(1)$ by Lemma \ref{lem:2}. By standard arguments for GMM estimators of linear models $\hat\theta_{(1)}-\theta^0=O_p(T^{-1/2})$. Therefore, $[I_q\otimes(\hat\theta_{(1)}-\theta)^\top][T^{-1}\sum_{i\gamma}z_tz_t^\top\otimes w_tw_t^\top][I_q\otimes (\hat\theta_{(1)}-\theta^0)]=O_p(T^{-1/2})O_p(1)O_p(T^{-1/2})=o_p(1)$. Similarly, $[T^{-1}\sum_{i\gamma}z_tz_t^\top\otimes\epsilon_tw_t^\top][I_q\otimes(\hat\theta_{(1)}-\theta^0)]=o_p(1)$. Hence, $\hat H_{\epsilon,i\gamma}=H_{\epsilon,i\gamma}+o_p(1)$.\\
\indent Last, we show that $T^{1/2}(\hat\theta_{i\gamma}-\theta^0)\Rightarrow\mathcal (N_{i\gamma}H_{\epsilon,i\gamma}^{-1}N_{i\gamma}^\top)^{-1}N_{i\gamma}H_{\epsilon,i\gamma}^{-1}\mathcal G_{\epsilon,i\gamma}$. Since $\hat\theta_{i\gamma}=\theta^0+(\hat N_{i\gamma}\hat H_{\epsilon,i\gamma}^{-1}\hat N_{i\gamma}^\top)^{-1}$ $(\hat N_{i\gamma}\hat H_{\epsilon,i\gamma}^{-1}[T^{-1}\sum_{i\gamma}z_t\epsilon_t])$, $T^{-1/2}\sum_{i\gamma}z_t\epsilon_t\Rightarrow\mathcal G_{\epsilon,i\gamma}$ by Lemma \ref{lem:3} $(ii)$, parts $(i)$ and $(ii)$ above, the continuous mapping theorem and Slutksy's theorem it follows that $T^{1/2}(\hat\theta_{i\gamma}-\theta^0)\Rightarrow(N_{i\gamma}H_{\epsilon,i\gamma}^{-1}N_{i\gamma}^\top)^{-1}N_{i\gamma}H_{\epsilon,i\gamma}^{-1}\mathcal G_{\epsilon,i\gamma}$.\\
\indent Using parts $(i)$--$(iii)$ and continuous mapping theorem and Slutsky's theorem concludes the proof.
\end{proof}
\subsubsection{Bootstrap Validity for GMM Tests}
\begin{lemma}\label{lem:GMM.1}
Suppose Assumptions \ref{assn:1} and \ref{assn:2} hold. Then, under $\mathbb H_0$ and for $i=1,2$, $(i)$ $T^{-1/2}\sum_{i\gamma}z_ty_t^b\xRightarrow{d^b_p}\cG_{\epsilon,i\gamma}$, $(ii)$ $T^{1/2}\hat\theta_{(1)}^b=O_p^b(1)$, and $(iii)$ $\hat H_{\epsilon,i\gamma}^b\xrightarrow{p^b}H_{\epsilon,i\gamma}$ 
\end{lemma}
\begin{proof}[\textbf{Proof of Lemma \ref{lem:GMM.1}}]
\indent \textit{Part $(i)$}. Note  $T^{-1/2}\sum_{i\gamma}z_ty_t^b=T^{-1/2}\sum_{i\gamma}z_t\eta_t\epsilon_t-\Big(T^{-1}\sum_{i\gamma}z_tw_t^\top\eta_t\Big)\Big(T^{1/2}[\hat\theta_{(2)}-\theta^0]\Big)$. From Lemma \ref{lem:3} $(iii)$, respectively standard arguments for GMM estimators for linear models it follows that $T^{-1/2}\sum_{i\gamma}z_t\eta_t\epsilon_t\xRightarrow{d^b_p}\cG_{\epsilon,i\gamma}$, respectively $T^{1/2}[\hat\theta_{(2)}-\theta^0]=O_p(1)$. Moreover, $T^{-1}\sum_{i\gamma}z_tw_t^\top\eta_t =\Big(T^{-1}\sum_{i\gamma}z_tz_t^\top\eta_t\Big)\tilde A^{\top}+T^{-1}\sum_{i\gamma}z_t\bar u_t^\top\eta_t=o_p^b(1)O(1)+o_p^b(1)=o_p^b(1)$ by Lemma \ref{lem:3} $(iii)$--$(iv)$ and since $\tilde A=O(1)$. Thus, $T^{-1/2}\sum_{i\gamma}^Tz_ty_t^b\xRightarrow{d^b_p}\mathcal G_{\epsilon,i\gamma}$ under the null hypothesis.\\
\indent \textit{Part $(ii)$}. Recall that $\hat\theta_{(1)}=(\hat N\hat M^{-1}\hat N^\top)^{-1}(\hat N\hat M^{-1}\sum_{t=1}^Tz_ty_t^b)$ where $\hat M\xrightarrow{p}M$ by Lemma \ref{lem:3}, $\hat N\xrightarrow{p}N$ as shown in the proof of Theorem \ref{theo:GMM.1} and $T^{-1/2}\sum_{t=1}^Tz_ty_t^b\xRightarrow{d_p^b}\cG_\epsilon$ by part $(i)$. Hence, $T^{1/2}\hat\theta_{(1)}^b=O_p^b(1)$.\\
\indent \textit{Part (iii)}. Note that $\hat\epsilon_t^b=\epsilon_t\eta_t-w_t^\top(\hat\theta_{(2)}-\theta^0)\eta_t-w_t^\top\hat\theta_{(1)}^b$. Hence,
\begin{eqnarray}
\hat H_{\epsilon,i\gamma}^b & = & T^{-1}\sum_{i\gamma}z_tz_t^\top\Big(\epsilon_t^2\eta_t^2+w_t^\top(\hat\theta_{(2)}-\theta^0)(\hat\theta_{(2)}-\theta^0)^\top w_t\eta_t^2+w_t^\top\hat\theta_{(1)}^b\hat\theta_{(1)}^{b\top}w_t\nonumber\\
& & \qquad\qquad-2\epsilon_t\eta_t^2w_t^\top(\hat\theta_{(2)}-\theta^0)-2\epsilon_t\eta_tw_t^\top\hat\theta_{(1)}^b+2\eta_tw_t^\top(\hat\theta_{(2)}-\theta^0)\hat\theta_{(1)}^{b\top}w_t\Big)\nonumber\\
& = & I+II+III-2IV-2V+2VI.\label{equ:PL4.1}
\end{eqnarray}
We proceed by first showing that $II,...,VI=o_p^b(1)$ and then that $I=H_{\epsilon,i\gamma}+o_p^b(1)$.\\
\indent Note that $II=\Big(I_q\otimes(\hat\theta_{(2)}-\theta^0)^\top\Big)\Big[T^{-1}\sum_{i\gamma}z_tz_t^\top\otimes w_tw_t^\top\eta_t^2\Big]\Big(I_q\otimes(\hat\theta_{(2)}-\theta^0)\Big)$. Using Markov's inequality, the term in brackets satisfies $P^b(\|T^{-1}\sum_{i\gamma}z_tz_t^\top\otimes w_tw_t^\top\eta_t^2\|\geq K)\leq K^{-1}T^{-1}\sum_{i\gamma}^T\|z_tz_t^\top\otimes w_tw_t^\top\|=O_p(1)$ where the last equality was shown in the proof of Theorem \ref{theo:GMM.1}. Hence, $T^{-\alpha}[T^{-1}\sum_{i\gamma}^Tz_tz_t^\top\otimes w_tw_t^\top\eta_t^2]=o_p^b(1)$ for any $\alpha>0$. By standard arguments for GMM-estimators in linear models, $\hat\theta_{(2)}-\theta^0=O_p(T^{-1/2})$ under Assumption 4.1. Hence, $II=O_p(T^{-1/2})o_p^b(T^\alpha)O_p(T^{-1/2})=T^{\alpha-1} O_p(1)o_p^b(1)O_p(1)=o_p^b(1)$ for $0<\alpha\leq1$.\\
\indent Moreover, $III=\Big(I_q\otimes\hat\theta_{(1)}^{b\top}\Big)\Big[T^{-1}\sum_{i\gamma}z_tz_t^\top\otimes w_tw_t^\top\Big]\Big(I_q\otimes\hat\theta_{(1)}^b\Big).$ By part $(ii)$, $T^{1/2}\hat\theta_{(1)}^b=O_p^b(1)$. As shown above, $T^{-1}\sum_{i\gamma}z_tz_t^\top\otimes w_tw_t^\top=O_p(1)$. Hence, $III=O_p^b(T^{-1/2})O_p(1)O_p^b(T^{-1/2})=o_p^b(1)$. \\
\indent By the same arguments as for terms $II$ and $III$ above, $IV=O_p(T^{-1/2})o_p^b(T^\alpha)=o_p^b(1)$ and $V=O_p^b(T^{-1/2})o_p^b(T^\alpha)=o_p^b(1)$ for any $0<\alpha\leq\tfrac{1}{2}$, respectively $VI=O_p(T^{-1/2})o_p^b(T^\alpha)O_p^b(T^{-1/2})=o_p^b(1)$ for any $0<\alpha \leq1$. Hence, $\hat H_{\epsilon,i\gamma}^b=T^{-1}\sum_{i\gamma}z_tz_t^\top\epsilon_t^2\eta_t^2+o_p^b(1)$, for any $0<\alpha\leq\tfrac{1}{2}$.\\
\indent Thus, it is left to show that $T^{-1}\sum_{i\gamma}z_tz_t^\top\epsilon_t^2\eta_t^2-T^{-1}\sum_{i\gamma}z_tz_t^\top\epsilon_t^2=o_p^b(1)$ since $T^{-1}\sum_{i\gamma}z_tz_t^\top\epsilon_t^2=H_{\epsilon,i\gamma}+o_p(1)$ by Assumptions 4.1 $(b)$, $(c)$, $(e)$ and Lemma \ref{lem:1}. Since $E^b[T^{-1}\sum_{i\gamma}z_tz_t^\top\epsilon_t^2\eta_t^2]=T^{-1}\sum_{i\gamma}z_tz_t^\top\epsilon_t^2$, we verify the claim by Chebychev's inequality: $P^b(\|T^{-1}\sum_{i\gamma}z_tz_t^\top\epsilon_t^2(\eta_t^2-1)\|>K)\leq K^{-2}E^b\|T^{-1}\sum_{i\gamma}z_tz_t^\top\epsilon_t^2(\eta_t^2-1)\|^2\leq K^{-2}T^{-1}[T^{-1}\sum_{i\gamma,t\neq s}\|z_tz_t^\top\epsilon_t^2z_sz_s^\top\epsilon_s^2\|E^b|(\eta_t^2-1)(\eta_s^2-1)|]$. By Assumption \ref{assn:2} and applying H\"older's inequality with $p=q=2r$ iteratively, $E^b|(\eta_t^2-1)(\eta_s^2-1)|<\infty$ and $\|z_tz_t^\top z_sz_s^\top\epsilon_t^2\epsilon_s^2\|_r\leq\|z_t\|_{8r}^4\|\epsilon_t\|_{8r}^4<\infty$ for some $r>1$. Hence, $P^b(\|T^{-1}\sum_{i\gamma}z_tz_t^\top\epsilon_t^2(\eta_t^2-1)\|>K)\leq T^{-1}O_p(1)=o_p(1)$ by Lemma \ref{lem:1}, concluding the proof.
\end{proof}

\begin{proof}[\textbf{Proof of Theorem \ref{theo:GMM.2}}]
First, by Lemma \ref{lem:GMM.1} $(iii)$, the first result in the proof of Theorem \ref{theo:GMM.1}, Slutky's theorem and the continuous mapping theorem it follows that $\hat N_{i\gamma}\hat H_{\epsilon,i\gamma}^{b^{-1}}\hat N_{i\gamma}^\top\xrightarrow{p^b}N_{i\gamma}H_{\epsilon,i\gamma}^{-1}N_{i\gamma}^\top\equiv V_{i\gamma}$ for $i=1,2$. Next, recall that $\hat\theta_{i\gamma,(2)}^b=(\hat N_{i\gamma}\hat H_{\epsilon,i\gamma}^{b^{-1}}\hat N_{i\gamma}^\top)^{-1}(\hat N_{i\gamma}\hat H_{\epsilon,i\gamma}^{b^{-1}}T^{-1}\sum_{i\gamma}z_ty_t^b)$ for $i=1,2$. Hence, $T^{1/2}(\hat\theta_{1\gamma,(2)}^b-\hat\theta_{2\gamma,(2)}^b)\xRightarrow{d_p^b}V_{1\gamma}^{-1}N_{1\gamma}H_{\epsilon,1\gamma}\cG_{\epsilon,1\gamma}-V_{2\gamma}^{-1}N_{2\gamma}H_{\epsilon,2\gamma}\cG_{\epsilon,2\gamma}$, by Lemma \ref{lem:GMM.1}, Slutsky's theorem and the continuous mapping theorem. Putting these results together yields the claim.
\end{proof}
\subsection{Proofs of 2SLS Results with a LFS}
\subsubsection{Asymptotic Distribution of 2SLS Test Statistics with a LFS}\label{sec:ProofsLFSAsyms}
\begin{definition}\label{def:QuantsLFS}
Let $C_{i\gamma}=A^0M_{i\gamma}A^{0\top}$, $R_{i\gamma} =  M_{i\gamma}M^{-1}$, $C_\gamma = [C_{1\gamma}^{-1}, -C_{2\gamma}^{-1}]$, and $Q_\gamma =C_{1\gamma} C^{-1} C_{2\gamma}$.  Also, define the Gaussian processes $\mathcal B_{i\gamma}=A^0 [(\tilde \theta^{0\top} \otimes I_q)\mathcal G_{i\gamma} - R_{i\gamma}(\check \theta^{0\top} \otimes I_q)\mathcal G]$ for $i=1,2$, and $\mathcal B= \mathcal B_{1\gamma} + \mathcal B_{2\gamma}$, as well as the processes $\mathcal B_\gamma = \vect(\mathcal B_{1\gamma}, \mathcal B_{2\gamma})$ and $\mathcal E_\gamma = C_\gamma \mathcal B_\gamma$. Let $V_{\mathcal B_{i\gamma}}=A^0 D_{i\gamma} H_{i\gamma} D_{i\gamma}^{\top} A^0$ for $i=1,2$, and $V_{\mathcal B}=A^0 D H D^\top A^0$ and $V_{\mathcal B, 12,\gamma} = A^0 D_{1\gamma} H_{1\gamma}D^\top A^{0\top}$, $V_{\mathcal B,2\gamma}$ denote the covariances of $\mathcal B_{i\gamma}$, $\mathcal B$, and the covariance between $\mathcal B_{1\gamma}$ and $\mathcal B_{2\gamma}$. Then, $V_{\mathcal B_\gamma}  = \Bigg(\begin{array}{cc} {\scriptstyle V_{\mathcal B_{1\gamma}}} & {\scriptstyle V_{\mathcal B_{12 \gamma}}} \\ {\scriptstyle V_{\mathcal B_{12\gamma}} ^\top} & {\scriptstyle V_{\mathcal B_{2\gamma}}}\end{array}\Bigg),$ and $V_{\gamma} = C_{\gamma}  V_{\mathcal B_\gamma} C_\gamma^\top,$ respectively. Define $\bar D_{i\gamma} = [\tilde \theta^{\top} \otimes I_q] - [\check \theta^{\top} \otimes \hat R_{i\gamma}]$, where $\tilde \theta = [1,\hat \theta_x]$ and $\check \theta = [0,\hat \theta_x]$. Then $\hat V_\gamma$ is defined as $V_{\gamma}$, but replacing $C_{i\gamma}$ with $\hat C_{i\gamma}$, $A^0$ with $\hat A$,  $D_{i\gamma}$ by $\bar D_{i\gamma}$, and $H_{1\gamma}$ by $\hat H_{1\gamma} = T^{-1} \suml_{\gamma} \hat v_t \hat v_t^\top \otimes z_t z_t^\top$, and $H$ by $\hat H = T^{-1} \suml_{t=1}^T \hat v_t \hat v_t^\top \otimes z_t z_t^\top$.
\end{definition}
\begin{lemma}\label{lem:4}
Suppose Assumption \ref{assn:1} holds, $y_t$ is generated by \eqref{equ:TMTS.1} and $x_t$ is generated by the LFS \eqref{equ:TMTS.2}. Then, under $\mathbb H_0$ and for $i=1,2$, $(i)$ $T^{-1}\sum_{i\gamma}\hat w_t \hat w_t^\top \inp C_{i\gamma}$; $(ii)$ $T^{-1/2}\sum_{i\gamma} \hat w_t \tilde \epsilon_t \Rightarrow \mathcal B_{i\gamma}$, and \\$T^{-1/2}\vect(\sum_{1\gamma} \hat w_t \tilde \epsilon_t,\sum_{2\gamma} \hat w_t \tilde \epsilon_t)  \Rightarrow \mathcal B_\gamma.$
\end{lemma}
\begin{proof}[\textbf{Proof of Lemma \ref{lem:4}.}] \textit{Part (i).}  $T^{-1}\suml_{i\gamma}\hat w_t \hat w_t^\top = T^{-1} \suml_{i\gamma} \hat A z_t z_t^\top \hat A^\top$. By Assumption \ref{assn:1} and standard arguments, we have that $\hat A = A^0+ \op(1)$, so $T^{-1}\sum_{i\gamma}\hat w_t \hat w_t^\top = (A^0+\op(1)) T^{-1}\sum_{i\gamma} z_t z_t^\top (A^0+\op(1))^{\top}$.  By Lemma \ref{lem:3}, $ T^{-1}\sum_{i\gamma} z_t z_t^\top \inp M_{i\gamma}$, for $i=1,2$. Hence, $T^{-1}\sum_{i\gamma}\hat w_t \hat w_t^\top \inp A^0 M_{i\gamma} A^{0\top} = C_{i\gamma}$.\\
\indent\textit{Part (ii).} Note, 
\begin{align*}
T^{-1/2}\sum_{i\gamma} \hat w_t \tilde \epsilon_t & =  \hat A\,\big( T^{-1/2}\sum_{i\gamma} z_t (\epsilon_t+ u_t^\top \theta_x^0) -\hat R_{i\gamma}  T^{-1/2}\sum_{t=1}^T z_t u_t^\top \theta_x^0\big)\\
& =\hat A \Big([\tilde \theta^{0\top} \otimes I_q] T^{-1/2} \sum_{i\gamma} v_t \otimes z_t-\hat R_{i\gamma}[\check \theta^{0\top} \otimes I_q] T^{-1/2} \sum_{t=1}^Tv_t \otimes z_t\Big).
\end{align*}
By Lemma \ref{lem:3}, $\hat R_{i\gamma}\inp R_{i\gamma}$, $ T^{-1/2}\sum_{i\gamma} v_t \otimes z_t \ind \mathcal{G}_{i\gamma}$ and $ T^{-1/2}\sum_{t=1}^Tv_t \otimes z_t \ind \mathcal{G}$. Hence, $T^{-1/2}\suml_{i\gamma} \hat w_t \tilde \epsilon_t \ind \mathcal B_{i\gamma}$ and  $T^{-1/2}\vect(\sum_{1\gamma} \hat w_t \tilde \epsilon_t,\sum_{2\gamma} \hat w_t \tilde \epsilon_t)  \Rightarrow \mathcal B_\gamma.$
\end{proof}
\begin{theorem}[Asymptotic Distribution LFS]\label{thm:LFSAsyms}
Let $y_t$ be generated by \eqref{equ:TMTS.1} and $x_t$ be generated by the LFS \eqref{equ:TMTS.2}. Then, under $\mathbb H_0$ and Assumption \ref{assn:1},
\begin{itemize}
\item[$(i)$] $T^{1/2}(\hat \theta_{1\gamma}-\hat \theta_{2\gamma}) \ind \mathcal E_\gamma$,
\item[$(ii)$] $\sup_{\gamma\in\Gamma} LR_{T}(\gamma)\Rightarrow\sup_{\gamma\in\Gamma}[ \mathcal E^\top_\gamma\ Q_\gamma \ \mathcal E_\gamma/\sigma^2],$
\item[$(iii)$] $\sup_{\gamma\in\Gamma}W_{T}(\gamma) \Rightarrow\sup_{\gamma\in\Gamma}\mathcal E^\top_\gamma V_\gamma^{-1}\mathcal E_\gamma$, where $V_\gamma$ was defined in Definition \ref{def:QuantsLFS}.
\end{itemize}
\end{theorem}
\begin{proof}[\textbf{Proof of Theorem \ref{thm:LFSAsyms}}]
\textit{Part (i).} $T^{1/2}(\hat \theta_{i\gamma} -\theta^0) = \hat C_{i\gamma}^{-1}( T^{-1/2} \sum_{i \gamma}\hat w_t \tilde \epsilon_t) \ind C_{i\gamma}^{-1} \mathcal B_{i\gamma}$ by Lemma \ref{lem:4}. So, $
T^{1/2}(\hat \theta_{1\gamma} - \hat \theta_{2\gamma}) \ind [C_{1\gamma}^{-1},-C_{2\gamma}^{-1}] \vect(\mathcal B_{1\gamma}, \mathcal B_{2\gamma}) = C_{\gamma} \, \mathcal B_\gamma = \mathcal E_\gamma.$\\
\indent \textit{Part (ii)}. Since $\hat \theta = \hat C^{-1} (T^{-1} \sum_{t=1}^T \hat w_t y_t)=\hat C^{-1} (\sum_{i=1}^2T^{-1} \sum_{i\gamma} \hat w_t y_t)$ and $\hat\theta_{i\gamma}=\hat C_{i\gamma}^{-1}(T^{-1}\sum_{i\gamma}\hat w_ty_y)$, it follows that $\hat \theta = \sum_{i=1}^2\hat C ^{-1} \hat C_{i\gamma} \hat \theta_{i\gamma}$. So $\hat \theta = \sum_{i=1}^2 C^{-1} C_{i\gamma} \hat \theta_{i\gamma}+\op(1)$, and therefore $\hat \theta_{1\gamma}-\hat \theta = C^{-1} C_{2\gamma} (\hat \theta_{1\gamma} - \hat \theta_{2\gamma})+\op(1)$ and $\hat \theta_{2\gamma} - \hat \theta =C^{-1} C_{1\gamma}(\hat \theta_{2\gamma} - \hat \theta_{1\gamma})+\op(1)$ . Hence,
 \begin{eqnarray*}
SSR_0-SSR_1(\gamma)&=&\sum_{i=1}^2  \left[\sum_{i\gamma}(y_t - \hat w_t^\top \hat \theta)^2 -(y_t - \hat w_t^\top \hat \theta_{i\gamma})^2 \right] \\
&=& \sum_{i=1}^2 (\hat \theta_{i\gamma}-\hat \theta)^\top \left[2 \sum_{i\gamma} \hat w_t \tilde \epsilon_t - \sum_{i\gamma} \hat w_t \hat w_t^\top (\hat \theta-\theta^0)-\sum_{i\gamma} \hat w_t \hat w_t^\top (\hat \theta_{i\gamma}-\theta^0)  \right] \\
 &=& \sum_{i=1}^2 T^{1/2} (\hat \theta_{i\gamma}-\hat \theta)^\top \left( T^{-1} \sum_{i\gamma} \hat w_t \hat w_t^\top \right) T^{1/2} (\hat \theta_{i\gamma}-\hat \theta) \\
&=& T^{1/2} (\hat \theta_{1\gamma} - \hat \theta_{2\gamma})  [ C_{2\gamma} C^{-1} C_{1\gamma} C^{-1} C_{2\gamma} + C_{1\gamma} C^{-1} C_{2\gamma} C^{-1} C_{1\gamma}]  T^{1/2}(\hat \theta_{1\gamma} - \hat \theta_{2\gamma})+\op(1)\\
&=&   T^{1/2}(\hat \theta_{1\gamma} - \hat \theta_{2\gamma})^\top Q_\gamma \, T^{1/2} (\hat \theta_{1\gamma} - \hat \theta_{2\gamma})+\op(1),
\end{eqnarray*}
where the last line follows because $C = \sum_{i=1}^2 C_{i\gamma}$, therefore $C^{-1}C_{1\gamma} = I_{p} - C^{-1}C_{2\gamma} $, $C_{1\gamma} C^{-1} C_{2\gamma} = (C- C_{2\gamma})C^{-1} (C- C_{1\gamma}) = C - C_{1\gamma} - C_{2\gamma} + C_{2\gamma} C^{-1} C_{1\gamma} = C_{2\gamma} C^{-1} C_{1\gamma}$, so $C_{2\gamma} C^{-1} C_{1\gamma} C^{-1} C_{2\gamma} + C_{1\gamma} C^{-1} C_{2\gamma} C^{-1} C_{1\gamma} = C_{2\gamma} C^{-1} C_{1\gamma} C^{-1} C_{2\gamma}  + C_{1\gamma} C^{-1} C_{2\gamma} (I_p -  C^{-1} C_{2\gamma})  = Q_\gamma +  (C_{2\gamma} C^{-1} C_{1\gamma}- C_{1\gamma} C^{-1} C_{2\gamma}) C^{-1} C_{2\gamma} = Q_\gamma.$ Since $T^{-1/2}(\hat \theta_{1\gamma} - \hat \theta_{2\gamma}) \ind \mathcal E_\gamma$, $SSR_0-SSR_1(\gamma)\ind \mathcal E^\top_\gamma Q_\gamma\, \mathcal E_\gamma$.\\
\indent Next, $SSR_1(\gamma)/(T-2p) = \sigma^2+\op(1)$, since, as shown below,  $T^{-1} SSR_1(\gamma)\inp \sigma^2$:
\begin{eqnarray*}
T^{-1} SSR_1(\gamma) & = & \sum_{i=1}^2 T^{-1}\sum_{i\gamma} (y_t - \hat w_t^\top \hat \theta_{i\gamma})^2 \\
&= & \sum_{i=1}^2 T^{-1}\sum_{i\gamma} (\tilde \epsilon_t -  \hat w_t^\top (\hat \theta_{i\gamma}-\theta^0))^2 \\
&= & T^{-1} \sum_{t=1}^T \tilde \epsilon_t^2 - 2 \sum_{i=1}^2 T^{-1}\sum_{i\gamma} \tilde \epsilon_t \hat w_t^\top (\hat \theta_{i\gamma}-\theta^0) + \sum_{i=1}^2 (\hat \theta_{i\gamma}-\theta^0)^\top (T^{-1} \sum_{i\gamma} \hat w_t \hat w_t^\top) (\hat \theta_{i\gamma}-\theta^0).
\end{eqnarray*}
By Lemma \ref{lem:4}, $T^{-1} \sum_{i\gamma} \hat w_t \hat w_t^\top \inp C_{i\gamma}$, and $ T^{-1}\suml_{i\gamma} \tilde \epsilon_t \hat w_t^\top =\op(1)$ and, as shown in \textit{Part (ii)} above, $\hat \theta_{i\gamma} -\theta^0=\op(1)$. Hence:
\begin{eqnarray} 
T^{-1} SSR_1(\gamma) & = & T^{-1} \sum_{t=1}^T (\epsilon_t + u_t^\top \theta_x^0 - z_t^\top (\hat \Pi-\Pi^0) \theta_x^0)^2 +\op(1)\nonumber\\ 
& = & T^{-1} \sum_{t=1}^T  (\epsilon_t + u_t^\top \theta_x^0)^2 + \theta_x^{0\top}(\hat \Pi-\Pi^0) T^{-1} \sum_{t=1}^T z_t z_t^\top (\hat \Pi-\Pi^0) \theta_x^0  \nonumber\\ 
& & -2 T^{-1} \sum_{t=1}^T z_t (\epsilon_t+ u_t^\top \theta_x^0)(\hat \Pi-\Pi^0) \theta_x^0+\op(1) \nonumber\\ 
&=&  T^{-1} \sum_{t=1}^T  (\epsilon_t + u_t^\top \theta_x^0)^2+\op(1),\label{equ:PT4.1}
 \end{eqnarray}
where the last equality used Lemma \ref{lem:3} $(ii)$ and the fact that $\hat \Pi-\Pi^0 = \op(1)$. We now apply Lemma \ref{lem:1} to $(\epsilon_t + u_t^\top \theta_x^0)^2$. First, $E[(\epsilon_t + u_t^\top \theta_x^0)^2]=\sigma^2$. Second, by Assumption \ref{assn:1} $(b)$, $(\epsilon_t + u_t^\top \theta_x^0)^2$ is strictly stationary with $\rho$-mixing coefficients satisfying condition $(i)$ in Lemma \ref{lem:1}. Third, by Minkowski's inequality, $\|(\epsilon_t + u_t^\top \theta_x^0)^2\|_{2} \leq \|\epsilon_t^2\|_{2} + \|(\theta_x^{0\top} u_t)^2\|_{2} + 2\| \epsilon_t u_t^\top \theta_x^0\|_{2} $. Note that $\|\epsilon_t^2\|_2 = \|\epsilon_t\|_4^2<K$, $\|(\theta_x^{0\top} u_t)^2\|_{2}\leq\|\theta_z^0\|^2\|u_t\|_4^2<K$ and $ \| \epsilon_t u_t^\top \theta_x^0\|_{2}\leq\|\theta_x^0\|\|\epsilon_t\|_4\|u_t\|_4^4<K$ by Assumption \ref{assn:1} $(c)$. Therefore, by Lemma \ref{lem:1},  $T^{-1} \suml_{t=1}^T  (\epsilon_t + u_t^\top \theta_x^0)^2\inp \sigma^2$, completing the proof of Part $(ii)$.\\
\indent \textit{Part (iii).}  We are left to show $\hat V_\gamma \inp V_\gamma$. Since $\hat C_{i\gamma} \inp C_{i\gamma}$, $\hat A \inp A^0$, $\hat \theta_x \inp \theta_x^0$, to show that $\hat V_\gamma \inp V_\gamma$, it suffices to show that $\hat H_{i\gamma} = T^{-1} \sum_{i\gamma} \hat v_t\hat v_t^\top \otimes z_tz_t^\top \inp H_{i\gamma}$ for $i=1,2$. We proceed to show that $\hat H_{\epsilon,i\gamma} = T^{-1} \suml_{i\gamma} z_t z_t^\top \hat \epsilon_t^2 \inp H_{\epsilon,i\gamma}$:
\begin{eqnarray}
\hat H_{\epsilon,i\gamma} &=& T^{-1} \sum_{1\gamma} z_t z_t^\top \hat \epsilon_t^2 = T^{-1} \sum_{1\gamma} z_t z_t^\top [\epsilon_t + w_t^\top (\theta^0 -\hat \theta)]^2 \label{equ:PT4.2}\\
&= &  T^{-1} \sum_{1\gamma} z_t z_t^\top \epsilon_t^2+T^{-1} \sum_{1\gamma} z_t z_t^\top [(\theta^0 -\hat \theta)^\top w_t w_t^\top (\theta^0 -\hat \theta)] \nonumber \\
&+&2 T^{-1} \sum_{1\gamma} z_t z_t^\top [ (\theta^0 -\hat \theta)^\top w_t \epsilon_t]\label{equ:PT4.3}. 
\end{eqnarray}
First, by Assumption \ref{assn:1} $(b)$, $(c)$ and $(e)$, and Lemma \ref{lem:1}, $T^{-1} \suml_{i\gamma} z_t z_t^\top \epsilon_t^2 \inp H_{\epsilon,i\gamma}$. Second, $ T^{-1} \suml_{i\gamma} z_t z_t^\top [(\theta^0 -\hat \theta)^\top w_t w_t^\top (\theta^0 -\hat \theta)]  =  [I_q \otimes (\theta^0 -\hat \theta)]^\top [ T^{-1} \suml_{i\gamma} z_t z_t^\top \otimes w_t w_t^\top ][I_q \otimes (\theta^0 -\hat \theta)]$. Note that $\theta^0 -\hat \theta=\Op(T^{-1/2})$. Moreover, the asymptotic behavior of the terms $II$ and $III$ in \eqref{equ:PL4.1}, implies that $T^{-1}\sum_{i\gamma}z_tz_t^\top\otimes w_tw_t^\top=O_p(1)$. Hence, $T^{-1} \suml_{i\gamma} z_t z_t^\top [(\theta^0 -\hat \theta)^\top w_t w_t^\top (\theta^0 -\hat \theta)] =\op(1)$. Similarly, $ T^{-1} \suml_{i\gamma} z_t z_t^\top  (\theta^0 -\hat \theta)^\top w_t \epsilon_t =\op(1)$. Therefore, $\hat H_{\epsilon,i\gamma} \inp H_{\epsilon,1\gamma}$. \\
\indent By similar arguments, $\hat H_{u,i\gamma} \inp H_{u,i\gamma}$ and $\hat H_{\epsilon u,i\gamma} \inp H_{\epsilon u,i\gamma}$, completing the proof of Part $(iii)$.
\end{proof}
\subsubsection{Bootstrap Validity for 2SLS Test Statistics with LFS}
\begin{lemma}\label{lem:5}
Let Assumptions \ref{assn:1}-\ref{assn:2} hold, $y_t$ be generated by \eqref{equ:TMTS.1}, and $x_t$ be generated by \eqref{equ:TMTS.2}. Then, under $\mathbb H_0$ and for $i=1,2$, $(i)$ $T^{-1/2} \sum_{i\gamma} v_t^b\,  \otimes z_t \indb \mathcal G_{i\gamma}$; $(ii)$ $T^{1/2} (\hat \Pi^b - \hat \Pi) = T^{1/2} (\hat \Pi - \Pi^0)+\opb(1)$; $(iii)$ $T^{-1} \sum_{i\gamma} \hat w_t^b \hat w_t^{b\top} \inpb C_{i\gamma}$; $(iv)$ $T^{-1/2} \sum_{i\gamma} \hat w_t^b \tilde \epsilon_t^b \indb \mathcal B_{i\gamma}$ and $\vect(T^{-1/2} \sum_{1\gamma} \hat w_t^b \tilde \epsilon_t^b,T^{-1/2} \sum_{2\gamma} \hat w_t^b \tilde \epsilon_t^b) \indb \mathcal B_\gamma$, where $\tilde \epsilon_t^b = y_t^b - \hat w_t^{b\top} \hat \theta$.
\end{lemma}
\begin{proof}[\textbf{Proof of Lemma \ref{lem:5}}] 
\textit{Part (i).} Since $\hat u_t = x_t-\hat x_t = u_t - (\hat \Pi-\Pi^0)^\top z_t$, 
\begin{eqnarray*}
T^{-1/2} \sum_{i\gamma} u_t^b\,  \otimes z_t  & = & T^{-1/2} \sum_{i\gamma} \hat u_t \eta_t \otimes z_t \\ 
&&=   T^{-1/2} \sum_{i\gamma} u_t \eta_t \otimes z_t  -  (T^{1/2}(\hat \Pi-\Pi^0)^\top \otimes I_q)\, \big(T^{-1} \sum_{i\gamma} z_t \eta_t \otimes z_t\big)  \\
&&=    T^{-1/2} \sum_{i\gamma} u_t \eta_t \otimes z_t + \Op(1) \times \vect(T^{-1}\sum_{i\gamma} z_t z_t^\top \eta_t) \\
&& = T^{-1/2} \sum_{i\gamma} u_t \eta_t \otimes z_t  + \opb(1).
\end{eqnarray*}
where the last equality follows from Lemma \ref{lem:3} $(iii)$. Similarly, $T^{-1/2} \suml_{i\gamma} \epsilon_t^b z_t =  T^{-1/2} \suml_{i\gamma} \epsilon_t \eta_t z_t +\op^b(1)$, therefore $T^{-1/2} \suml_{i\gamma} v_t^b\,  \otimes z_t \indb \mathcal G_{i\gamma}$ by Lemma \ref{lem:3} $(iv)$.\\
\indent \textit{Part (ii).} Since $\vect(z_tu_t^{b\top})=u_t^b\otimes z_t$, from \textit{Part (i)}, and Lemma \ref{lem:3} $(ii)$ and $(iv)$, $T^{1/2} (\hat \Pi^b - \hat \Pi)  = \hat M^{-1} (T^{-1/2} \sum_{t=1}^T z_t u_t^{b\top}) = \hat M^{-1} (T^{-1/2} \sum_{t=1}^T z_t u_t^{\top} +\opb(1)) =  T^{1/2}(\hat \Pi-\Pi^0)+\opb(1).$\\
\indent \textit{Part (iii).} $T^{-1} \sum_{i\gamma} \hat w_t^b \hat w_t^b = \hat A^{b} (T^{-1} \sum_{i\gamma} z_t z_t^\top) \hat A^{b\top}$, where $\hat A^b = [\hat \Pi^b,S^\top]^\top$. By \textit{Part (ii)}, $T^{1/2}(\hat \Pi^b - \Pi^0)=\Op^b(1)$, so $\hat \Pi^b \inpb \Pi^0$, therefore $\hat A^b \inpb A^0$, and  $T^{-1} \sum_{i\gamma} \hat w_t^b \hat w_t^{b\top} \inpb C_{i\gamma}$.\\
\indent \textit{Part (iv).} By \textit{Parts} $(i)$ and $(ii)$ and Lemma \ref{lem:4} $(i)$-$(ii)$, and recalling that $[\tilde \theta^\top \otimes I_q] - [\check \theta^\top \otimes \hat R_{i\gamma}] \inp [\tilde \theta^{0\top}\otimes I_q] - [\check \theta^{0\top} \otimes R_{i\gamma}]$, we have:
\begin{align*}
T^{-1/2} \sum_{i\gamma} \hat w_t^b \tilde \epsilon_t^b& = \hat A^b \left( T^{-1/2} \sum_{i\gamma} z_t (\epsilon_t^b + (w_t^b - \hat w_t^b)^\top \hat \theta)\right) \\
&= \hat A^b  \left( T^{-1/2} \sum_{i\gamma} z_t (\epsilon_t^b + u_t^{b\top}\hat\theta_x) - [T^{-1} \sum_{i\gamma} z_t z_t^\top] T^{1/2}(\hat \Pi^b-\hat \Pi) \hat \theta_x\right) \\
&= \hat A^b  \left( T^{-1/2} \sum_{i\gamma} z_t (\epsilon_t^b + u_t^{b\top}\hat\theta_x) - \hat M_{i\gamma}\hat M^{-1} T^{-1/2}\sum_{t=1}^Tz_tu_t^\top\hat \theta_x\right) +o_p^b(1)\\
& = A^0 \Big([\tilde \theta^\top \otimes I_q] T^{-1/2}\sum_{i\gamma}v_t^b\otimes z_t - \hat R_{i\gamma} [\check \theta^\top \otimes I_q] \Big( T^{-1/2} \sum_{t=1}^T v_t^b \otimes z_t \Big)\Big)  + \opb(1) \\
& = A^0 \Big([\tilde \theta^{0\top} \otimes I_q] T^{-1/2}\sum_{i\gamma}v_t^b\otimes z_t - R_{i\gamma} [\check \theta^{0\top} \otimes I_q] \Big( T^{-1/2} \sum_{t=1}^T v_t^b \otimes z_t \Big)\Big)  + \opb(1) \\
&  \indb \mathcal B_{1\gamma}.
\end{align*}
Hence, also $\vect(T^{-1/2} \sum_{1\gamma} \hat w_t^b \tilde \epsilon_t^b,T^{-1/2} \sum_{2\gamma} \hat w_t^b \tilde \epsilon_t^b) \indb \mathcal B_\gamma$.
\end{proof}

\begin{proof}[\textbf{Proof of Theorem \ref{thm:1}.}] \textit{Part (i).} We have 
\begin{eqnarray}\nonumber
SSR_0^b-SSR_1^b(\gamma) & = & \sum_{i=1}^2  \left[   \sum_{i\gamma}(y_t^b - \hat w_t^{b\top} \hat \theta^b)^2 -(y_t^b - \hat w_t^{b\top} \hat \theta_{i\gamma}^b)^2 \right] \\ \nonumber
&=& \sum_{i=1}^2 (\hat \theta_{i\gamma}^b-\hat \theta^b)^\top \left[2 \suml_{i\gamma} \hat w_t^b \tilde \epsilon_t^b - \sum_{i\gamma} \hat w_t^b \hat w_t^{b\top } (\hat \theta^b-\hat \theta)-\sum_{i\gamma} \hat w_t^b \hat w_t^{b\top} (\hat \theta_{i\gamma}^b-\hat \theta)  \right] 
\end{eqnarray}
\begin{eqnarray} \nonumber
 &=& \sum_{i=1}^2 T^{1/2} (\hat \theta_{i\gamma}^b-\hat \theta^b)^\top \left(  T^{-1} \sum_{i\gamma} \hat w_t^b \hat w_t^{b\top} \right)T^{1/2} (\hat \theta_{i\gamma}^b-\hat \theta^b) \\  \nonumber
&=& \sum_{i=1}^2 T^{1/2} (\hat \theta_{i\gamma}^b-\hat \theta^b)^\top C_{i\gamma}  T^{1/2} (\hat \theta_{i\gamma}^b-\hat \theta^b)+\op^b(1)\\ 
&=& T^{1/2}(\hat \theta_{1\gamma}^b - \hat \theta_{2\gamma}^b)^\top Q_\gamma  T^{1/2} (\hat \theta_{1\gamma}^b - \hat \theta_{2\gamma}^b)+\op^b(1), \label{equ:PT3.1}
\end{eqnarray}
where the second to last line follows by Lemma \ref{lem:5} $(iii)$. By Lemma \ref{lem:5} $(iii)$ and $(iv)$, 
 \begin{eqnarray}\nonumber
 T^{1/2} (\hat \theta_{1\gamma}^b - \hat \theta_{2\gamma}^b) & = & T^{1/2} (\hat \theta_{1\gamma}^b - \hat \theta) -T^{1/2} (\hat \theta_{2\gamma}^b - \hat \theta) \\ \nonumber
&= & (T^{-1} \sum_{1\gamma} \hat w_t^b \hat w_t^{b\top})^{-1} T^{-1/2} \sum_{1\gamma} \hat w_t^b \tilde \epsilon_t^{b\top} -(T^{-1} \sum_{2\gamma} \hat w_t^b \hat w_t^{b\top})^{-1} T^{-1/2} \sum_{2\gamma} \hat w_t^b \tilde \epsilon_t^{b\top}
\\ 
& \indb & C_{1\gamma}^{-1} \mathcal B_{1\gamma} -   C_{2\gamma}^{-1} \mathcal B_{2\gamma} = \mathcal E_\gamma.\label{equ:PT3.2}
\end{eqnarray}
Using \eqref{equ:PT3.2} in \eqref{equ:PT3.1}, we have: $SSR_0^b-SSR_1^b(\gamma) \indb \mathcal E^\top_\gamma Q_\gamma \mathcal E_\gamma.$\\
\indent It remains to show that $SSR_1^b(\gamma)/(T-2p) \inpb \sigma^2$, or, equivalently, that $T^{-1} SSR_1^b(\gamma) \inpb \sigma^2$:
\begin{eqnarray*}
T^{-1} SSR_1^b(\gamma) & = & \sum_{i=1}^2 T^{-1}\sum_{i\gamma} (y_t^b - \hat w_t^{b\top} \hat \theta_{i\gamma}^b)^2\\
& = & \sum_{i=1}^2 T^{-1}\sum_{i\gamma} (\tilde \epsilon_t^b -  \hat w_t^{b\top} (\hat \theta_{i\gamma}^b-\hat \theta))^2 \\
& = & T^{-1} \sum_{t=1}^T (\tilde \epsilon_t^b)^2 - 2 \sum_{i=1}^2 T^{-1}\sum_{i\gamma} \tilde \epsilon_t^b \hat w_t^{b\top} (\hat \theta_{i\gamma}^b-\hat \theta) \\
&&+ \sum_{i=1}^2 (\hat \theta_{i\gamma}^b-\hat \theta)^\top \Big(T^{-1} \sum_{i\gamma} \hat w_t^b \hat w_t^{b\top}\Big) (\hat \theta_{i\gamma}^b-\hat \theta).
\end{eqnarray*}
By Lemma \ref{lem:5} $(iii)$ and $(iv)$, $T^{-1} \sum_{i\gamma} \hat w_t^b \hat w_t^{b\top} \inpb C_{i\gamma}$, and $ T^{-1}\suml_{i\gamma} \tilde \epsilon_t^b \hat w_t^{b\top} =\Op^b(T^{-1/2})$. From \eqref{equ:PT3.2}, $\hat \theta_{i\gamma}^b-\hat \theta=\Op^b(T^{-1/2}).$ Therefore, $T^{-1} SSR_1^b(\gamma) =  T^{-1} \sum_{t=1}^T (\tilde \epsilon_t^b)^2+\opb(1) = T^{-1} \suml_{t=1}^T (\epsilon_t^b + u_t^{b\top} \hat \theta_x - z_t^\top (\hat \Pi^b-\hat \Pi)^\top \hat \theta_x)^2 +\opb(1)= T^{-1} \suml_{t=1}^T  (\epsilon_t^b + u_t^{b\top} \theta_x^0)^2 +\opb(1)$, where the last equality used Lemma \ref{lem:5} $(ii)$, which implies $\hat \Pi^b-\hat \Pi=\op(1)$, and the fact that $\hat\theta_x\inp\theta_x^0$. We now show that $T^{-1} \suml_{t=1}^T  (\epsilon_t^b + u_t^{b\top} \theta_x^0)^2 \inpb \sigma^2$, which then completes the proof of \textit{Part (i)}.\\
\indent Since $T^{-1} \suml_{t=1}^T  (\epsilon_t^b + u_t^{b\top} \theta_x^0)^2 = \tilde \theta^{0\top}( T^{-1} \suml_{t=1}^T  \hat v_t \hat v_t^\top \eta_t^2 )\tilde \theta^0$, we analyze $T^{-1} \suml_{t=1}^T  \hat v_t \hat v_t^\top \eta_t^2$. First consider $T^{-1} \suml_{t=1}^T  \hat \epsilon_t^2 \eta_t^2$, the first element of this matrix. We have: 
\begin{eqnarray}
T^{-1} \sum_{t=1}^T  \hat \epsilon_t^2 \eta_t^2 & =& T^{-1} \sum_{t=1}^T (\epsilon_t - w_t^\top(\hat \theta - \theta^0))^2 \eta_t^2 \nonumber\\
& = & T^{-1} \sum_{t=1}^T \epsilon_t^2 \eta_t^2 + (\hat \theta - \theta^0)^\top A^0 \left(T^{-1}  \sum_{t=1}^T z_t z_t^\top \eta_t^2\right ) A^{0\top} (\hat \theta - \theta^0) \nonumber\\
 & & - 2 \left(T^{-1} \sum_{t=1}^T \epsilon_t z_t^\top \eta_t^2 \right) A^{0\top} (\hat \theta - \theta^0)\nonumber
\end{eqnarray}
\begin{eqnarray} 
& = & T^{-1} \sum_{t=1}^T \epsilon_t^2 \eta_t^2 + \op(1)\Big(T^{-1}  \sum_{t=1}^T z_t z_t^\top \eta_t^2\Big) \op(1) - \Big(T^{-1} \sum_{t=1}^T \epsilon_t z_t^\top \eta_t^2\Big) \op(1).\label{equ:PT3.3}
\end{eqnarray}
First, we show that $T^{-1} \suml_{t=1}^T \epsilon_t^2 \eta_t^2\inpb E[\epsilon_t^2]$. Note that $E^b[T^{-1} \suml_{t=1}^T \epsilon_t^2 (\eta_t^2-1)] =0$. Hence, we have  $\Var^b(T^{-1} \suml_{t=1}^T \epsilon_t^2 (\eta_t^2-1))=E^b[(T^{-1} \suml_{t=1}^T \epsilon_t^2 (\eta_t^2-1))^2]=E^b[(\eta_t^2-1)^2]T^{-2}\sum_{t=1}^T\epsilon_t^4=o_p(1),$ since $\epsilon_t^4$ satisfies the assumptions of Lemma \ref{lem:1} and $E^b[(\eta_t^2-1)^2]<\infty$ by Assumption \ref{assn:2} $(i)$. Thus, $P^b(|T^{-1} \suml_{t=1}^T \epsilon_t^2 (\eta_t^2-1)|>K)\leq K^{-2}E^b[(\eta_t^2-1)^2]T^{-2}\sum_{t=1}^T\epsilon_t^4=o_p(1)$. Hence, $T^{-1}\sum_{t=1}^T\epsilon_t^2\eta_t^2\xrightarrow{p^b}E[\epsilon_t^2]$.\\
\indent Second, we show that $T^{-1}\suml_{t=1}^T z_t z_t^\top \eta_t^2\xrightarrow{p^b}M$. Note $E^b[T^{-1}\suml_{t=1}^T z_t z_t^\top\eta_t^2]=T^{-1}\suml_{t=1}^T z_t z_t^\top\inp E[z_tz_t^\top]$ by Lemma \ref{lem:1}. Hence, $P^b(\|T^{-1}\suml_{t=1}^T z_t z_t^\top(\eta_t^2-1)\|>K)\leq K^{-2}T^{-1}(T^{-1}\sum_{t,s}\|z_t\|^2\|z_s\|^2E^b[|(\eta_t^2-1)(\eta_s^2-1)|]=o_p(1)$ where the last equality follows from Assumption \ref{assn:2} $(i)$, implying $E^b[|(\eta_t^2-1)(\eta_s^2-1)|]\leq \tilde K<\infty$ for some $\tilde K >0$ and for all $s,t$, and Lemma \ref{lem:1} applied to $\|z_t\|^2\|z_s\|^2$, which ensures $T^{-1}\sum_{t,s}\|z_t\|^2\|z_s\|^2=O_p(1)$. Thus, $T^{-1}\sum_{t=1}^Tz_tz_t^\top\eta_t^2\inpb E[z_tz_t^\top]$.\\
\indent Third, we show that $T^{-1} \sum_{t=1}^T \epsilon_t z_t \eta_t^2=\opb(1)$. Note that $E^b[T^{-1} \sum_{t=1}^T \epsilon_t z_t (\eta_t^2-1)]=0$. Hence, $P^b(\|T^{-1} \sum_{t=1}^T \epsilon_t z_t(\eta_t^2-1)\|>K)\leq K^{-2}T^{-1}(T^{-1}\sum_{t=1}^T\|z_t\|\|z_s\||\epsilon_t\epsilon_s|)=o_p^b(1)$, by the same arguments as before. Thus, $T^{-1}\suml_{t=1}^T \epsilon_t z_t\eta_t\inpb0$. Substituting these results into \eqref{equ:PT3.3},  $T^{-1} \suml_{t=1}^T  \hat \epsilon_t^2 \eta_t^2 \inpb E[\epsilon_t^2]$. Next, 
\begin{eqnarray*}
T^{-1} \sum_{t=1}^T  \hat u_t \hat u_t^\top \eta_t^2 &=& T^{-1} \sum_{t=1}^T (u_t + (\hat \Pi- \Pi^0)^\top z_t) (u_t + (\hat \Pi- \Pi^0)^\top z_t)^\top \eta_t^2 \\
&=& T^{-1} \sum_{t=1}^T u_t u_t^\top  \eta_t^2 +  (\hat \Pi- \Pi^0)^\top  T^{-1}  \sum_{t=1}^T z_t z_t^\top \eta_t^2 (\hat \Pi- \Pi^0)\\
& & + [T^{-1} \sum_{t=1}^T u_t z_t^\top \eta_t^2 (\hat \Pi- \Pi^0)]+ (\hat \Pi- \Pi^0)]^\top[T^{-1} \sum_{t=1}^T u_t z_t\eta_t^2 \\
& \inpb & E[u_t u_t^\top],  
\end{eqnarray*}
by similar arguments as for $T^{-1} \suml_{t=1}^T  \hat \epsilon_t^2 \eta_t^2 \inpb E[\epsilon_t^2]$. Similarly, $T^{-1} \suml_{t=1}^T  \hat u_t \hat \epsilon_t^\top \eta_t^2 \inpb E[u_t \epsilon_t]$. Therefore, $T^{-1} \suml_{t=1}^T  (\epsilon_t^b + u_t^{b\top} \theta_x^0)^2 = \tilde \theta^{0\top}( T^{-1} \suml_{t=1}^T  \hat v_t \hat v_t^\top \eta_t^2 )\tilde \theta^0 =\tilde \theta^{0\top} E[v_t v_t^\top] \tilde \theta^{0} +\op^b(1)=\sigma^2+\op^b(1)$, completing the proof of \textit{Part (i)}.\\
\textit{Part (ii).} From \eqref{equ:PT3.2}, $T^{1/2}(\hat \theta_{1\gamma}^b - \hat \theta_{2\gamma}^b) \indb \mathcal E_\gamma$, so it remains to show that $\hat V_{\gamma}^b \inpb V_\gamma$. We will only show that $ \hat V_{\mathcal B,i\gamma}^b = \hat A^b\{(\tilde\theta^{b\top}\otimes I_q)\hat H_{i\gamma}^b(\tilde\theta^{b}\otimes I_q)+(\check\theta^{b\top}\otimes I_q)\hat H^b(\check\theta^{b}\otimes I_q)-(\tilde\theta^{b\top}\otimes I_q)\hat H_{i\gamma}^b(\check\theta^{b}\otimes I_q)-(\check\theta^{b\top}\otimes I_q)\hat H_{i\gamma}^b(\tilde\theta^{b}\otimes I_q) \}\hat A^{b\top}\inp A^0 \{(\tilde\theta^{0\top}\otimes I_q) H_{i\gamma}(\tilde\theta^{0}\otimes I_q)+(\check\theta^{0\top}\otimes I_q)H(\check\theta^{0}\otimes I_q)-(\tilde\theta^{0\top}\otimes I_q)H_{i\gamma}(\check\theta^{0}\otimes I_q)-(\check\theta^{0\top}\otimes I_q)H_{i\gamma}(\tilde\theta^{0}\otimes I_q) \} A^{0\top}= V_{\mathcal B,i\gamma}$, where $\hat H_{i\gamma}^b = T^{-1} \sum_{i\gamma} \hat v_t^b \hat v_t^b \otimes z_t z_t^\top$; the rest follows by similar arguments. We have already shown that $\hat A^b \inpb A^0$, $\tilde\theta^b\inp \tilde\theta^0$, and $\check\theta^b\inp \check\theta^0$. Since the proof for $\hat H_{i\gamma}^b  \inpb H_{i\gamma}$ is similar to $\hat H_{\epsilon,i\gamma}^b  \inpb H_{\epsilon,i\gamma}$, where $\hat H_{\epsilon,i\gamma}^b = T^{-1}\suml_{i\gamma} (\hat \epsilon_t^b)^2 z_t z_t^\top$, we only show $\hat H_{\epsilon,i\gamma}^b  \inpb H_{\epsilon,i\gamma}$. As in Equation \eqref{equ:PT4.3} in the proof of Theorem \ref{thm:LFSAsyms}, replacing $\epsilon_t$, $w_t$ and $\theta^0$ by $\epsilon_t^b$, $w_t^b$ and $\hat\theta$:
\begin{eqnarray*}
\hat H_{\epsilon,i\gamma}^b & = & T^{-1} \sum_{i\gamma} z_t z_t^\top (\epsilon_t^b)^2+ T^{-1} \sum_{i\gamma} z_t z_t^\top [w_t^{b\top}(\hat \theta_{i\gamma}^b-\hat \theta)]^2 + 2 T^{-1} \sum_{i\gamma} z_t z_t^\top [\epsilon_t^b  w_t^{b\top}(\hat \theta_{i\gamma}^b-\hat \theta)] \\
& = &T^{-1} \sum_{i\gamma} z_t z_t^\top (\epsilon_t^b)^2 + [I_q \otimes ((\hat \theta_{i\gamma}^b-\hat \theta)^{\top} \hat A^b)]\, \left(T^{-1} \sum_{i\gamma} z_t z_t^\top  \otimes z_t z_t^{\top}\right)\, [I_q \otimes (\hat A^{b\top}(\hat \theta_{i\gamma}^b-\hat \theta))] \\
& &+ 2 [I_q \otimes ((\hat \theta_{i\gamma}^b-\hat \theta)^{\top} \hat A^b)] \, \left(T^{-1} \sum_{i\gamma} z_t z_t^\top \epsilon_t^b  \otimes z_t\right).
\end{eqnarray*}
We have already shown that $T^{-1}\sum_{t=1}^Tz_tz_t^\top\otimes z_tz_t^\top=O_p(1)$ (in the proof of Theorem \ref{theo:GMM.1}), $\hat \theta_{i\gamma}^b-\hat\theta=\Op^b(T^{-1/2})$ (from equation \eqref{equ:PT3.2}), and $\hat A^b =A^0+\op^b(1)$ (proof of Part $(ii)$ in Lemma \ref{lem:3}). Recall that $\epsilon_t^b=\hat\epsilon_t\eta_t$, so
\begin{eqnarray*}
\hat H_{\epsilon,i\gamma}^b &=& T^{-1} \sum_{i\gamma} z_t z_t^\top \hat \epsilon_t^2 \eta_t^2 +\Op^b(T^{-1/2})\left(T^{-1} \sum_{i\gamma} z_t z_t^\top \hat\epsilon_t\eta_t  \otimes z_t \right)+\op^b(1) 
\end{eqnarray*}
As before, 
\begin{eqnarray*}
T^{-1} \suml_{i\gamma} z_t z_t^\top \hat \epsilon_t^2 \eta_t^2 & = & T^{-1} \sum_{i\gamma} z_t z_t^\top \epsilon_t^2 \eta_t^2 \nonumber\\
&&+ [I_q \otimes ((\hat \theta - \theta^0)^\top A^0)]\, \left(T^{-1}  \sum_{i\gamma} z_t z_t^\top \otimes z_t z_t^\top \eta_t^2 \right)\, [I_q \otimes (A^{0\top} (\hat \theta - \theta^0))] \nonumber\\
&&- 2 [I_q \otimes ((\hat \theta - \theta^0)^\top A^0)] \,\left( T^{-1} \sum_{i\gamma} z_t z_t^\top \otimes \epsilon_t z_t \eta_t^2 \,\right) \nonumber\\
& \inpb & H_{\epsilon,i\gamma}
\end{eqnarray*}
where the convergence was already shown in the Proof of Theorem \ref{theo:GMM.2} (cf. the asymptotic behaviour of the terms $I$, $II$ and $IV$ in Equation \eqref{equ:PL4.1}). By similar arguments, 
\begin{eqnarray*}
T^{-1} \sum_{i\gamma} z_t z_t^\top \hat \epsilon_t \eta_t \otimes z_t & =&  T^{-1} \sum_{i\gamma} z_t z_t^\top \otimes z_t \epsilon_t \eta_t - (T^{-1} \sum_{i\gamma} z_t z_t^\top \otimes z_t z_t^\top \eta_t ) \, [I_q \otimes A^{0\top}(\hat \theta - \theta^0)] \\
& = & o_p^b(T^\alpha)+o_p^b(T^\alpha)O_p(T^{-1/2}).
\end{eqnarray*}
Hence, $O_p(T^{-1/2})\Big(T^{-1} \sum_{1\gamma} z_t z_t^\top \hat\epsilon_t\eta_t \otimes z_t\Big)=O_p^b(T^{-1/2})(o_p^b(T^\alpha)+o_p^b(T^\alpha)O_p(T^{-1/2}))=o_p^b(1)$ if $\alpha\leq\frac{1}{2}$. Hence, we conclude that $\hat H^b_{\epsilon,i\gamma}=\mathbb E[z_tz_t^\top\epsilon_t^2\indgam]+o_p^b(1)$ for any $\alpha\leq\frac{1}{2}$.
\end{proof}
\subsection{Proofs of 2SLS Results with a TFS}\label{sec:ProofsTFS}
\begin{theorem}\label{theo:0}
Under Assumption \ref{assn:1} and the TFS in \eqref{equ:TMTS.3}, (i) $ \hat\rho-\rho^0 \inp 0$; (ii) $T(\hat\rho-\rho^0)=\mathcal O_p(1)$; (iii) $T^{1/2}\vect(\hat\Pi_i-\Pi_i^0) = T^{1/2}\vect(\hat\Pi_{i\rho^0}-\Pi_i^0) +\op(1)$.
\end{theorem}
\begin{proof}[\textbf{Proof of Theorem \ref{theo:0}.}] 
\noindent Wlog, assume that $\hat\rho\leq\rho^0$ (the proofs for $\hat\rho>\rho^0$ are similar and omitted for simplicity).\\
\indent \textit{Part (i).} Let $ Q(\hat \Pi_{1\hat \rho},\hat \Pi_{2\hat \rho}, \hat \rho)=\tr\Big(T^{-1} \suml_{t=1}^T \hat u_t\hat u_t^\top\Big)= T^{-1} \suml_{t=1}^T \hat u_t^\top \hat u_t = T^{-1} \suml_{i=1}^2 \suml_{i\hat \rho} (x_t^\top - z_t^\top \hat \Pi_{i\hat \rho})(x_t - \hat \Pi_{i\hat \rho}^\top z_t)$, and $d_t = \hat u_t - u_t$. Then, by definition, $Q(\hat \Pi_{1\hat \rho},\hat \Pi_{2\hat \rho},\hat \rho) = T^{-1} \suml_{t=1}^T (u_t + d_t)^\top (u_t + d_t) \leq  Q(\Pi_{1}^0,\Pi_{2}^0,\rho^0) =T^{-1} \sum_{t=1}^T u_t^\top u_t$ with probability one. This implies that $2 T^{-1} \suml_{t=1}^T u_t^\top d_t + T^{-1} \suml_{t=1}^T d_t^\top d_t \leq 0$ with probability one. We now prove consistency in two steps. In \textit{part (i1)}, we show that $T^{-1} \suml_{t=1}^T u_t^\top d_t=\op(1)$, implying $T^{-1} \suml_{t=1}^T d_t^\top d_t=o_p(1)$, since $T^{-1} \suml_{t=1}^T d_t^\top d_t\geq0$. In \textit{part (i2)}, we show that if $\hat \rho \stackrel{p}\nrightarrow \rho^0$, then, with strictly positive probability, $T^{-1} \suml_{t=1}^T d_t^\top d_t>K$ for some $K>0$, contradicting  $T^{-1} \suml_{t=1}^T d_t^\top d_t =o_p(1)$, and therefore yielding $\hat \rho \inp \rho^0$.\\
\indent \textit{Part (i1).} Define $\sum_{\Delta}(\cdot)=\suml_{t=1}^T(\cdot)\mathbf{1}[\rho<q_t\leq \rho^0]$. Then, for any $\rho\leq\rho^0$ instead of just $\hat \rho$\footnote{Note that $\sum_\Delta()=0$ in case of $\rho=\rho^0$.}, and using $\vect(ABC)= (C^\top \otimes A) \vect(B)$, we have:
\begin{align*}
T^{-1} \suml_{t=1}^T u_t^\top d_t & =T^{-1}\suml_{1\rho }u_t^\top (\Pi_1^0 - \hat \Pi_{1\rho})^\top z_t + T^{-1}\suml_{\Delta }u_t^\top (\Pi_1^0 - \hat \Pi_{2\rho})^\top z_t+ T^{-1}\suml_{2\rho^0 }u_t^\top (\Pi_2^0 - \hat \Pi_{2\rho})^\top z_t\\
& =  (T^{-1} \suml_{1\rho }z_t\otimes u_t)^\top\vect(\Pi_1^0 - \hat \Pi_{1\rho}) +(T^{-1} \suml_{\Delta}z_t \otimes u_t)^\top\vect(\Pi_1^0 - \hat \Pi_{2\rho})\\
&\quad+(T^{-1} \suml_{2\rho^0 }z_t \otimes u_t)^\top \vect(\Pi_2^0 - \hat \Pi_{2\rho}). 
\end{align*}
Since $\rho\leq\rho^0$, it follows that $\hat\Pi_{1\rho}$ is computed based on observations corresponding to subsamples $\indrhoo$ only and, therefore, by standard arguments $\hat\Pi_{1\rho}-\Pi_1^0=o_p(1)$. On the other hand, $\hat\Pi_{2\rho}$ is constructed using observations from both subsamples satisfying $\indrhoo$, respectively $\indinvrho.$ Straightforward calculations give $\hat\Pi_2=(T^{-1}\sum_{2\rho}z_tz_t^\top)^{-1}[(T^{-1}\sum_\Delta z_tz_t^\top)\Pi_1^0+(T^{-1}\sum_{2\rho^0} z_tz_t^\top)\Pi_2^0+(T^{-1}\sum_{2\rho} z_tu_t^\top)]=O_p(1)[O_p(1)+O_p(1)+o_p(1)]=O_p(1)$, by Lemma \ref{lem:3} $(i)$ and $(ii)$. Hence, $\Pi_1^0 - \hat \Pi_{2\rho}=\Op(1)$, and $\Pi_2^0 -\hat \Pi_{2\rho} = \Op(1)$. Also by Lemma \ref{lem:3} $(ii)$, $T^{-1} \suml_{1\rho }z_t \otimes u_t,\,T^{-1} \suml_{2\rho^0 }z_t \otimes u_t,\,T^{-1} \suml_{\Delta}z_t \otimes u_t=\op(1)$. Therefore, $T^{-1} \suml_{t=1}^T u_t^\top d_t =\op(1)$ (uniformly in $\rho$). Because these results hold uniformly over $\rho$, we have $\op(1) + T^{-1} \suml_{t=1}^T d_t^\top d_t \leq 0$ uniformly over $\rho$, and therefore also at $\hat \rho$, so $T^{-1} \suml_{t=1}^T d_t^\top d_t=o_p(1)$.\\
\indent \textit{Part (i2).} 
By the continuity assumption \ref{assn:1} $(e)$, there exists an $\epsilon>0$ such that with positive probability, $q_t \in[\rho^0-\epsilon,\rho^0+\epsilon]$. If $\hat \rho \stackrel{p}\nrightarrow \rho^0$, then $\hat \rho < \rho^0-\epsilon$ because $\hat \rho \leq \rho^0$. Consequently, the residuals evaluated over the  sub-sample $q_t \in[\rho^0-\epsilon,\rho^0+\epsilon]$ will also be evaluated at $\hat \Pi_{2\hat\rho}=\hat \Pi_2$ since $\hat\Pi_{2}$ is the multivariate LS estimator in the sample $q_t>\hat\rho$. However, the true parameter values are $\Pi_1^0$  for $q_t \in [\rho^0-\epsilon,\rho^0]$, and $\Pi_2^0$ for $q_t \in(\rho^0,\rho^0+\epsilon]$. Let $\suml_{A} = \sum \mathbf{1}[q_t \in[\rho^0-\epsilon,\rho^0]]$, $\suml_B = \sum \mathbf{1}[q_t \in (\rho^0,\rho^0+\epsilon]]$, and $\suml_{AB} = \suml \mathbf{1}[q_t \in [\rho^0-\epsilon,\rho^0+\epsilon]]$. Denote by $\pi_{1,i}^0$, $\pi_{2,i}^0$ and $\hat\pi_{2,i}$ the $i$-th columns of $\Pi_1^0$, $\Pi_2^0$, respectively $\hat\Pi_2$. Further, let $\eta_1$ and $\eta_2$ be the minimum eigenvalues of $M_{1\rho^0}-M_{1\rho^0-\epsilon}$, respectively $M_{1\rho^0+\epsilon}-M_{1\rho^0}$. Recall that $\tr(ABC)=\vect(A^\top)^\top(I\otimes B)\vect(C)$ for conformable matrices $A$, $B$ and $C$. Then:
\begin{align*}
T^{-1} \sum_{AB} d_t^\top d_t & =T^{-1}\sum_{A}\tr\Big([\Pi_1^{0}-\hat\Pi_2]^\top z_tz_t^\top[\Pi_1^{0}-\hat\Pi_2]\Big)+T^{-1}\sum_{B}\tr\Big([\Pi_2^{0}-\hat\Pi_2]^\top z_tz_t^\top[\Pi_2^{0}-\hat\Pi_2]\Big)\\
& =\vect(\Pi_1^0-\hat\Pi_2)^\top\Big(T^{-1}\sum_AI_{p_1}\otimes z_tz_t^\top\Big)\vect(\Pi_1^0-\hat\Pi_2)\\
& \quad+\vect(\Pi_2^0-\hat\Pi_2)^\top\Big(T^{-1}\sum_BI_{p_1}\otimes z_tz_t^\top\Big)\vect(\Pi_2^0-\hat\Pi_2)
\end{align*}
\begin{align*}
& =\sum_{i=1}^{p_1}(\pi_{1,i}^0-\hat\pi_{2,i})^\top(M_{1\rho^0}-M_{1\rho^0-\epsilon})(\pi_{1,i}^0-\hat\pi_{2,i})\\
& \quad+\sum_{i=1}^{p_1}(\pi_{2,i}^0-\hat\pi_{2,i})^\top(M_{1\rho^0+\epsilon}-M_{1\rho^0})(\pi_{2,i}^0-\hat\pi_{2,i})+o_p(1)\\
& \geq\eta_1\sum_{i=1}^{p_1}\|\pi_{1,i}^0-\hat\pi_{2,i}\|^2+\eta_2\sum_{i=1}^{p_1}\|\pi_{2,i}^0-\hat\pi_{2,i}\|^2+o_p(1)\\
& \geq\min\{\eta_1,\eta_2\}\sum_{i=1}^{p_1}\frac{\|\pi_{1,i}^0-\pi_{2,i}^0\|^2}{2}+o_p(1)
\end{align*}
where the last equality used continuity of $M_{1\rho}$ at $\rho^0$ (Assumption \ref{assn:1} $(d)$) and the last inequality used the fact that for any vectors $a,b,c$ of the same length,  $(a-b)^\top (a-b) +(c-b)^\top(c-b) \geq (a-c)^\top(a-c)/2$. Next, $\min\{\eta_1,\eta_2\}>K>0$ by Assumption \ref{assn:1} $(d)$ and there exists at least one $i=1,...,p_1$ such that $\|\Pi_{1,i}^0-\Pi_{2,i}^0\|>0$ by Assumption \ref{assn:1} $(g)$. Hence, $\plim T^{-1} \sum_{t=1}^t d_t^\top d_t >K>0$ with strictly positive probability, reaching a contradiction. \\
\indent \textit{Part (ii).} Let $\zeta_0=T^{-1}\sum_{t=1}^T \mathbf 1[q_t \leq \rho^0]$, $\zeta=T^{-1}\sum_{t=1}^T \mathbf 1[q_t\leq \rho]$ and $\hat \zeta = T^{-1}\sum_{t=1}^T \mathbf 1[q_t\leq\hat \rho]$. Then $\rho_0-
\rho< \epsilon$, for some small $\epsilon$, can equivalently be written as $\zeta_0-\zeta<\omega$, for some small $\omega$ (where there is a one-to-one correspondence between $\omega$ and $\epsilon$), or in other words, the difference between the \textit{fraction} of observations below the $\rho^0$ and below the $\rho$ quantiles is smaller than $\omega$, as the distribution of $q_t$ is continuous. Define the set $V_\omega=\{\zeta:\zeta_0-\zeta<\omega\}$ for some $\omega>0$. Since $\hat\rho\inp\rho^0$ by Part $(i)$, it holds that $\hat\zeta\in V_\omega$, for large enough $T$ and small enough $\omega$. Thus, we only need to consider the case where $\zeta_0-\zeta<\omega$, for some small $\omega>0$. For $C>0$, define the set $V_{\omega}(C)=\{\zeta:\,\zeta^0-\zeta\leq\omega,\,T(\zeta^0-\zeta)>C\}$, i.e. the set of all values such that the \textit{number} of observations between the two quantiles is larger $C$. By construction $V_{\omega}(C)\subset V_{\omega}$. Thus, if we can show that $\hat \zeta\notin V_{\omega}(C)$, it follows that the number of observations $T\zeta^0 -T \zeta \leq C$, which means in turn that   $T(\hat\rho-\rho^0)=O_p(1)$. Hence, the quantiles $\hat \rho$ and $\rho$ cannot be more than a fixed number apart with large probability. To this extent, recall the definition $Q(\hat\Pi_{1\rho},\hat\Pi_{2\rho},\rho)$ of the multivariate SSR in Part $(i)$ and that $Q(\hat\Pi_{1\hat\rho},\hat\Pi_{2\hat\rho},\hat\rho)\leq Q(\hat\Pi_{1\rho^0},\hat\Pi_{2\rho^0},\rho^0)$ with probability one. To verify the claim, we show that for each $K>0$ there exist $C>0$ and $\epsilon>0$ such that $P(\min_{\zeta\in V_{\omega}(C)}\{Q(\hat\Pi_{1\rho},\hat\Pi_{2\rho},\rho)- Q(\hat\Pi_{1\rho^0},\hat\Pi_{2\rho^0},\rho^0)\}\leq0)<K$ for large $T$. That is, we are going to show that $\hat\zeta\notin V_{\omega}(C)$. Thus, it must follow that $T(\hat\zeta-\zeta^0)\leq C$ with large probability.\\
\indent To this extent, define $Q(\hat\Pi_{1\rho},\hat\Pi_\Delta,\hat\Pi_{2\rho^0},\rho,\rho^0)$ where $\hat\Pi_\Delta$ is constructed using observations $\rho<q_t\leq\rho^0$; i.e., next to $\rho$ we  introduce a second threshold at $\rho^0$ such that $\hat\Pi_\Delta$ is evaluated over samples associated with $\Pi_1^0$ only. Note that $Q(\hat\Pi_{1\hat\rho},\hat\Pi_{2\hat\rho},\hat\rho)-Q(\hat\Pi_{1\rho^0},\hat\Pi_{2\rho^0},\rho^0)=(Q(\hat\Pi_{1\hat\rho},\hat\Pi_{2\hat\rho},\hat\rho)-Q(\hat\Pi_{1\rho},\hat\Pi_\Delta,\hat\Pi_{2\rho^0},\rho,\rho^0))-(Q(\hat\Pi_{1\rho^0},\hat\Pi_{2\rho^0},\rho^0)-Q(\hat\Pi_{1\rho},\hat\Pi_\Delta,\hat\Pi_{2\rho^0},\rho,\rho^0))$. Define $\hat Q_\Delta=\frac{1}{T(\zeta^0-\zeta)}\sum_\Delta z_tz_t^\top$ and $\hat M_\Delta=(\zeta^0-\zeta)\hat Q_\Delta$. Then, by \cite{baiperron1998}, pp. 70, 
\begin{align*}
\frac{Q(\hat\Pi_{1\hat\rho},\hat\Pi_{2\hat\rho},\hat\rho)-Q(\hat\Pi_{1\rho},\hat\Pi_\Delta,\hat\Pi_{2\rho^0},\rho,\rho^0)}{T(\zeta^0-\zeta)} & = \tr[(\hat \Pi_{2\rho^0}-\hat \Pi_\Delta)^\top ( \hat Q_\Delta - \hat Q_\Delta \hat M_{2\rho}^{-1} \hat M_\Delta)(\hat \Pi_{2\rho^0}-\hat \Pi_\Delta)] \\
\frac{Q(\hat\Pi_{1\rho^0},\hat\Pi_{2\rho^0},\rho^0)-Q(\hat\Pi_{1\rho},\hat\Pi_\Delta,\hat\Pi_{2\rho^0},\rho,\rho^0)}{T(\zeta^0-\zeta)} & = \tr[(\hat \Pi_{1\rho}-\hat \Pi_\Delta)^\top ( \hat Q_\Delta - \hat Q_\Delta \hat M_{1\rho^0}^{-1} \hat M_\Delta)(\hat \Pi_{1\rho}-\hat \Pi_\Delta)].
\end{align*}
By Lemma \ref{lem:3} $(i)$ and $(ii)$, $\hat M_{1\rho^0}=M_{1\rho^0}+o_p(1)$ and $\hat M_{2\rho}=M_{2\rho}+o_p(1)$ so that their inverses are $O_p(1)$. Thus, also on $V_{\omega}(C)$. Similarly, $\hat\Pi_{1\rho}=\Pi_{1}^0+o_p(1)$ (since it only uses observations satisfying $q_t\leq\rho<\rho^0$) and $\hat\Pi_{2\rho^0}=\Pi_2^0+o_p(1)$. Straightforward calculations show that $\hat\Pi_\Delta=\Pi_1^0+(T^{-1}\sum_\Delta z_tz_t^\top)^{-1}(T^{-1}\sum_\Delta z_tu_t^\top)$ if $\rho<\rho^0$. By Lemma 3, we have $\hat M_\Delta=T^{-1}\sum_\Delta z_tz_t^\top=T^{-1}\sum_{1\rho^0}z_tz_t^\top-T^{-1}\sum_{1\rho}z_tz_t^\top=M_{1\rho^0}-M_{1\rho}+o_p(1)$ and, similarly $T^{-1}\sum_\Delta z_tu_t^\top=o_p(1)$. For the former term, we have that $M_{1\rho^0}-M_{1\rho}$ has eigenvalues bounded away from 0 on $V_{\omega}(C)$ by Assumption \ref{assn:1} $(d)$ and its inverse is therefore bounded (in the sense that its largest eigenvalue is smaller/equal than some finite constant). Hence, it follows that $\hat\Pi_\Delta=\Pi_1^0+o_p(1)$ on $V_{\omega}(C)$. Moreover, on $V_{\omega}(C)$, we have
\begin{align*}
\|\hat M_\Delta\| &=\| E[z_t z_t^\top (\indrhoo - \indrho)]+\op(1)\| \leq (E\|z_t\|^4)^{1/2} \left(\int_{\rho}^{\rho^0} f(x)dx\right)^{1/2}+ \op(1) \\
&  \leq M\left(\int_{\rho^0-\epsilon}^{\rho^0} f(x)dx \right)^{1/2} +\op(1) = M \sqrt{f(b) \epsilon} +\op(1) =\sqrt{\epsilon}\Op(1),  
\end{align*}
where the second-to-last equality holds by continuity and boundedness of $f(x)$ (Assumption \ref{assn:1} $(e)$) for some $b \in [\rho^0-\epsilon,\rho^0]$. Moreover, $\hat Q_\Delta=\frac{1}{T(\zeta^0-\zeta)}\sum_\Delta z_tz_t^\top=O_p(1)$ on $V_{\omega}(C)$ for large enough $C$ by Lemma \ref{lem:3} since $T(\zeta^0-\zeta)>C>0$ such that $[T(\zeta^0-\zeta)]^{-1}$ is bounded.
Thus, on $V_\omega(C)$ it follows that
\begin{align*}
\frac{Q(\hat\Pi_{1\hat\rho},\hat\Pi_{2\hat\rho},\hat\rho)-Q(\hat\Pi_{1\rho},\hat\Pi_\Delta,\hat\Pi_{2\rho^0},\rho,\rho^0)}{T(\zeta^0-\zeta)} & =\tr[(\Pi_2^0-\Pi_1^0)^\top\hat Q_\Delta(\Pi_2^0-\Pi_1^0)^\top]+\sqrt{\epsilon}O_p(1)+o_p(1)\\
\frac{Q(\hat\Pi_{1\rho^0},\hat\Pi_{2\rho^0},\rho^0)-Q(\hat\Pi_{1\rho},\hat\Pi_\Delta,\hat\Pi_{2\rho^0},\rho,\rho^0)}{T(\zeta^0-\zeta)} & = o_p(1)
\end{align*}
and consequently that
$$
\frac{Q(\hat\Pi_{1\hat\rho},\hat\Pi_{2\hat\rho},\hat\rho)-Q(\hat\Pi_{1\rho^0},\hat\Pi_{2\rho^0},\rho^0)}{T(\zeta^0-\zeta)}=\tr[(\Pi_2^0-\Pi_1^0)^\top\hat Q_\Delta(\Pi_2^0-\Pi_1^0)^\top]+\sqrt{\epsilon}O_p(1)+o_p(1).
$$
Finally, $\hat Q_\Delta=\frac{M_{1\rho^0}-M_{1\rho}}{\zeta^0-\zeta}+o_p(1)$ on $V_\omega(C)$ for large enough $C$. Since $M_{1\rho^0}-M_{1\rho}$ has smallest eigenvalue bounded away from zero (Assumption \ref{assn:1} $(d)$) and since $\rho^0-\rho>0$ it follows that $\tr[(\Pi_2^0-\Pi_1^0)^\top\frac{M_{1\rho^0}-M_{1\rho}}{\zeta^0-\zeta}(\Pi_2^0-\Pi_1^0)^\top]>K>0$. Hence, $\frac{Q(\hat\Pi_{1\hat\rho},\hat\Pi_{2\hat\rho},\hat\rho)-Q(\hat\Pi_{1\rho^0},\hat\Pi_{2\rho^0},\rho^0)}{T(\zeta^0-\zeta)}>K>0$. Hence, $\hat\zeta\notin V_\omega(C)$ and therefore, $T(\zeta^0-\hat \zeta) \leq C$ with large probability, hence $T(\rho^0-\rho)=O_p(1)$.\\
\indent \textit{Part (iii).} Since any partial sum in the expression of $\hat \Pi_i$ differs by the partial sum in the expression of $\hat \Pi_{i\rho^0}$ by $|T(\rho^0-\hat \rho)|<D$ terms, for some $D>0$, which are uniformly bounded by Assumption \ref{assn:1}(c), it follows that $T^{1/2}\vect(\hat\Pi_i-\Pi_i^0) = T^{1/2}\vect(\hat\Pi_{i\rho^0}-\Pi_i^0)+\op(1)$, for $i=1,2$. The rest of the proof follows standard arguments.
\end{proof}
\noindent \textbf{\textit{By Theorem \ref{theo:0} and its proof, wlog, we treat in what follows $\hat \rho$ as if it was equal to $\rho^0$.}}
\subsubsection{Asymptotic Distribution of 2SLS Test Statistics with a TFS}\label{sec:ProofsTFSAsyms}
\begin{definition}\label{def:QuantsTFS}
Let $A_i^0=[\Pi_i^0,S^\top]^\top$ be the augmented matrices of the FS slope parameters, where $S=[I_{p_2},\mathbf{0}_{p_2\times q_1}]$,  $q_1=q-p_2$, and $\hat A_i = [\hat \Pi_i,S^\top]^\top$. Hence, $z_{1t} = Sz_t$ and $w_t = A_1^0 z_t \indrhoo + A_2^0 z_t\indinvrhoo+ \bar u_t$. Let $ A_t^0 = A_1^0 \indrhoo+ A_2^0 \indinvrhoo$, $ \Pi_t = \Pi_1^0 \indrhoo+ \Pi_2^0 \indinvrhoo$.\\
\indent Let $\wedge$ and $\vee$ define the minimum and maximum operators. Let $C_{1\gamma}= A_1^0 M_{1,\mi} A_1^{0\top} +  A_2^0 (M_{1\gamma}- M_{1,\mi}) A_2^{0\top}$ and $C_{2\gamma} = A_1^{0}  (M_{1,\ma}- M_{1 \gamma}) A_1^{0\top}  + A_2^{0}  M_{2, \ma } A_2^{0\top}$. Also, $C_\gamma = [C_{1\gamma}^{-1}, -C_{2\gamma}^{-1}]$, $C = C_{1\gamma}+C_{2\gamma}$,  $R_{i\gamma} = M_{i\gamma} M_{i\rho^0}^{-1}$ for $i=1,2$,   $D=[1,0_{1\times p_1}]\otimes I_q$, and $Q_\gamma =C_{1\gamma} C^{-1} C_{2\gamma}$. Let $F_{i\gamma} = [\check \theta^{0\top} \otimes R_{i\gamma}]$ and $D_{i\gamma}= [\tilde \theta^{0\top} \otimes I_q]-F_{i\gamma}$.
 Also define the Gaussian processes:
\begin{align*}
\mathcal B_1(\gamma) = \begin{cases} A_1^0 \big( D_{1\gamma}\mathcal G_1(\gamma) -F_{1\gamma} (\mathcal G_1(\rho^0)-\mathcal G_1(\gamma)\big),&\gamma \leq \rho^0 \\
\mathcal B  - A_2^0\big( D_{2\gamma}\mathcal G_2(\gamma) -F_{2\gamma} (\mathcal G_1(\gamma)-\mathcal G_1(\rho^0)) &\gamma >\rho^0,
\end{cases}
\end{align*}
where $\mathcal B = \sum_{i=1}^2 A_i^0D\mathcal G_i(\rho^0)$. Let $V_{\mathcal B} =  \sum_{i=1}^2 A_i^0DH_{i,\rho^0}D^\top A_i^{0\top}$, 
\begin{align*}
V_{\mathcal B,1\gamma}& =\begin{cases} A_1^0 [ D_{1\gamma} H_{1\gamma}D_{1\gamma}+ F_{1\gamma} ( H_{1,\rho^0}-H_{1\gamma})F_{1\gamma}]A_1^{0\top}, & \gamma \leq \rho^0, \\ 
 A_1^0\, D \, H_{1\rho^0} D A_1^0 + A_2^0 [ (D+ F_{2\gamma}) (H_{1\gamma}-H_{1\rho^0})(D+ F_{2\gamma})^\top \\
\quad+ (D-D_{2\gamma})H_{2\gamma} (D-D_{2\gamma})^\top] A_2^{0\top},  &   \gamma > \rho^0 
\end{cases},
\\
V_{\mathcal B,12\gamma} & = \begin{cases} 
A_1^0  [ D_{1\gamma} H_{1\gamma}-F_{1\gamma} (H_{1,\rho^0}-H_{1\gamma})]\, D^\top A_1^{0\top} -V_{\mathcal B,1\gamma},& \gamma \leq \rho^0\\
A_2^0[ (D-D_{2\gamma}) H_{2\gamma} D_{2\gamma}^\top - (F_{2\gamma}+D) (H_{1\gamma}-H_{1\rho^0})F_{2\gamma}^\top] A_2^{0\top}, & \gamma > \rho^0
\end{cases},
\end{align*}
and $V_{\mathcal B,2\gamma} = V_{\mathcal B}-V_{\mathcal B, 1\gamma}-V_{\mathcal B, 12,\gamma}- V_{\mathcal B, 12,\gamma}^\top$ denote the covariances of $\mathcal B$, $\mathcal B_{1\gamma}$, the covariance between $\mathcal B_{1\gamma}$ and $\mathcal B_{2\gamma}$, as well as the covariance of $\mathcal B_{2\gamma}$. Then $V_{\mathcal B,\gamma}  = \begin{bmatrix} V_{\mathcal B,1\gamma} & V_{\mathcal B,12, \gamma} \\ V_{\mathcal B,12, \gamma} ^\top & V_{\mathcal B, 2\gamma}\end{bmatrix},$ and $V_{\gamma} = C_{\gamma}  V_{\mathcal B,\gamma} C_\gamma^\top$. \\
Let $\bar F_{i\gamma} = [\check \theta^{\top} \otimes \hat R_{i\gamma}]$ and $\bar D_{i\gamma}= [\tilde \theta^{\top} \otimes I_q]-\bar F_{i\gamma}$, where $\tilde \theta^\top = [1,\hat \theta_x]$ and $\check \theta = [0,\hat \theta^\top]$.  where $\tilde \theta = [1,\hat \theta_x]$ and $\check \theta = [0,\hat \theta_x]$. Then $\hat V_\gamma$ is defined as $V_{\gamma}$, but replacing $C_{i\gamma}$ with $\hat C_{i\gamma}$, $A_i^0$ with $\hat A_i$,  $D_{i\gamma}$ by $\bar D_{i\gamma}$, $F_{\i\gamma}$ by $\bar F_{i\gamma}$, $H_{i\gamma}$ by $\hat H_{i\gamma} = T^{-1} \suml_{i\gamma} \hat v_t \hat v_t^\top \otimes z_t z_t^\top$.
\end{definition}

\noindent \textbf{\textit{With this new notation, we now reprove Lemmas \ref{lem:4} and \ref{lem:5}, and Theorems  \ref{thm:LFSAsyms} and \ref{thm:1}, for $x_t$ generated by the TFS \eqref{equ:TMTS.3} instead of the LFS \eqref{equ:TMTS.2}.}}
\begin{proof}[\textbf{Proof of Lemma \ref{lem:4}.}]
\noindent \textit{Part (i).  } $T^{-1} \sum_{i\gamma} z_t z_t^\top \inp M_{i\gamma}$ still holds, as the result is not specific to a LFS or TFS. So,
\begin{align*}
T^{-1} \suml_{1\gamma} \hat w_t \hat w_t^\top &= \hat A_1 T^{-1}\suml_{1,\mi}z_t z_t^\top \hat A_1^\top + \hat A_2 \left (T^{-1} \suml_{1\gamma} z_t z_t^\top - T^{-1} \suml_{1,\mi} z_t z_t^\top \right)\hat A_2^\top \\
& \inp A_1^0 M_{1,\mi} A_1^{0\top} +  A_2^0 (M_{1\gamma}- M_{1,\mi}) A_2^{0\top} = C_{1\gamma}  \\
T^{-1} \suml_{2\gamma} \hat w_t \hat w_t^\top &= \hat A_1 \left(T^{-1} \suml_{1,\ma}  z_t z_t^\top - T^{-1} \suml_{1,\gamma}z_t z_t^\top \right) \hat A_1^\top + \hat A_2 T^{-1} \suml_{2\ma} z_t z_t^\top  \hat A_2^\top \\
&\inp A_1^{0}  (M_{1,\ma}- M_{1 \gamma}) A_1^{0\top}  + A_2^{0}  M_{2\ma} A_2^{0\top}= C_{2\gamma}.
\end{align*}
\indent \textit{Part (ii).  } The result $T^{-1/2} \sum_{i\gamma} v_t \otimes z_t \Rightarrow \mathcal G_i(\gamma)$ still holds. But now, $$\tilde \epsilon_t = \epsilon_t +(\hat x_t-x_t)^\top \theta_x^0 =   \epsilon_t + u_t^\top \theta_x^0 - \indrhoo[z_t^\top (\hat \Pi_1 -\Pi_1^0)\theta_x^0]- \indinvrhoo[z_t^\top (\hat \Pi_2-\Pi_2^0) \theta_x^0].$$ Therefore, for $\gamma \leq \rho^0$,
\begin{align*}
&T^{-1/2} \suml_{1\gamma} \hat w_t \tilde \epsilon_t = A_1^0 \, \big( T^{-1/2} \suml_{1\gamma} z_t (\epsilon_t+ u_t^\top \theta_x^0) - M_{1\gamma} M_{1\rho^0}^{-1} T^{-1/2} \suml_{1\rho^0} z_t u_t^\top \theta_x^0)\big) \\
& \ind A_1^0 \big( [\tilde \theta^0 \otimes I_q] \,\mathcal G_1(\gamma) -[\check \theta^0 \otimes R_{1\gamma}] \mathcal G_1(\rho^0)\big) = A_1^0 \big( D_{1\gamma} \,\mathcal G_1(\gamma) -F_{1\gamma} (\mathcal G_1(\rho^0)-\mathcal G_1(\gamma)\big)\\
&=\mathcal B_1(\gamma). 
\end{align*}
For $\gamma >\rho^0$,
\begin{align*}
T^{-1/2} \suml_{1\gamma} \hat w_t \tilde \epsilon_t & =T^{-1/2}  \suml_{t=1}^T\hat w_t \tilde \epsilon_t -T^{-1/2} \suml_{2\gamma} \hat w_t \tilde \epsilon_t\\
&\ind A_1^0 \, D \, \mathcal G_1(\rho^0) + A_2^0\,  D\mathcal G_2(\rho^0)   - A_{2}^0 ( [\tilde \theta^0 \otimes I_q] \mathcal G_2(\gamma) -F_{2\gamma} \mathcal G_2(\rho^0)) \\
& =\mathcal B  - A_2^0\big( D_{2\gamma} \,\mathcal G_2(\gamma) -F_{2\gamma} (\mathcal G_1(\gamma)-\mathcal G_1(\rho^0)) = \mathcal B_1(\gamma). 
\end{align*}
Because  $T^{-1/2}  \suml_{t=1}^T\hat w_t \tilde \epsilon_t \ind A_1^0 \, D \, \mathcal G_1(\rho^0) + A_2^0\,  D\mathcal G_2(\rho^0) $,
$
T^{-1/2} \suml_{2\gamma} \hat w_t \tilde \epsilon_t \Rightarrow A_1^0 \, D \, \mathcal G_{1}(\rho^0) + A_2^0 \, D \, \mathcal G_2(\rho^0) - \mathcal B_1(\gamma)= \mathcal B - \mathcal B_1(\gamma)=\mathcal B_2(\gamma)$,
and   $\vect(T^{-1/2} \suml_{1\gamma} \hat w_t \tilde \epsilon_t, T^{-1/2} \suml_{2\gamma} \hat w_t \tilde \epsilon_t)\Rightarrow \mathcal B(\gamma)$.
\end{proof}
\begin{proof}[\textbf{Proof of Theorem \ref{thm:LFSAsyms}.}] \textit{Part (i)}. Because $T^{-1/2} (\hat \theta_{1\gamma}- \hat \theta_{2\gamma}) = \hat C_{1\gamma}^{-1} T^{-1/2} \, \sum_{1\gamma} \hat w_t \tilde \epsilon_t - \hat C_{2\gamma}^{-1}\,  T^{-1/2} \sum_{2\gamma} \hat w_t \tilde \epsilon_t$, the desired result follows directly from Lemma \ref{lem:3}.
 
\indent \textit{Part (ii)}. Follows the same steps as for the LFS proof until equation \eqref{equ:PT4.1}. Then note that because $\hat \Pi_i- \Pi_i^0 =\op(1)$, 
\begin{align*}
T^{-1} SSR_1(\gamma) &= T^{-1} \suml_{i=1}^2 \left(\suml_{i\rho^0} (\epsilon_t + u_t^\top \theta_x^0)^2 - 2 \suml_{i\rho^0} (\epsilon_t + u_t^\top \theta_x^0)z_t^\top (\hat \Pi_i- \Pi_i^0)\theta_x^0 \right.\\
&\qquad\left.+ \suml_{i\rho^0} \theta_x^{0\top}(\hat \Pi_i- \Pi_i^0)^\top  T^{-1} \suml_{i\rho^0} z_t z_t^\top  (\hat \Pi_i- \Pi_i^0) \theta_x^0\right) \\
& = \sigma^2 +\op(1),
\end{align*}    
 following the same arguments as in the LFS proof.
 
\indent \textit{ Part (iii).} It can be shown by similar arguments to the LFS, but now separately for cases $\gamma \leq \rho^0$ and $\gamma \geq \rho^0$, and taking to account the different parameter estimates in different regimes, that $\hat V_\gamma \inp V_\gamma$. Because of part (i) of this theorem, the desired result follows. 
\end{proof}
        
\begin{proof}[\textbf{Proof of Lemma \ref{lem:4} and Theorem \ref{thm:1}.}] As evident from the proof of Theorem \ref{thm:LFSAsyms} for a TFS, besides replacing $\hat \Pi$ with $\hat \Pi_i$, and $\hat \Pi^b$ with $\hat \Pi_i^b$, and re-deriving the terms involving these, there are no essential differences between the proofs for a LFS and a TFS, and for brevity we omit these proofs.
\end{proof}
\end{document}